%% file: ex_article.tex
\documentclass[final,hidelinks,onefignum,onetabnum]{siamart220329}

\usepackage{enumerate}
\usepackage{enumitem}
\usepackage{csquotes}

\usepackage[caption=false]{subfig}

\definecolor{ForestGreen}{RGB}{34,139,34}

\newcommand{\seb}[1]{%
{\leavevmode\color{blue}#1}%
}

\newcommand{\ben}[1]{%
{\leavevmode\color{red}#1}%
}

\newcommand{\brian}[1]{%
{\leavevmode\color{ForestGreen}#1}%
}

\newtheorem{rem}{Remark}

\newtheorem{lem}{Lemma}
\newtheorem{prop}{Proposition}
\newtheorem{cor}{Corollary}
\newtheorem{thm}{Theorem}
\newtheorem{assm}{Assumption}

\input{ex_shared}

\ifpdf
\hypersetup{
  pdftitle={Towards Understanding the Endemic Behavior of a Competitive  Tri-Virus SIS Networked Model},
  pdfauthor={S.~Gracy, M.~Ye, BDO~Anderson and C.~A.~Uribe}
}
\fi

\begin{document}

\maketitle

\begin{abstract}
This paper studies the endemic behavior of a multi-competitive networked \break susceptible-infected-susceptible (SIS) model. Specifically, the paper deals with three competing virus systems (i.e., tri-virus systems). First, we show that a tri-virus system, unlike a bi-virus system, is not a monotone dynamical system. Using the Parametric Transversality Theorem, we show that, generically, a tri-virus system has a finite number of equilibria and that the Jacobian matrices associated with each equilibrium are nonsingular.
The endemic equilibria of this system can be classified as follows: a) single-virus endemic equilibria (also referred to as the boundary equilibria), where precisely one 
of the three viruses is alive; b) 2-coexistence equilibria, where exactly two of the three viruses are alive; and c) 3-coexistence equilibria, where all three viruses survive in the network. We provide a necessary and sufficient condition that guarantees local exponential convergence to a boundary equilibrium.  Further, we secure conditions for the nonexistence of 3-coexistence equilibria (resp. for various forms of 2-coexistence equilibria). We also identify sufficient conditions for the existence of a 2-coexistence (resp. 3-coexistence) equilibrium. We identify conditions on the model parameters that give rise to a continuum of coexistence equilibria. More specifically, we establish i) a scenario that  admits the existence and local exponential attractivity of a line of coexistence equilibria; and ii) scenarios that admit the existence of,  and, in the case of one such scenario, global convergence to, a  plane of 3-coexistence equilibria.
\end{abstract}



\section{Introduction}\label{sec:introduction}
Spreading processes are observed in several settings. Prominent examples include the spread of viruses, in particular, pandemics such as COVID-19~\cite{wang2020clinical}, Spanish flu~ \cite{johnson2002updating}; spread of opinions on social media~\cite{wang2019systematic,bradshaw2017troops}; adoption of products in a market \cite{armington1969theory}; and propagation of species in an ecological environment \cite{macarthur1967limiting}. Given these manifestations in several diverse settings, the study of spreading processes has naturally attracted the attention of multiple scientific communities viz. physics \cite{newman2005threshold,karrer2011competing,wang2012dynamics}, biology \cite{laurie2018evidence,wu2020interference}, mathematics \cite{castillo1989epidemiological,carlos2,hethcote2000mathematics}, computer science \cite{prakash2010virus}, to cite a few. 
\par One of the primary objectives behind all of the aforementioned research efforts is to understand the various equilibria of the systems involved and 
the behavior of spread processes in the neighborhood of an equilibrium point.
Several models have been proposed and extensively studied in mathematical epidemiology: notable among them being the susceptible-infected-recovered (SIR) \cite{mei2017dynamics,van2014exact}; susceptible-exposed-infected-recovered (SEIR) \cite{li1995global,arcede2020accounting};  susceptible–asymptomatic–infected–recovered–susceptible (SAIRS) \cite{rothe2020transmission,pare2020modeling};
 susceptible-infected (SI) \cite{matouk2020complex}; and susceptible-infected-susceptible (SIS) models \cite{kermack1932contributions,lajmanovich1976deterministic,van2009virus,khanafer2016stability}.
The aforementioned models provide rigorous guarantees for (some of) the various outcomes of spreading processes; the conditions involved in providing said guarantees are tractable; and finally, they can generate insightful observations due to the low cost of simulations. Thus,  these models are extremely useful in understanding the behavior of spreading processes and, consequently, have gained widespread acceptance in the broader research community. Our focus in the present paper is on networked SIS models, which consider a population of individuals that are divided 
into subpopulations, called  \emph{agents}. A (possibly) directed graph captures how the various agents are interconnected. Given that there is a virus spreading in the whole population,
it can spread both within a subpopulation (i.e.,
between individuals in a subpopulation) and across subpopulations (i.e.,
between individuals belonging to different subpopulations). 
If none of the individuals in a subpopulation are infected, then this subpopulation is said to be in the susceptible state; otherwise, it is said to be in the infected state. Networked SIS models have been extensively studied in the literature; see, for instance, \cite{FallMMNP07,khanafer2016stability,rami2014switch,peng2010epidemic,wang2003epidemic,gracy2020analysis,pare2020modeling}. 

\par Available models account for the first instance of the spread of only one virus. However,
one encounters scenarios where multiple viruses are simultaneously circulating in a population.
 In such a context, the viruses involved could either be cooperative (see \cite{gracy2022modeling} and references therein for details), or they could be competitive - the present paper focuses on the competitive case. The notion of competing viruses can be observed in several settings; 
some examples include spread of multiple strains of a virus \cite{nowak1991evolution,laurie2018evidence,wu2020interference}, spread of competing opinions on social networks
\cite{prakash2012winner,wei2013competing}, and that of products competing in a market \cite{trpevski2010model}. More concretely, supposing that two viruses (say virus~1 and virus~2) are competing, an individual can be either infected with virus~1 or with virus~2 or with neither. In contrast, no individual can be infected with \emph{both} viruses simultaneously. 
The dynamics of competing viruses are far richer than the case of single-virus propagation \cite{newman2005threshold}. Specifically, in multi-competitive virus settings, some possibilities include one  virus dominating the rest and all (or a subset of all) viruses existing in some balance -- a phenomenon referred to as \emph{coexistence}. 
Various mathematical models have been proposed in the literature to deal with multi-competitive viruses. Prime among them is the susceptible-infected-recovered (SIR) model (proposed in \cite{newman2005threshold} and improved upon in \cite{zhang2022networked}) and the susceptible-infected-susceptible (SIS) model (first proposed in \cite{castillo1989epidemiological}). This paper focuses on the continuous-time multi-competitive networked SIS model. 

The study of multi-competitive networked SIS models has been pursued in a plethora of works starting with \cite{castillo1989epidemiological,carlos2}, and, more recently, in 
\cite{sahneh2014competitive,axel2020TAC,santos2015bi,ye2021convergence,liu2019analysis,ben:brian:opinion:lcss,pare2021multi,anderson2022equilibria,doshi2022convergence,gracy2022trivirus}. Note that the vast majority of the literature in this area, namely \cite{sahneh2014competitive,santos2015bi,doshi2022convergence,axel2020TAC,liu2019analysis,ye2021convergence,anderson2022equilibria}, focuses on the special case when two viruses compete with each other to infect a given population. The more general case where $m>2$ viruses are competing has been relatively less well studied; see \cite{pare2021multi} (resp. \cite{pare2020analysis}) for a preliminary analysis of (resp. discrete-time) multi-competitive networked SIS model. While \cite{pare2021multi,pare2020analysis} provide the general modeling framework for multi-competitive networked SIS models, in that the models therein admit the possibility of more than two viruses competing, the analysis of the possible endemic behavior in said papers is rather limited; in particular most of the results in \cite{pare2020analysis} pertain to the existence  of boundary equilibria. Thus, to our knowledge, an overarching analysis accounting for $m$ competing viruses, where $m$ is some arbitrary but finite positive integer, is not available in the existing literature. Therefore, as a first step, in this direction, the present paper deals with the case where $m=3$ (tri-virus system). 
The equilibria of the tri-virus system can be broadly classified into three categories: the disease-free equilibrium (DFE) (all three viruses have been eradicated); the boundary equilibria (two viruses are dead, and one is endemic); 
and $k$-coexistence equilibria (where $k \geq 2$ viruses infect separate fractions of every population node in the network).

We 
 make the following key contributions:
\begin{enumerate}[label=\roman*)]
\item We show that it (i.e., the tri-virus system) is not monotone; see Theorem~\ref{prop:tri-virus-not-monotone}. Consequently, one cannot use the existing tools in the literature on competitive bivirus systems, which are deeply rooted 
in monotone dynamical systems (see \cite{hirsch1988stability,smith1988systems}) to study the limiting behavior of our model.
\item We show, using the Parametric Transversality Theorem of differential topology \cite[page~145]{lee2013introduction}, that it has a finite number of equilibria, in a generic \footnote{A precise definition of the term \enquote{generic} is provided in Section~\ref{sec:trivirus:monotone}.} sense; see Proposition~\ref{prop:finite:equilibria}. Additionally, Proposition~\ref{prop:finite:equilibria} also establishes that the  Jacobian matrices evaluated at each of the equilibrium points are nonsingular.
\item We identify a sufficient condition, and multiple necessary conditions, for local exponential convergence to a boundary equilibrium; see Theorem~\ref{thm:local}.
\item We identify a sufficient condition for the nonexistence of any 3-coexistence equilibria, and nonexistence of certain forms of 2-coexistence equilibria; see Theorem~\ref{claim:virus3:strongest} and Theorem~\ref{claim:virus1weaker}, respectively. 
\item Leveraging that a bivirus system is monotone, we identify a novel and simpler sufficient condition for a 2-coexistence equilibrium; see Proposition~\ref{prop:2-coexistence}.
\item Using the notion of saturated fixed points, we identify a sufficient condition for the existence of a 3-coexistence equilibrium; see Corollary~\ref{cor:hofbauer}.
\item We identify a special scenario (with respect to the class of system parameters) that gives rise to the existence of, and, subject to the fulfillment of certain conditions, local convergence to, a line of coexistence equilibria; see Theorem~\ref{thm:init:condns}.
\item Assuming that the three viruses are identical, we show that, regardless of the initial non-zero infection levels, the dynamics of the tri-virus system converge to a plane of coexisting equilibria; see Theorem~\ref{thm:global:plane}.
\end{enumerate}
In addition, we make several auxiliary contributions. More specifically, we show that for low-dimensional systems, i.e., $n=1$ or $n=2$, certain kinds of coexistence equilibria can, generically, never exist; see Propositions~\ref{prop:brian:1} and~\ref{prop:brian:2}. Using the notion of saturated fixed points, we prove that either a stable boundary equilibrium exists and/or a saturated 2-coexistence equilibrium and/or a 3-coexistence equilibrium; see Theorem~\ref{thm:hofbauer}. We identify a special setting, but one that subsumes the setting in Theorem~\ref{thm:global:plane}, that permits the existence of a plane of coexistence equilibria; see Proposition~\ref{prop:plane:equilibrium}.

A subset of the results of this paper has been accepted for publication in the Proceedings of the  $2023$ American Control Conference; see~\cite{gracy2022trivirus}. The present paper is a significant extension from~\cite{gracy2022trivirus}; Theorem~\ref{claim:virus3:strongest}, Theorem~\ref{claim:virus1weaker}, Corollary~\ref{cor:hofbauer}, Theorem~\ref{thm:hofbauer}, Propositions~\ref{prop:finite:equilibria}, ~\ref{prop:brian:1}, ~\ref{prop:brian:2}, ~\ref{prop:2-coexistence} and ~\ref{prop:plane:equilibrium} were not provided in \cite{gracy2022trivirus}.

\subsection*{Paper Outline}
The paper unfolds in the following manner: We use the rest of this section to collect notations used in the sequel. The formal presentation of the tri-virus model and the questions of interest in the present paper are provided in Section~\ref{sec:prob:formulation}. We investigate the monotonicity of the tri-virus system in Section~\ref{sec:trivirus:monotone}, whereas Section~\ref{sec:finiteness:of:equiibria} deals with the question of the finiteness of the number of equilibria (in a generic sense). Note that the proof for the result in Section~\ref{sec:finiteness:of:equiibria} is relegated to the Appendix since it uses different tools compared to the rest of the results in the paper. Section~\ref{sec:persistence:one:virus} presents a necessary and sufficient condition for local exponential stability of the boundary equilibrium. Sections~\ref{sec:nonexistence:coexistence:equilibria} and~\ref{sec:existence:coexistence} deal with the nonexistence and existence of various kinds of coexistence equilibria. Certain special setting(s), all of which have measure zero, that yield a line (resp. plane) of coexistence equilibria, 
are provided in Section~\ref{sec:line:attractivity} (resp. Section~\ref{sec:global:plane}). Simulations illustrating our theoretical findings are provided in Section~\ref{sec:simulations}. A summary of the paper, along with certain future research directions, is provided in Section~\ref{sec:conclusions}.

\subsection*{Notations}
We denote the set of real numbers by $\mathbb{R}$ and the set of nonnegative real numbers by $\mathbb{R}_+$. For any positive integer $n$, we use $[n]$ to denote the set $\{1,2,...,n\}$. The $i^{\rm{th}}$ entry of a vector $x$ is denoted by $x_i$. The element in the $i^{\rm{th}}$ row and $j^{\rm{th}}$ column of a matrix $M$ is denoted by $M_{ij}$. We use $\textbf{0}$ and $\textbf{1}$ to denote the vectors whose entries all equal $0$ and $1$, respectively, and use $I$ to denote the identity matrix, while the sizes of the vectors and matrices are to be understood from the context. For a vector $x$ we denote the square matrix with $x$ along the diagonal by $\diag(x)$. For any two real vectors $a, b \in \mathbb{R}^n$ we write $a \geq b$ if $a_i \geq b_i$ for all $i \in [n]$, $a>b$ if $a \geq b$ and $a \neq b$, and $a \gg b$ if $a_i > b_i$ for all $i \in [n]$. Likewise, for any two real matrices $A, B \in \mathbb{R}^{n \times m}$, we write $A \geq B$ if $A_{ij} \geq B_{ij}$ for all $i \in [n]$, $j \in [m]$, and $A>B$ if $A \geq B$ and $A \neq B$,  and $A \gg B$ if $A_{ij} > B_{ij}$ for all $i \in [n]$, $j \in [m]$.
For a square matrix $M$, we use $\sigma(M)$ to denote the spectrum of $M$, $\rho(M)$ to denote the spectral radius of $M$, and $s(M)$ to denote the largest real part among the eigenvalues of $M$, i.e., $s(M) = \max\{\rm{Re}(\lambda) : \lambda \in \sigma(M)\}$. 
\par A square matrix $A$ is said to be Hurwitz if $s(A)<0$.
A real square matrix $A$ is called Metzler if all its off-diagonal entries are nonnegative. A real square matrix $A$ is said to be a Z-matrix if all of its off-diagonal entries are nonpositive. 
A Z-matrix is an M-matrix if all its eigenvalues have nonnegative real parts. An M-matrix is singular or nonsingular if it has an eigenvalue at the origin or if all its eigenvalues have strictly positive parts. The matrix $A$ is positive semidefinite if $x^\top A x \geq 0$ for all vectors $x$; we denote this by $A\succeq 0$.

\section{Problem Formulation}\label{sec:prob:formulation}
This section presents a model representing the possible spread of multiple competing viruses over a population network. Thereafter, we state the requisite assumptions and definitions needed for the theoretical developments of this paper. Lastly, we will formally state the problems of interest.

\subsection{Model}
We consider  a network of $n\geq 2$ 
nodes\footnote{As an aside, multivirus SIS models where $n=1$ have been studied in, among others, \cite[Section~2]{pare2020modeling}.}. Each node represents a well-mixed population of individuals (i.e., any two individuals in the population can interact with the same positive probability) with a large and constant size. One of the central assumptions underpinning this model is that of homogeneity within the population node and (possible) heterogeneity outside the population node. That is, all individuals within a population node have the same infection (resp. healing) rates, but individuals in different population nodes need not necessarily have the same healing (resp. infection) rate~\cite{lajmanovich1976deterministic}. We suppose that $m$ viruses compete with each other to infect the nodes. The fact that the viruses are competing implies that at least two (but possibly more) viruses are simultaneously  circulating in the population.

The case where $m=1$ (meaning there is no competition) has been well-studied in the literature; see, for instance, \cite{khanafer2016stability,lajmanovich1976deterministic,fall2007epidemiological,van2008virus}, whereas the case where $m=2$ (i.e., two competing viruses) has been explored relatively well in recent times, see, for example, \cite{sahneh2014competitive,santos2015bi,ye2021convergence,liu2019analysis,pare2021multi}. 
This paper, for reasons mentioned in Section~\ref{sec:introduction}, focuses on the case where $m=3$.
Hence, within each population node, individuals can be partitioned into four mutually exclusive health compartments: susceptible, infected with virus~1, infected with virus~2, and infected with virus~3. Note that  no individual can be \emph{simultaneously} infected by more than one virus \cite{laurie2018evidence,wu2020interference}. A population node is termed \emph{healthy} if all individuals in said node belong to the susceptible compartment; otherwise, we say it is infected. An individual belonging to population node $i$ (where $i\in [n]$), in the susceptible compartment, as a consequence of coming into contact with its infected neighbors, transitions to the ``infected with virus $k$" (for $k \in [m]$) compartment at a rate $\beta_i^k > 0$. An individual in population~$i$ that is infected with virus~$k$ recovers from it 
based on 
said individual's
healing rate with respect to virus~$k$, i.e.,~$\delta_i^k > 0$.

\par We model the spread of $m$-competing viruses using an $m$-layer graph $G$; the vertices of the graph represent the population nodes. The contact  graph for the spread of virus~$k$, for each $k \in [m]$, is denoted by the $k^\textrm{th}$ layer. More specifically, for the  graph $G$, there exists a directed edge from node $j$ to node $i$ in layer $k$ if, assuming an individual in population $j$ is infected with virus~$k$, then said individual can infect at least one (but possibly more) healthy individual in node~$i$. The edge set corresponding to the $k^\textrm{th}$ layer of $G$ is denoted by $E^k$, whereas $ A^k$ denotes the weighted adjacency matrix corresponding to layer~$k$ (where $a_{ij}^k \geq 0$). The elements in $A^k$ are in one-to-one correspondence with the existence (or lack thereof) of edges in layer $k$.  That is,  $(i,j)\in E^k$ if, and only if, $a_{ji}^k\neq 0$. We use $x_i^k(t)$ to represent the fraction of individuals infected with virus~$k$ in population~$i$ at time instant $t$. Given that the viruses are circulating in the population nodes, the infection level $x_i^k(t)$ possibly changes over time. 
Hence,  the evolution of said fraction can, then, be represented by the following scalar differential equation \cite[Equation~4]{pare2021multi}:
\begin{equation} \label{eq:scalar}
   \dot{x}_i^k(t) = - \delta_i^k x_i^k(t) + \big{(} 1 - \textstyle \sum_{l=1}^m x_i^l(t) \big{)} 
  \textstyle \sum_{j=1}^{n} \beta_{ij}^k x_j^k(t),
  \end{equation}
where $\beta_{ij}^k = \beta_i^ka_{ij}^k$.

Define $x^k(t) = [x_1^k(t), \hdots, x_n^k(t)]^\top$, $D^k =\diag{(\delta_i^k)}$, and $B^k=[\beta_{ij}^k]_{n \times n}$. Therefore, the $n$ coupled equations in~\eqref{eq:scalar} for $i \in [n]$
can be written as
\begin{equation} \label{eq:vec}
   \dot{x}^k(t) =  \Big{(} - D^k+ \big{(} I - \textstyle \sum_{l=1}^m \diag(x^l(t)) \big{)} B^k  \Big{)} x^k(t), 
  \end{equation}
Defining $x(t):=[x^1(t),  \dots , x^m(t)]^T$,  and  \break  $R^k(x(t)) := \big{(} - D^k + (I - \textstyle \sum_{l=1}^m \diag(x^l(t)))B^k \big{)}$, the dynamics of the system of all $m$ viruses are given by 
\begin{gather} \label{eq:full}
 \dot{x}(t)
 =
 \begin{bmatrix}
 R^1 \big{(} x(t) \big{)} & 0 & \dots & 0 \\
 0 & R^2 \big{(} x(t) \big{)} & \dots & 0 \\
 \vdots & \vdots & \ddots & \vdots \\
 0 & 0 & \dots & R^m \big{(} x(t) \big{)}
 \end{bmatrix}
 x(t).
\end{gather}

\normalsize
\par Define, for $k \in [3]$, $X^k=\diag(x^k)$. Based on~\eqref{eq:vec}, the dynamics of the tri-virus system can be written as follows:

\begin{align} 
   \dot{x}^1(t) &=  \Big{(} \big{(} I - (X^1+X^2+X^3) \big{)} B^1 - D^1 \Big{)} x^1(t), \label{eq:x1}\\
      \dot{x}^2(t) &=  \Big{(} \big{(} I - (X^1+X^2+X^3) \big{)} B^2 - D^2 \Big{)} x^2(t) \label{eq:x2}\\
         \dot{x}^3(t) &=  \Big{(} \big{(} I - (X^1+X^2+X^3) \big{)} B^3 - D^3 \Big{)} x^3(t). \label{eq:x3}
  \end{align}
 
\subsection{Assumptions} 
 We need the following assumptions to ensure that the aforementioned model is well-defined. Further, note that these assumptions are standard in the literature on (multi-competitive) networked SIS models; see, for instance, \cite{lajmanovich1976deterministic,fall2007epidemiological,khanafer2016stability,liu2019analysis}.
 \begin{assm} \label{assum:base}
Suppose that $\delta_i^k>0,  \beta_{ij}^k \geq 0$  for all $i, j \in [n]$ and $k \in [3]$. 
\end{assm}

\begin{assm}\label{assum:irreducible}
The matrix $B^k$, for $ k \in [3]$ is irreducible. \end{assm} 
 
Under Assumption~\ref{assum:base},  for all $k \in [3]$, $B^k$ is a nonnegative matrix, and $D^k$ is a positive diagonal matrix. 
Moreover, recall that a square nonnegative matrix $M$ has the irreducibility property if, and only if, supposing $M$ is the  (un)weighted adjacency matrix of a graph, the 
said graph is strongly connected.
Then, noting that non-zero elements in $B^k$ represent directed edges in the set $E^k$, we see that $B^k$ is irreducible whenever the $k^{\rm{th}}$ layer of the multi-layer network $G$ is strongly connected. 

\par  As a consequence of 
Assumption~\ref{assum:base}, 
our analysis of the model in~\eqref{eq:x1}-\eqref{eq:x3} is limited to the sets $\mathcal{D}:= \{x(t): x^k(t) \in [0,1]^n,  \forall k \in [3], \sum_{k=1}^{3}x^k \leq \textbf{1}\}$ and $\mathcal{D}^k:= \{x^k(t) \in [0,1]^n\}$. 
Given that $x_i^k(t)$ is interpreted as a fraction of a population, 
the aforementioned sets represent the sensible domain of the 
system. 
That is, 
if $x^k(t)$ takes values outside of $\mathcal{D}^k$, then those values will not correspond to physical reality. 
The following lemma shows that $x(t)$ never leaves the set $\mathcal{D}$.
%
%
\begin{lem}{\cite[Lemma~1]{pare2021multi}}\label{lem:pos}
Let Assumption~\ref{assum:base} hold. Then $\mathcal{D}$ is positively invariant with respect to~\eqref{eq:full}.
\end{lem} 

\begin{lem}\cite[Lemma~7]{axel2020TAC} \label{lem:pos_never_zero}
Let Assumption~\ref{assum:base} hold. Then $\mathcal{D} \setminus \{ \textbf{0} \}$ is positively invariant with respect to system~\eqref{eq:x1}-\eqref{eq:x3}.
\end{lem}

\subsection{Various Equilibria of the Model}\label{ssec:equilibria_types}
Clearly,  $(\textbf{0}, \textbf{0},\textbf{0})$ is an equilibrium of~\eqref{eq:x1}-\eqref{eq:x3}, and is referred to as the DFE. 
A sufficient condition for global exponential stability (GES) of the DFE is as follows:
\begin{prop}\cite[Theorem~1]{axel2020TAC}\label{prop:exp:convergence}
Consider system~\eqref{eq:x1}-\eqref{eq:x3} under Assumption~\ref{assum:base}. If $s(-D^k+B^k) < 0$, for each $k \in [3]$, then the DFE  is exponentially stable, with a domain of attraction containing~$\mathcal{D}$.
\end{prop}

Clearly, the conditions in Proposition~\ref{prop:exp:convergence} also imply asymptotic convergence to the DFE. However, even when the strict inequalities in the 
aforementioned proposition are 
weakened to allow equality,   asymptotic convergence to the DFE can still be achieved.
The next proposition, which is  a generalization of an analogous result for the bivirus setting (see \cite[Theorem~1]{liu2019analysis}), formalizes the same.
\begin{prop}\cite[Lemma~2]{pare2021multi}
Consider system~\eqref{eq:x1}-\eqref{eq:x3} under Assumption~\ref{assum:base}. If $s(-D^k+B^k) \leq 0$, for each $k \in [3]$, 
then the DFE is the unique equilibrium of system~\eqref{eq:x1}-\eqref{eq:x3}. Moreover, it is asymptotically stable with the domain of attraction $\mathcal D$.\label{prop:phil}
\end{prop}



Note that Proposition~\ref{prop:exp:convergence} and~\ref{prop:phil} provide guarantees on convergence to the DFE. It is natural to ask what happens if, for some $k \in [m]$, one of the eigenvalues of the matrix $-D^k+B^k$ has a positive real part. It turns out for each $k \in [m]$, that violates the eigenvalue condition in Proposition~\ref{prop:phil}, there exists an equilibrium of the form $(\textbf{0}, \dots, \Tilde{x}^k, \dots, \textbf{0})$ in $\mathcal{D}$. 
 The following proposition formalizes this.

\begin{prop} 
\cite[Theorem~2.1]{fall2007epidemiological} \label{prop:necessity}
Consider system~\eqref{eq:x1}-\eqref{eq:x3} under Assumptions~\ref{assum:base} and~\ref{assum:irreducible}. For each $k \in [3]$, 
there exists a unique 
equilibrium $(\textbf{0}, \dots, \Tilde{x}^k, \dots, \textbf{0})$ in $\mathcal{D}$, with $\textbf{0} \ll \Tilde{x}^k \ll \textbf{1}$ if, and only if, $s(B^k - D^k) > 0$.

\end{prop}
Indeed, it is straightforward to show that $\tilde x^k$ is the endemic equilibrium for the single virus system (i.e., $m=1$) defined by the pair of matrices ($D^k,B^k$).
Assuming $m=1$, analytic methods for computing the 
endemic equilibrium\footnote{In the single-virus case, an endemic equilibrium, when it exists, is unique} have been provided in 
\cite[Theorem~5]{van2008virus}. 

Any non-zero equilibrium in $\mathcal D$ is 
referred to as an \emph{endemic} equilibrium. Endemic equilibria can be further classified as follows:
Equilibria of the form $(\textbf{0}, \dots, \Tilde{x}^k, \dots, \textbf{0})$ are referred to as the \emph{boundary equilibria}. The equilibria of the form $(\bar{x}^1, \bar{x}^2, \bar{x}^3)$, 
where at least $\bar{x}^i$ and $\bar{x}^j$ ($i, j \in [3], i\neq j$) 
are nonnegative vectors with at least one positive entry in each of $\bar{x}^i$ and $\bar{x}^j$ are referred to as \emph{coexistence equilibria}. 
The coexistence equilibria can be further classified as i) 2-coexistence equilibria, which are  equilibria of the form $(\bar{x}^1, \bar{x}^2, \bar{x}^3)$ where  $\bar{x}^i$ and $\bar{x}^j$ for some ($i, j \in [3], i\neq j$) 
are nonnegative vectors with at least one positive entry in each of $\bar{x}^i$ and $\bar{x}^j$, and, for $k\neq i, k \neq j$ there holds $\bar{x}^k=\textbf{0}$, and ii) 3-coexistence equilibria, which are  equilibria of the form $(\bar{x}^1, \bar{x}^2, \bar{x}^3)$ where, for each $i \in [3]$,  $\bar{x}^i$ is a nonnegative vector, that has at least one positive entry.
Indeed, such vectors are in fact strictly positive; see \cite[Lemma~6]{axel2020TAC}. That is, at any endemic equilibrium, there is at least one individual in every population node who is infected.

Let  $J(x^1,x^2,x^3)$ denote the Jacobian matrix of system~\eqref{eq:x1}-~\eqref{eq:x3} for an arbitrary  point in the state space. For future reference, it is easily checked that  $J(x^1,x^2,x^3)$ is as given in~\eqref{jacob} below.

	{\noindent}
\begin{align}\label{jacob}
&J(x^1,x^2,x^3) = \\
&\scriptsize
\begin{bmatrix}
-D^1+(I-X^1-X^2-X^3)B^1-\diag(B^1x^1) & -\diag(B^1x^1)   & -\diag(B^1x^1)  \\
-\diag(B^2x^2) & -D^2+(I-X^1-X^2-X^3)B^2-\diag(B^2x^2)  & -\diag(B^2x^2)\\
 -\diag(B^3x^3)& -\diag(B^3x^3)& -D^3+(I-X^1-X^2-X^3)B^3-\diag(B^3x^3)  \end{bmatrix}\normalsize.\nonumber
\end{align}

 \subsection{Preliminary lemmas}

 \begin{lem} \label{lem:eigspec}
\cite[Proposition~1]{liu2019analysis} Suppose that $\Lambda$ is a negative diagonal matrix and $N$ is an irreducible nonnegative matrix. 
Let $M$ be the irreducible Metzler matrix $M = \Lambda+N$. 
Then, $s(M) < 0$ if and only if $\rho(-\Lambda^{-1} N) < 1, s(M)=0$ if and only if $\rho(-\Lambda^{-1} N) = 1$, and $s(M)>0$ if and only if, $\rho(-\Lambda^{-1} N) > 1$.
\end{lem}
We will also use the following variants of the Perron-Frobenius theorem for irreducible matrices.

\begin{lem} \label{lem:perron_frob}
\cite[Chapter 8.3]{meyer2000matrix} \cite[Theorem~2.7]{varga1999matrix} 
Suppose that $N$ is an irreducible nonnegative matrix. Then,
\begin{enumerate}[label=(\roman*)]
    \item $r = \rho(N)$ is a simple eigenvalue of $N$. \label{item:perfrob_simpleeig}
    \item There is an eigenvector $\zeta \gg \textbf{0}$ corresponding to the eigenvalue $r$. \label{item:perfrob_pos_exists}
    \item $x > \textbf{0}$ is an eigenvector only if $Nx = rx$ and $x \gg \textbf{0}$. 
    \label{item:perfrob_pos_necess}
    \item If $A$ is a nonnegative matrix such that $A < N$, then $\rho(A) < \rho(N)$. \label{item:perfrob_matrix_ineq}
\end{enumerate}
\end{lem}
%
\begin{lem} \label{lem:perron_frob_metz}
\cite[Lemma~2.3]{varga1999matrix} Suppose that $M$ is an irreducible Metzler matrix. Then $r = s(M)$ is a simple eigenvalue of $M$, and there exists a corresponding eigenvector  $\zeta \gg \textbf{0}$.
\end{lem}

\begin{lem}[Restrictions on the trajectories of the tri-virus system] Consider system~\eqref{eq:x1}-\eqref{eq:x3} under Assumption~\ref{assum:base}. Suppose that the initial conditions satisfy a) $x^k(0) \gg \textbf{0}$ for $k \in [3]$, and ii) $(x^1(0), x^2(0), x^3(0)) \in \mathcal D$. Further, suppose that matrix $B^k$ for $k \in [3]$ is irreducible.  For all finite $t>0$, $\textbf{0} \ll x^k(t) \ll \textbf{1}$ for $k \in [3]$; and $x^1(t)+x^2(t)+x^3(t) \ll \textbf{1}$.\label{lem:inward:pointing}
\end{lem}
\par The proof closely follows that of \cite[Lemma~3.2]{ye2021convergence} for the bivirus problem. 
\par \textit{Proof:}  Define $z:=\textbf{1}-x^1-x^2-x^3$, and 
$\hat{B}^i:=\diag(B^ix^i)$. Hence,~\eqref{eq:x1}-\eqref{eq:x3} can be rewritten as:
\begin{align}\label{eq:trajectories}
    \dot{x}^i(t)&=-D^ix^i(t)+\hat{B}^iz(t), i=1,2,3 \nonumber\\
    \dot{z}(t)&= D^1x^1(t)+D^2x^2(t)+ D^3x^3(t)-[\hat{B}^1+\hat{B}^2+\hat{B}^3]z(t)
\end{align}
Suppose that for some $\tau \in \mathbb{R}_{>0}$, and for some $i \in [n]$, $z_i(\tau)=0$, which, since $z_i(\tau)=1-x_i^1(\tau)-x_i^2(\tau)-x_i^3(\tau)$, implies that either $x_i^1(\tau)\neq 0$ and/or $x_i^2(\tau)\neq 0$ and/or $x_i^3(\tau)\neq 0$. Therefore, since a) $D^1, D^2, \text{and } D^3$ are positive diagonal matrices, and b) since, by assumption, $z_i(\tau)=0$, it must be that
 $\dot{z}_i(\tau)>0$. 
This implies that $z_i(t) > 0$ for all $t > \tau$. This argument holds for all $i$,
and hence we have that $z_i(t)>0$ for all $t$ and $i \in [n]$. As a result, we obtain $z(t) \gg \textbf{0}$ for all $t >\tau$, which implies that $x^1(t)+x^2(t)+x^3(t) \ll \textbf{1}$. Hence, we have $x^k(t) \ll \textbf{1}$ for $k \in [3]$.
\par Suppose that for some $t \in \mathbb{R}_{\geq 0}$, $\ell$ (where $\ell <n$) 
entries in
$x^1(t)$ equal zero, and let us label these 
entries $i_1, i_2, \hdots, i_\ell$. Since, by assumption, the matrix $B^1$ is irreducible, it follows from the property of irreducible matrices that there is at least one entry 
in the vector $B^1x^1$ that is nonzero even though the corresponding 
entry in the vector $x^1$ equals zero \cite[Lemma~1]{liu2019analysis}. Let us assume that this occurs for the $i_1^{th}$ entry, i.e., $x^1_{i_1}=0$ yet $[B^1x^1]_{i_1}>0$. Observe that the evolution of the infection level for virus~1 in node $i_1$ is as follows:
\begin{align}\label{ineq:nozeros}
    \dot{x}^1_{i_1}(t)&=[-D^1x^1(t)]_{i_1} +[(I-X^1-X^2-X^3)B^1x^1(t)]_{i_1} \nonumber \\
    &>0,
\end{align}
where the inequality in~\eqref{ineq:nozeros} follows by noting that $x_{i_1}(t)=0$ by assumption, and $[B^1x^1]_{i_1}>0$ as discussed above. This means that there must exist some time instant $t^\prime$, with $t^\prime-t$ not too large, such that $x^1(t^\prime)$ has fewer than $\ell$ zero entries. Repeating the argument for all the other zero entries in the  vector $x^1(t^\prime)$ (and this can be done since the choice of node $i_1$ was arbitrary), we have that $x^1(t) \gg \textbf{0}$ for $t>0$. Analogously, we can prove that $x^2(t), x^3(t) \gg \textbf{0}$ for $t>0$.~\qed
\par Lemma~\ref{lem:inward:pointing}  substantially limits the equilibria that can lie on the boundary of the set $\{x^1, x^2, x^3 \in \mathbb{R}_{\geq 0}\mid x^1+x^2+x^3 \leq \textbf{1}\}$.

It turns out that if $x = (x^1, \dots, x^3) \in \mathcal{D}$ is an equilibrium of~\eqref{eq:x1}-\eqref{eq:x3}, then a certain dichotomy arises for $x^k$, $k=1,2,3$. The following lemma formalizes said dichotomy.

\begin{lem}{\cite[Lemma~6]{axel2020TAC}} \label{lem:equi_non-zero_nonone:1}
Consider system~\eqref{eq:x1}-\eqref{eq:x3} under Assumptions~\ref{assum:base} and~\ref{assum:irreducible}. 
 If $x = (x^1, \dots, x^3) \in \mathcal{D}$ is an equilibrium of~\eqref{eq:x1}-\eqref{eq:x3}, then, for each $k \in [3]$, either $x^k = \textbf{0}$, or $\textbf{0} \ll x^k \ll \textbf{1}$. Moreover, 
$\textstyle \sum_{k=1}^3 x^k \ll \textbf{1}$.
\end{lem}


\subsection{Problem Statements} 

The focus of the present paper is  to answer the following questions:
\begin{enumerate}[label=\roman*)]
  \item Is the tri-virus system monotone?
 \item Does the tri-virus system generically admit a finite number of equilibria?
    \item Can we identify a sufficient condition for local exponential convergence to a boundary equilibrium? 
    \item Can we identify a sufficient condition for the nonexistence of 3-coexistence equilibria (resp. specific forms of 2-coexistence equilibria)?
\item Can we identify a sufficient condition for the existence of a 3-coexistence equilibrium?
    \item Can we identify a sufficient condition for the existence and local attractivity of a line of coexistence equilibria?
    \item Can we identify special case(s) where, irrespective of the non-zero initial infection levels, the tri-virus dynamics converge to a plane of coexistence equilibria?
\end{enumerate}

\section{Monotonicity (or lack thereof) of the tri-virus system} \label{sec:trivirus:monotone}
Note that the system defined by \eqref{eq:x1}-\eqref{eq:x3} is a (networked) nonlinear system. A particular class of nonlinear systems is \emph{monotone dynamical systems} (MDS). For a detailed overview of MDS, see, for instance, \cite{smith2008monotone}. Before investigating whether or not the tri-virus system is monotone, we detail the importance of MDS in the context of the competitive bi-virus SIS networked model.


\subsection{Monotone dynamical systems and competitive bivirus networked SIS models} 


Notice that if $m=2$, then, under Assumption~\ref{assum:irreducible}, system~\eqref{eq:full} is monotone; see \cite[Lemma~3.3]{ye2021convergence} (and, for $k=1,2$, assuming $\delta_i^k=1$ for $i \in [n]$, also  \cite[Theorem~18]{santos2015bi}). 
 That is, setting $m=2$ for system~\eqref{eq:full}, suppose that $(x_A^1(0), x_A^2(0))$ and $(x_B^1(0), x_B^2(0))$ are two initial conditions in $\textrm{int}(D)$ satisfying i) $x_A^1(0)>x_B^1(0)$ and ii) $x_A^2(0)<x_B^2(0)$. Since the bivirus system is monotone, it follows that, for all $t \in \mathbb{R}_{\geq 0}$, i) $x_A^1(t)\gg x_B^1(t)$ and ii) $x_A^2(t)\ll x_B^2(t)$. 
Additionally, it is also known that 
for almost all \footnote{The term \enquote{almost all} has a precise mathematical meaning: for all but a set of parameter values that has measure zero. This set of exceptional values is defined by an algebraic or semi-algebraic set.} choices of $D^i$, $B^i$, $i=1,2$, system~\eqref{eq:full} has a finite number of equilibria \cite[Theorem~3.6]{ye2021convergence}. Therefore, from \cite[Theorems~2.5 and~2.6]{smith1988systems} (or \cite[Theorem~9.4]{hirsch1988stability}), we know that for almost all initial conditions in $\mathcal D$, system~\eqref{eq:full} with $m=2$ converges to a stable equilibrium point, assuming such an equilibrium  exists. There are initial conditions for which convergence to a stable equilibrium does not occur. In such cases, said initial condition is  either 
a) 
in the stable manifold of an unstable equilibrium, 
or b) in the stable manifold of a nonattractive limit cycle. The set of initial conditions for which either a) or b) 
happens has measure zero. 


\subsection{The tri-virus system is not monotone}\label{sec:tri:virus:loss:of:generality}
\par It is natural to ask how well the notion of MDS 
applies to 
multi-competitive networked SIS epidemics; this subsection aims to answer this question conclusively.
To this end, we construct a graph associated with the Jacobian \eqref{jacob} of system~\eqref{eq:x1}-\eqref{eq:x3}, say $\bar{G}$. 
The construction follows the outline provided in \cite{sontag2007monotone}. More specifically,
the graph $\bar{G}$ has $3n$ nodes. The edges of  $\bar{G}$ are based on the entries in the Jacobian matrix $J(x^1, x^2, x^3)$ in \eqref{jacob}. Specifically, if $[J(x^1, x^2, x^3)]_{ij} <0$ for  $i \neq j$, then we draw an edge labelled with ``-" sign;  if  $[J(x^1, x^2, x^3)]_{ij} > 0$ for  $i \neq j$, then we draw an edge labelled with ``+" sign. Thus,  $\bar{G}$ is a signed graph. Note that $\bar{G}$ has no self-loops. As an aside, also observe that since $x^k(t)\geq 0$ for $k \in [3]$ and $t\in \mathbb{R}_+$, it is immediate that the sign of the elements in $J(x^1, x^2, x^3)$ do not change with the argument so that $\bar G$ 
is the same for all points in the interior of $\mathcal D$.  

\par We also need the following concept from graph theory. A signed graph is considered consistent if every undirected cycle in the graph has a net positive sign, i.e., it has an even number of ``-" signs \cite{sontag2007monotone}. We have the following result.
\begin{thm} \label{prop:tri-virus-not-monotone}
System~\eqref{eq:x1}-\eqref{eq:x3} is not monotone.
\end{thm}
\begin{proof}
Note that the Jacobian $J(x^1, x^2, x^3)$ is a block matrix, with all blocks along the off-diagonal being negative diagonal matrices. Pick any node $i$, where $i \in \{1,2, \hdots, n\}$. Observe that, since  all blocks along the off-diagonal of $J(x^1, x^2, x^3)$ are negative diagonal matrices, it is clear that there exists an edge from node $i$ to node $i+n$, an edge from node $i+n$ to node $i+2n$, and an edge from node $i+2n$ to node $i$. Furthermore, each of these edges has a ``-" sign. Hence, a loop starting from node $i$, traversing through nodes $i+n$, $i+2n$, and back to node $i$ is a 3-length cycle with an odd number of negative signs. Therefore, from \cite[page 62]{sontag2007monotone}, the signed graph $\bar{G}$ is not consistent. 
Consequently, from \cite[page 63]{sontag2007monotone}, it follows that the system~\eqref{eq:x1}-\eqref{eq:x3} is not monotone.\hfill \proofbox
\end{proof}


\par 
The conclusion of Theorem~1 offers several important and interesting insights, including key differences with bi-virus systems (m = 2). 
The fact that a bivirus system is monotone coupled with the fact that for almost all choices of $D^k$, $B^k$, $k=1,2$, the bivirus system has a finite number of equilibria allows one to draw general conclusions on the limiting behavior of bivirus dynamical systems. One can thus focus on the characterization of equilibria given network parameters $D^i$, $B^i$, such as the number of equilibria and their regions of attraction; it is known that bivirus systems can have multiple stable endemic equilibria~\cite{anderson2022equilibria}. A consequence of Theorem~\ref{prop:tri-virus-not-monotone} 
is that one cannot draw upon the rich literature on monotone dynamical systems (see\cite{smith1988systems}) to study the limiting behavior of system~\eqref{eq:x1}-\eqref{eq:x3}. In general, for non-monotone systems, no dynamical behavior, including chaos, can be definitively ruled out 
\cite{sontag2007monotone}. Indeed, and in contrast to the bivirus system, the basic convergence properties of the trivirus system to endemic equilibria are not well understood.

Another possible consequence of the lack of monotonicity is as follows: It is  known that setting $D^k=I$ for $k \in [2]$ has no bearing on either the location of equilibria of system~\eqref{eq:full}
with $m=2$ nor on their (local) stability properties \cite[Lemma~3.7]{ye2021convergence}. That is, consider two bivirus systems, namely $\mathcal S$ and $\hat{\mathcal {S}}$, where $\mathcal S$ is defined by $(B^1,D^1, B^2, D^2)$ and $\hat{\mathcal {S}}$ is defined by $(\hat{B}^1 =(D^1)^{-1}B^1, \hat{D}^1=I,  \hat{B}^2 =(D^2)^{-1}B^2, \hat{D}^2=I)$. Then, the location of equilibria is the same for both bivirus systems. Moreover, local stability of an equilibrium in  bivirus system $\mathcal S$ implies, and is implied by,  that in bivirus system $\hat{\mathcal {S}}$.
For system~\eqref{eq:full} with $m=3$, by extending the arguments from \cite[Lemma~3.7]{ye2021convergence},  it is straightforward to show that the \emph{location} of the equilibria is the same  for an arbitrary system defined by $D^k, B^k$, $k \in [3]$ and the system defined by $\hat{D}^k=I$, and $\hat{B}^k=(D^k)^{-1}B^k$.
However, since the tri-virus system is not monotone, the arguments for the stability of equilibria in the proof of \cite[Lemma~3.7]{ye2021convergence} cannot be adapted. 
Hence, for the tri-virus case,  when the healing rates for all nodes with respect to all viruses are `scaled' in the manner above to become unity, the preservation of stability properties remains an open question.

As a matter of independent interest, in the next section, we ask whether 
for almost all choices of $D^k$, $B^k$, $k=1,2,3$, the trivirus system has a finite number of equilibria.

\section{Finiteness of equilibria for generic trivirus networks}\label{sec:finiteness:of:equiibria}
In this section, we 
show that for generic tri-virus networks, i.e., for almost all  choices of $D^k,B^k, k=1,2,3$, the number of equilibria is finite, and the associated Jacobian matrices are nonsingular. 
For  bi-virus networks, an argument for the corresponding result based on algebraic geometry ideas was provided in \cite[Theorem~3.6]{ye2021convergence}. That argument becomes much more intricate for tri-virus systems. So an alternative proof is provided, based on a topological tool, termed the Parametric Transversality Theorem, see \cite[p.145]{lee2013introduction} and \cite[p.68]{guillemin2010differential}.
This tool has the advantage that it will apply to variations of the tri-virus dynamics, where the right side of \eqref{eq:scalar} is replaced by $-f(x_i^k) + (1-\sum x_i^k) g(x^1,x^2,x^3)$ for general nonlinear functions $f$ and $g$, such as those arising where feedback control is present, see, e.g., \cite{ye2021_PH_TAC}, or more refined models are constructed, see, e.g., \cite{yang2017bi,doshi2022convergence}.

\begin{prop}\label{prop:finite:equilibria}
For generic parameters $D^i$, $B^i$, with $i=1,2,3$, the trivirus system~\eqref{eq:x1}-\eqref{eq:x3} has a finite number of equilibria which are all nondegenerate, i.e. the associated Jacobian matrices are nonsingular. 
\end{prop}

The proof of this proposition is contained in the appendix, using as it does very different tools to the rest of the paper. The first component of the proof establishes the nondegeneracy claim. Further, we prove a slight modification, viz. that for any fixed set of $B^k$, almost all choices of the $D^k$ result in a finite number of equilibria. Proving that for fixed $D^k$, almost all choices of $B^k$ result in a finite number of equilibria is achieved with completely analogous calculations, but these calculations are not provided. Together, these two observations establish the proposition. Note that one \textit{cannot} replace the `almost all' requirement with an `all' requirement. As for the bivirus problem, there can exist special values of the parameters for which there is a continuum of equilibria; see Sections~\ref{sec:line:attractivity} and~\ref{sec:global:plane} for details.

One of the practical ramifications of Proposition~\ref{prop:finite:equilibria} is recorded in the following remark.
\begin{rem}
The fact that, for almost all choices of infection and recovery rates, the  tri-virus networked SIS model has a finite number of equilibria means that assuming chaos and limit cycles are known to be not present, then,
 even in the absence of interventions for controlling the spread, the range of outcomes that health administration officials need to prepare for are limited (although not necessarily few). 
\end{rem}

\section{Persistence of one or more viruses}\label{sec:persistence:one:virus}

In this section, we first show that for low-dimensional systems existence of certain kinds of equilibria can be ruled out for almost all choices of system parameters. Further, we identify a sufficient condition for local exponential convergence to a boundary equilibrium.  

\subsection{The special case of low dimension systems}
When $n=1$ or $n=2$, certain types of equilibria are generically impossible, as we now show. 
\begin{prop}\label{prop:brian:1}
Consider system~\eqref{eq:x1}-\eqref{eq:x3} and suppose that $n=1$. Suppose that $(\tilde x^1,\tilde x^2,\tilde x^3)$ is an equilibrium.  Provided no two values of $(D^k)^{-1}B^k$ (with $k=1,2,3$) are identical, at most one  of the $\tilde x^k$ is nonzero at an equilibrium.
\end{prop}
\textit{Proof:}
At any equilibrium, there holds (recall that the $\tilde x^k$ are scalar, so that $\tilde X^k=\tilde x^k)$
\[
(-D^k+(1-\tilde x^1-\tilde x^2-\tilde x^3)B^k)\tilde x^k=0
\]

 Note that if all $B^i$ are nonzero, there cannot hold $(1-\tilde x^1-\tilde x^2-\tilde x^3)=0,$
 else a contradiction is easily obtained. 

It is immediate that either 
\[(1-\tilde x^1-\tilde x^2-\tilde x^3)^{-1}=(D^k)^{-1})B^k\quad\mbox{or}\quad\tilde x^k=0
\]
The first alternative corresponds to $\tilde x^k\neq 0$, and if, for example, $\tilde x^1\neq 0, \tilde x^2\neq 0$ then the parameters are constrained by $(D^1)^{-1}B^1=(D^2)^{-1}B^2$. \hfill \proofbox

For the case $n=2$, we shall appeal to the following lemma.

\begin{lem} 
Consider three positive matrices $ B^i, i=1,2,3$ of dimension $2\times 2$. For generic values of the matrix entries, there cannot exist a diagonal matrix, $\Lambda$ say, such that each of $\Lambda B^i$ has a unity eigenvalue.
\end{lem}
\textit{Proof:}
Denote the diagonal entries of $\Lambda$ by $\lambda_1,\lambda_2$, The eigenvalue condition on $\Lambda B^i$ is det$(I-\Lambda B^i)=0$, or 
\[(\beta^i_{11}\beta^i_{22}-\beta^i_{12}\beta^i_{21})\lambda_1\lambda_2-\beta^i_{11}\lambda_1-\beta^i_{22}\lambda_2+1=0
\]
For fixed $B^i$, the set of $\lambda_i$ satisfying this equation is evidently defined by a  rectangular hyperbola. It is clear that for generic $B^i$, there cannot be a common point of intersection between three such hyperbolae. Algebraic geometry gives methods for calculating the actual semialgebraic set of $B^i$ for which a common point of intersection exists, but this is irrelevant for our purposes. For almost all $B^i$, no such intersection exists.~\proofbox

Now we can state the following.

\begin{prop}\label{prop:brian:2}
Consider system~\eqref{eq:x1}-\eqref{eq:x3} and suppose that $n=2$.  Suppose that $(\tilde x^1,\tilde x^2,\tilde x^3)$ is an equilibrium. Then for almost all values of the parameters $D^k,B^k$ ($k=1,2,3$), at least one $\tilde x^k$  must be equal to $\textbf{0}$.
\end{prop}
\textit{Proof:}
The equilibrium equations yield
\[
(-D^k+(I-\tilde X^1-\tilde X^2-\tilde X^3)B^k)\tilde x^k=0
\]
Set $\Lambda=I-\tilde X^1-\tilde X^2-\tilde X^3$ and observe that the equilibrium equations can be rewritten as 
\[
\tilde x^k=\Lambda[(D^k)^{-1}B^k]\tilde x^k
\]
Hence either $\tilde x^k$ is an eigenvector of $\Lambda[(D^k)^{-1}B^k]$ or it is zero. The preceding lemma applies (with $(D^k)^{-1}B^k$ replacing $B^k$ in the lemma statement) and the conclusion of the proposition is immediate. \hfill  \proofbox
\par The proposition indicates that for generic $D^i$, $B^i$, and with $n=2$, the only equilibria of the trivirus system are the healthy equilibrium, those defined by the three associated single virus systems, and those defined by the three associated bivirus systems. Note that for $n=2$, all these equilibria are analytically computable, \cite{ye2021convergence}.




It turns out that we cannot have multiple 3-coexistence equilibria that  differ  in the coordinate vector corresponding to only one of the viruses; we formalize the same in the following lemma.
\begin{lem} \label{lem:equi_non-zero_nonone}
Consider system~\eqref{eq:x1}-\eqref{eq:x3} under Assumptions~\ref{assum:base} and~\ref{assum:irreducible}. 
If $(\bar{x}^1, \bar{x}^2, \bar{x}^3)$ and  $(\bar{x}^1, \bar{x}^2, \hat{x}^3)$ are endemic  equilibria (i.e., $\textbf{0}\ll (\bar{x}^1,\bar{x}^2, \bar{x}^3) \ll \textbf{1}$ and $\textbf{0}\ll  \hat{x}^3 \ll \textbf{1}$) of~\eqref{eq:x1}-\eqref{eq:x3}, then $\bar{x}^3 =\hat{x}^3$.
\end{lem}
\textit{Proof:} 
Suppose that $(\bar{x}^1, \bar{x}^2, \bar{x}^3)$ and  $(\bar{x}^1, \bar{x}^2, \hat{x}^3)$ are endemic equilibria of~\eqref{eq:x1}-\eqref{eq:x3}, then by writing the equilibrium version of equation~\eqref{eq:x1}, we obtain:
\begin{align}
    \textbf{0}&=-D^1\bar{x}^1+((I-\bar{X}^1-\bar{X}^2-\bar{X}^3))B^1\bar{x}^1 \label{eq:x1bar}\\
    \textbf{0}&=-D^1\bar{x}^1+((I-\bar{X}^1-\bar{X}^2-\hat{X}^3))B^1\bar{x}^1 \label{eq:x1hat}
\end{align}
From~\eqref{eq:x1bar} and~\eqref{eq:x1hat}, it is clear that
\begin{align}
    \bar{x}^1&=(D^1)^{-1}((I-\bar{X}^1-\bar{X}^2-\bar{X}^3))B^1\bar{x}^1\nonumber \\
    &=(D^1)^{-1}((I-\bar{X}^1-\bar{X}^2-\hat{X}^3))B^1\bar{x}^1, \nonumber
\end{align}
which, since i) $\bar{x}^1 \gg \textbf{0}$ and ii) by Assumption~\ref{assum:irreducible} it follows that
for each $i \in [n]$, $\sum_{j=1}^{n}B^1_{ij} >0$,
which implies that $\bar{X}^3=\hat{X}^3$, i.e., $\bar{x}^3=\hat{x}^3$.\hfill \proofbox
\qed

\subsection{Stability of boundary equilibria}
In this subsection, 
we establish
that the local stability (resp. instability) of the boundary equilibrium corresponding to 
virus~1
is dependent on whether (or not) the state matrices, obtained by linearizing the dynamics of viruses~$2$ and~$3$ around the boundary  
equilibrium of virus~$1$, are Hurwitz.


\begin{thm}\label{thm:local}
Consider system~\eqref{eq:x1}-\eqref{eq:x3} under Assumptions~\ref{assum:base} and~\ref{assum:irreducible}. \break
The boundary equilibrium $(\Tilde{x}^1, \textbf{0}, \textbf{0})$ is locally exponentially 
stable if, and only if, 
each of the following conditions is satisfied: 
\begin{enumerate}[label=\roman*)]
    \item $\rho((I-\tilde{X}^1)(D^2)^{-1}B^2)<1$; and
    \item $\rho((I-\tilde{X}^1)(D^3)^{-1}B^3)<1$.
\end{enumerate}
If $\rho((I{-}\tilde{X}^1)(D^2)^{-1}B^2)>1$ or if \mbox{$\rho((I{-}\tilde{X}^1)(D^3)^{{-}1}B^3)>1$}, then $(\Tilde{x}^1, \textbf{0}, \textbf{0})$ is unstable.
\end{thm}
The proof follows the strategy outlined for the $m=2$ case in~\cite[Theorem~3.10]{ye2021convergence}.\\
\textit{Proof:} Consider the equilibrium point $(\Tilde{x}^1, \textbf{0}, \textbf{0})$, and note that the Jacobian evaluated at this point is as follows:
\begin{align}\label{eq:boundaryJac}
&J(\Tilde{x}^1,\textbf{0}, \textbf{0}) = \\
&\scriptsize
\begin{bmatrix}
-D^1+(I-\Tilde{X}^1)B^1-\hat{B}^1  & {-}\hat B^{1}   & {-}\hat B^{1} \\
 \textbf{0} & {-}D^2{+}(I{-}\Tilde{X}^1)B^2 & \textbf{0} \\
\textbf{0}  & \textbf{0} & {-}D^3{+}(I{-}\Tilde{X}^1)B^3 \end{bmatrix}\normalsize,\nonumber
\end{align}
where $\hat{B}^i=\diag(B^i\Tilde{x}^i)$, for $i=1,2,3$.\\
Observe that the matrix $J(\Tilde{x}^1,\textbf{0}, \textbf{0})$ is block upper triangular. Hence, it is Hurwitz if, and only if,  the blocks along the diagonal are Hurwitz. We will now show that this condition is fulfilled as a consequence of the assumptions of Theorem~\ref{thm:local}.
\par Since $(\Tilde{x}^1, \textbf{0}, \textbf{0})$ is an equilibrium point of system~\eqref{eq:x1}-\eqref{eq:x3}, by considering the equilibrium version of equation~\eqref{eq:x1}, we have the following:
\begin{align}\label{eq:eqm:version:1}
    (-D^1+(I-\Tilde{X}^1)B^1)\tilde{x}^1=\textbf{0}.
\end{align}
By Assumption~\ref{assum:base}, we have that $D^1$ is positive diagonal and $B^1$ is nonnegative. Furthermore, by assumption, we know that $B^1$ is irreducible. Moreover, from \cite[Lemma~6]{axel2020TAC}, it follows that $(I-\Tilde{X}^1)$ is positive diagonal, implying that $(I-\Tilde{X}^1)B^1$ is nonnegative irreducible. Thus, we can conclude that the matrix $(-D^1+(I-\Tilde{X}^1)B^1)$ is irreducible Metzler. 
From \cite[Lemma~6]{axel2020TAC}
we know that $\textbf{0} \ll \tilde{x}^1$. Hence, by applying \cite[Lemma~2.3]{varga1999matrix} to~\eqref{eq:eqm:version:1}, it must be that $\tilde{x}^1$ is, up to scaling, the only eigenvector of $(-D^1+(I-\Tilde{X}^1)B^1)$ with all entries being strictly positive. Furthermore, $\tilde{x}^1$ is the eigenvector that is associated with, and only with, $s(-D^1+(I-\Tilde{X}^1)B^1)$. Therefore,   $s(-D^1+(I-\Tilde{X}^1)B^1)=0$.\\
Define $Q:=D^1-(I-\Tilde{X}^1)B^1$, and note that $Q$ is an M-matrix. Since $s(-Q)=0$ and $B^1$ is irreducible, $Q$ is a singular irreducible M-matrix. Observe that $\hat{B}^1$ is a nonnegative matrix and because $B^1$ is irreducible and $\Tilde{x}^1 \gg \textbf{0}$, it must be that at least one element in $\hat{B}^1$ is strictly positive. Therefore, from \cite[Lemma~4.22]{qu2009cooperative}, it follows that $Q+\hat{B}^1$ is an irreducible non-singular M-matrix, which from \cite[Section~4.3, page~167]{qu2009cooperative} implies that $-Q-\hat{B}^1$ is Hurwitz. Therefore, we have that $s(-D^1+(I-\Tilde{X}^1)B^1-\hat{B}^1)<0$.
\par By assumption, $\rho((I-\tilde{X}^1)(D^2)^{-1}B^2)<1$ and $\rho((I-\tilde{X}^1)(D^3)^{-1}B^3)<1$. Therefore, by noting that $D^2$ (resp. $D^3$) are positive diagonal matrices and $B^2$ (resp. $B^3$) are nonnegative irreducible matrices,
from 
\cite[Proposition~1]{liu2019analysis} it follows that $s(-D^2+(I-\tilde{X}^1)B^2)<0$ (resp. $s(-D^3+(I-\tilde{X}^1)B^3)<0$). Since each diagonal block of $J(\Tilde{x}^1,\textbf{0}, \textbf{0})$ is Hurwitz, it is clear that $J(\Tilde{x}^1,\textbf{0}, \textbf{0})$ is Hurwitz. Local 
exponential 
stability of $(\Tilde{x}^1,\textbf{0}, \textbf{0})$, then, follows from \cite[Theorem 4.15 and Corollary~4.3]{khalil2002nonlinear}. 
\par The proof of necessity follows by first noting that if either condition in statement~i) or that in statement~ii) is violated, then, since at least one of the blocks along the diagonal of $J(\Tilde{x}^1,\textbf{0}, \textbf{0})$ is not Hurwitz, the  matrix $J(\Tilde{x}^1,\textbf{0}, \textbf{0})$ is not Hurwitz. Then, by invoking the necessary part of \cite[Theorem 4.15 and Corollary~4.3]{khalil2002nonlinear}, the result follows.
\par The claim for instability can be proved by noting that if either of the eigenvalue conditions is violated, then, since  $J(\Tilde{x}^1,\textbf{0}, \textbf{0})$ is block diagonal, the matrix  $J(\Tilde{x}^1,\textbf{0}, \textbf{0})$ is not Hurwitz. The result follows from \cite[Theorem~4.7, item ii)]{khalil2002nonlinear}.~\hfill \proofbox
\par Analogous results for the boundary equilibria $(\textbf{0}, \Tilde{x}^2, \textbf{0})$ and $(\textbf{0}, \textbf{0},\Tilde{x}^3)$ can be similarly obtained.

\section{Nonexistence of various coexistence equilibria}\label{sec:nonexistence:coexistence:equilibria}
In this section,  we establish conditions on the system parameters that guarantee the nonexistence of different kinds of coexistence equilibria. Specifically, we identify a condition for the nonexistence of a 3-coexistence (resp. 2-coexistence) equilibria. We draw on the stability properties established in the previous section.


\begin{thm}[Nonexistence of 3-coexistence equilibria]\label{claim:virus3:strongest}
Consider system~\eqref{eq:x1}-\eqref{eq:x3} under Assumption~\ref{assum:base}. Suppose that matrix $B^k$ for $k \in [3]$ is irreducible. Suppose that  $\rho((D^k)^{-1}B^k)>1$ for $k \in [3]$. Let $\tilde{x}^1$,  $\tilde{x}^2$, and $\tilde{x}^3$ denote the single-virus endemic equilibrium corresponding to virus~1,~2, and~3, respectively. If 
$(D^3)^{-1}B^3>(D^2)^{-1}B^2>(D^1)^{-1}B^1$, then
\begin{enumerate}[label=\roman*)]
\item the equilibrium point $(\textbf{0}, \textbf{0}, \textbf{0})$ is unstable;
\item the equilibrium point $(\tilde{x}^1, \textbf{0}, \textbf{0})$ is unstable;
\item the equilibrium point $(\textbf{0}, \tilde{x}^2, \textbf{0})$ is unstable;
\item the equilibrium point $(\textbf{0}, \textbf{0}, \tilde{x}^3)$ is locally exponentially stable;
\item there does not exist a 3-coexistence equilibrium.
 \end{enumerate}
\end{thm}
\textit{Proof:} 
The proof strategy is quite similar to those in \cite[Theorem~6]{axel2020TAC} and \cite[Corollary~3.10]{ye2021convergence}.\\
\par The proof of statements~i)-iii) are quite similar to that of statements~i) and ii) in Theorem~\ref{claim:virus1weaker}. 
\textcolor{black}{\textit{Proof of statement~i): }
 Consider the equilibrium point $(\textbf{0}, \textbf{0}, \textbf{0})$, and observe that the Jacobian computed at this point is as follows: 
 \begin{align}
&J(\textbf{0},\textbf{0}, \textbf{0}) = \\
&\footnotesize
\begin{bmatrix}
-D^1+B^1  & \textbf{0}   & \textbf{0} \\
 \textbf{0} & -D^2+B^2 & \textbf{0} \\
\textbf{0}  & \textbf{0} & -D^3+B^3 \end{bmatrix}\normalsize,\nonumber
\end{align}
Since by assumption, $B^k$ is non-negative irreducible and $\rho((D^k)^{-1}B^k)>1$ for $k \in [3]$, it follows from Lemma~\ref{lem:eigspec} that $s(-D^k+B^k)>0$ for $k \in [3]$, which further implies that $s(J(\textbf{0},\textbf{0}, \textbf{0}))>0$. Instability of the equilibrium point $(\textbf{0},\textbf{0}, \textbf{0})$, then, follows from \cite[Theorem~4.7, item ii)]{khalil2002nonlinear}.\\
\textit{Proof of statement~ii):}  Consider the equilibrium point $(\tilde{x}^1, \textbf{0}, \textbf{0})$, and observe that the Jacobian computed at this point is as follows:
\begin{align}
&J(\Tilde{x}^1,\textbf{0}, \textbf{0}) = \\
&\scriptsize
\begin{bmatrix}
-D^1+(I-\Tilde{X}^1)B^1-\hat{B}^1  & -\hat B^{1}   & -\hat B^{1} \\
 \textbf{0} & -D^2+(I-\Tilde{X}^1)B^2 & \textbf{0} \\
\textbf{0}  & \textbf{0} & -D^3+(I-\Tilde{X}^1)B^3 \end{bmatrix}\normalsize,\nonumber
\end{align}
Note that since $\tilde{x}^1$ is the single-virus endemic equilibrium for virus~1, by the definition of an equilibrium point, we have the following:
$$[-D^1+(I-\tilde{X}^1)B^1]\tilde{x}^1=\textbf{0}.$$
Observe that $-D^1+(I-\tilde{X}^1)B^1$ is an irreducible Metzler matrix, and, since $\tilde{x}^1\gg \textbf{0}$, from Lemma~\ref{lem:perron_frob_metz} we have that $s(-D^1+(I-\tilde{X}^1)B^1)=0$. Consequently, from Lemma~\ref{lem:eigspec}, $\rho((I-\tilde{X}^1)(D^1)^{-1}B^1)=1$. By assumption $(D^2)^{-1}B^2> (D^1)^{-1}B^1$, which implies $(I-\tilde{X}^1)(D^1)^{-1}B^1 < (I-\tilde{X}^1)(D^2)^{-1}B^2$. Therefore, since the matrices $(I-\tilde{X}^1)(D^1)^{-1}B^1$ and 
$(I-\tilde{X}^1)(D^2)^{-1}B^2$ are nonnegative, from Lemma~\ref{lem:perron_frob} item iv), we have that $\rho((I-\tilde{X}^1)(D^2)^{-1}B^2)>1$. Consequently, from Lemma~\ref{lem:eigspec}, it must be that $s(-D^2+(I-\tilde{X}^1)B^2)>0$, which implies that $s(J(\Tilde{x}^1,\textbf{0}, \textbf{0}))>0$, and hence the equilibrium point $(\Tilde{x}^1,\textbf{0}, \textbf{0})$ is unstable.\\
 The proof of statement~iii) is analogous to that of statement~ii), and is, therefore, omitted. 
}

\noindent \textit{Proof of statement~iv):} Consider the equilibrium point $(\textbf{0}, \textbf{0}, \tilde{x}^3)$, and observe that the Jacobian is as follows:
\begin{align}
&J(\textbf{0}, \textbf{0}, \Tilde{x}^3) = \\
&\scriptsize
\begin{bmatrix}
-D^1+(I-\Tilde{X}^3)B^1  & \textbf{0}  & \textbf{0}\\
 \textbf{0} & -D^2+(I-\Tilde{X}^3)B^2 & \textbf{0} \\
-\hat B^{3}   & -\hat B^{3}    & -D^3+(I-\Tilde{X}^3)B^3-\hat B^{3}    \end{bmatrix}\normalsize,\nonumber
\end{align}
(Recall that $\hat B^3$ is shorthand for $\diag(B^3\tilde x^3))$. By following similar arguments as in the proof of Theorem~\ref{thm:local} (for establishing that the matrix $-D^1+(I-\Tilde{X}^1)B^1-\hat B^{1}$ is Hurwitz), we can show that the matrix $-D^3+(I-\Tilde{X}^3)B^3-\hat B^{3}$ is Hurwitz, i.e., $s(-D^3+(I-\Tilde{X}^3)B^3-\hat B^{3})<0$. Since $\tilde{x}^3$ is the single-virus endemic equilibrium corresponding to virus~3, we have that
$$\rho((I-\Tilde{X}^3)(D^3)^{-1}B^3)=1.$$
Since, by assumption, $(D^3)^{-1}B^3>(D^2)^{-1}B^2>(D^1)^{-1}B^1$, it is clear that \break
$(D^3)^{-1}B^3>(D^2)^{-1}B^2$ and $(D^3)^{-1}B^3>(D^1)^{-1}B^1$.
Hence, it follows that a) $(I-\Tilde{X}^3)(D^3)^{-1}B^3> (I-\Tilde{X}^3)(D^2)^{-1}B^2$, and b) $(I-\Tilde{X}^3)(D^3)^{-1}B^3> (I-\Tilde{X}^3)(D^1)^{-1}B^1$. Therefore, from Lemma~\ref{lem:perron_frob} item iv), it must be that a) $\rho((I-\Tilde{X}^3)(D^2)^{-1}B^2)<1$, and b) $\rho((I-\Tilde{X}^3)(D^1)^{-1}B^1)<1$, which from Lemma~\ref{lem:eigspec} further implies that a) $s(-D^2+(I-\Tilde{X}^3)B^2)<0$, and b) $s(-D^3+(I-\Tilde{X}^3)B^1)<0$. As a consequence, $s(J(\textbf{0}, \textbf{0}, \Tilde{x}^3))<0$, which, from \cite[Theorem~4.7]{khalil2002nonlinear}, leads us to conclude that the equilibrium point $(\textbf{0}, \textbf{0}, \Tilde{x}^3)$ is locally exponentially stable. \\
\textit{Proof of statement~v):} 
We now show that under the conditions of the Theorem, a 3-coexistence equilibrium (i.e., of the form $\bar{x}=(\bar{x}^1, \bar{x}^2, \bar{x}^3)$ with all $\bar x^i>0$ ) cannot exist. 
Suppose that, by way of contradiction, there exists a 3-coexistence 
equilibrium point $\bar{x}=(\bar{x}^1, \bar{x}^2, \bar{x}^3)$. Therefore, from Lemma~\ref{lem:equi_non-zero_nonone} it must be that $\textbf{0} \ll (\bar{x}^1, \bar{x}^2, \bar{x}^3) \ll \textbf{1}$, and $\bar{x}^1+ \bar{x}^2+ \bar{x}^3 \ll \textbf{1}$. Furthermore, by the definition of an equilibrium point, we also have the following:
\begin{align} 
   \textbf{0} &=  (-D^1 + (I-\bar X^1-\bar X^2-\bar X^3) B^1) \bar x^1, \label{eq:x1_3-equib}\\
      \textbf{0}  &=  (-D^2 + (I-\bar X^1-\bar X^2-\bar X^3) B^2) \bar x^2, \label{eq:x2_3-equib}\\
         \textbf{0}  &=  (-D^3 + (I-\bar X^1-\bar X^2-\bar X^3) B^3) \bar x^3 \label{eq:x3_3-equib}.
  \end{align}
Since $\tilde{x}^k$ is the single-virus endemic equilibrium for virus $k$, $k \in [3]$, we also have the following:
\begin{align} 
   \textbf{0} &=  \Big{(} -D^1+ \big{(} I - \Tilde{X}^1 \big{)} B^1  \Big{)} \tilde{x}^1(t), \label{eq:x1a:3-eqm}\\
    \textbf{0} &=  \Big{(}-D^2+ \big{(} I - \Tilde{X}^2\big{)} B^2 \Big{)} \tilde{x}^2(t) \label{eq:x2a:3-eqm}\\
       \textbf{0} &=  \Big{(}-D^3+ \big{(} I - \Tilde{X}^3\big{)} B^3 \Big{)} \tilde{x}^3(t) \label{eq:x3a:3-eqm}
  \end{align}
Define $\kappa:= \max_{i \in [n]}[\bar{x}^1+ \bar{x}^2+ \bar{x}^3 ]_i/\tilde{x}^3_i$, which implies that $\bar{x}^1+ \bar{x}^2+ \bar{x}^3 \leq \kappa \tilde{x}^3$. Let $j$ be the maximizing value of $i$ in the definition of $\kappa$. Assume by way of contradiction that $\kappa \geq 1$. Consider \scriptsize
\begin{align} 
   (\bar{x}^1_j+ \bar{x}^2_j+ \bar{x}^3_j)&=(1-\bar{x}^1_j-\bar{x}^2_j-\bar{x}^3_j)((D^1)^{-1}B^1\bar{x}^1+(D^2)^{-1}B^2\bar{x}^2+ (D^3)^{-1}B^3\bar{x}^3)_j\nonumber\\
    &< (1-\bar{x}^1_j-\bar{x}^2_j-\bar{x}^3_j)((D^3)^{-1}B^3\bar{x}^1+(D^3)^{-1}B^3\bar{x}^2+ (D^3)^{-1}B^3\bar{x}^3)_j \label{ineq:11bis}\\
    &= (1-\bar{x}^1_j-\bar{x}^2_j-\bar{x}^3_j)((D^3)^{-1}B^3(\bar{x}^1+\bar{x}^2+ \bar{x}^3))_j \nonumber  \\
    &\leq (1-\bar{x}^1_j-\bar{x}^2_j-\bar{x}^3_j)((D^3)^{-1}B^3\kappa\tilde{x}^3)_j  \label{ineq:22bis} \\
    &=\kappa (1-\bar{x}^1_j-\bar{x}^2_j-\bar{x}^3_j)(1-\tilde{x}^3_j)^{-1}\tilde{x}^3_j \label{ineq:33bis} \\
    &=(1-\bar{x}^1_j-\bar{x}^2_j-\bar{x}^3_j)(1-\tilde{x}^3_j)^{-1}(\bar{x}^1_j+ \bar{x}^2_j+ \bar{x}^3_j) \label{ineq:44bis}\\
    &<(\bar{x}^1_j+ \bar{x}^2_j+ \bar{x}^3_j), \label{ineq:55bis}
\end{align}
\normalsize
where~\eqref{ineq:11bis} follows from the assumption that $(D^3)^{-1}B^3>(D^2)^{-1}B^2$ and \break $(D^3)^{-1}B^3>(D^1)^{-1}B^1$ leads to $(D^3)^{-1}B^3y>(D^2)^{-1}B^2y$ and  \break $(D^3)^{-1}B^3y>(D^1)^{-1}B^1y$ for any $y \gg \textbf{0}$. 
The inequality~\eqref{ineq:22bis} is due to $(\bar{x}^1+\bar{x}^2+ \bar{x}^3) \leq \kappa \tilde{x}^3$;~\eqref{ineq:33bis} is immediate from~\eqref{eq:x3a:3-eqm}; and~\eqref{ineq:44bis} stems from the fact that $\bar{x}^1_j+\bar{x}^2_j+\bar{x}^3_j=\kappa \tilde{x}^3_j$. The inequality~\eqref{ineq:55bis} follows by noting that since by assumption $\kappa \geq 1$, it must be that $(1-\tilde{x}^3_j)>(1-\bar{x}^1_j- \bar{x}^2_j- \bar{x}^3_j)$. Observe that~\eqref{ineq:55bis} is a contradiction, thus implying that $\kappa < 1$. Thus, $\bar{x}^1+ \bar{x}^2+ \bar{x}^3\ll \tilde{x}^3$. 
\par Note that since $\tilde{x}^3 \gg \textbf{0}$, and $\bar{x}^3 \gg \textbf{0}$, it follows from ~\eqref{eq:x3_3-equib} and~~\eqref{eq:x3a:3-eqm} that $\rho((I-\bar X^1-\bar X^2-\bar X^3)(D^3)^{-1} B^3)=1$, and $\rho(\big{(} I - \Tilde{X}^3\big{)}(D^3)^{-1} B^3)=1$. Since $\bar{x}^1+ \bar{x}^2+ \bar{x}^3\ll \tilde{x}^3$, it must be that $(I-\bar{X}^1-\bar{X}^2-\bar{X}^3)> (I-\tilde{X}^3)$, which from Lemma~\ref{lem:perron_frob} item iv) implies that $\rho((I-\bar{X}^1-\bar{X}^2-\bar{X}^3)(D^3)^{-1}B^3)> \rho((I-\tilde{X}^3)(D^3)^{-1}B^3)$, which is a contradiction of the assumption that there exists a 3-coexistence equilibrium. Consequently, there does not exist a 3-coexistence equilibrium.~\hfill \proofbox
\par 
Theorem~\ref{claim:virus3:strongest} identifies a sufficient condition for the non-existence of a 3-coexistence equilibrium. For the bivirus system, a  condition based on the row sum of matrices $B^k$, for $k \in [2]$, guarantees the non-existence of any 2-coexistence equilibria; see \cite[Corollary~3.10, item ii)]{ye2021convergence}. It is unknown if an analogous condition (but for the nonexistence of any 3-coexistence equilibria) can be obtained for the tri-virus system; a rigorous investigation is left for future work. 
It turns out that a condition less restrictive   than that in Theorem~\ref{claim:virus3:strongest} precludes the possibility of the existence of   certain 2-coexistence equilibria, as formalized in the following theorem.
\begin{thm}\label{claim:virus1weaker}
 Consider system~\eqref{eq:x1}-\eqref{eq:x3} under Assumptions~\ref{assum:base} and~\ref{assum:irreducible}. Suppose that  $\rho((D^k)^{-1}B^k)>1$ for $k \in [3]$.
 \begin{enumerate}[label=\roman*)]
\item   The equilibrium point $(\textbf{0}, \textbf{0}, \textbf{0})$ is unstable.
\item If $(D^2)^{-1}B^2> (D^1)^{-1}B^1$ and $(D^3)^{-1}B^3> (D^1)^{-1}B^1$, then
 \begin{enumerate}[label=\alph*)]
        \item the equilibrium point $(\tilde{x}^1, \textbf{0}, \textbf{0})$ is unstable; and 
        \item there does not exist a 2-coexistence equilibrium of the form a) $(\bar{x}^1, \bar{x}^2, \textbf{0})$, or b) $(\bar{x}^1,  \textbf{0}, \bar{x}^3)$, where $\textbf{0}\ll \bar{x}^1 \ll \textbf{1}$, $\textbf{0}\ll \bar{x}^2 \ll \textbf{1}$, and $\textbf{0}\ll \bar{x}^3 \ll \textbf{1}$.
 \end{enumerate}
 \end{enumerate}
\end{thm}
 \textit{Proof:} The proof of statement~i) (resp. statement~ii a)) is the same as that of statement~i) (resp. statement~ii)) in Theorem~\ref{claim:virus3:strongest}.
\par \textit{Proof of statement~ii) b):} We now show that, under the conditions of the Theorem,  a 2-coexistence equilibrium of the form $(\bar{x}^1, \bar{x}^2, \textbf{0})$, where $\textbf{0} \ll \bar{x}^1 \ll \textbf{1}$ and $\textbf{0} \ll \bar{x}^2 \ll \textbf{1}$,  cannot exist. Suppose that, by way of contradiction, there exists a 2-coexistence equilibrium of the form $(\bar{x}^1, \bar{x}^2, \textbf{0})$ with $\textbf{0} \ll \bar{x}^1 \ll \textbf{1}$ and $\textbf{0} \ll \bar{x}^2 \ll \textbf{1}$. By taking recourse to the equilibrium version of equations~\eqref{eq:x1}-\eqref{eq:x3}, we have the following:
\begin{align} 
   \textbf{0} &=  \Big{(} \big{(} I - (\bar{X}^1+\bar{X}^2) \big{)} B^1 - D^1 \Big{)} \bar{x}^1, \label{eq:x1a:eqm-1-bar}\\
    \textbf{0} &=  \Big{(} \big{(} I - (\bar{X}^1+\bar{X}^2) \big{)} B^2 - D^2 \Big{)} \bar{x}^2,\label{eq:x2a:eqm-2-bar} 
  \end{align}
By a suitable rearrangement of terms in~\eqref{eq:x1a:eqm-1-bar} and~\eqref{eq:x2a:eqm-2-bar}, and by noting that the matrices $I - (\bar{X}^1+\bar{X}^2) \big{)} B^1 $ and $I - (\bar{X}^1+\bar{X}^2) \big{)} B^2 $ are irreducible, from Lemma~\ref{lem:perron_frob} item i), we have the following:
\begin{align}
    \rho(I - (\bar{X}^1+\bar{X}^2) \big{)}(D^1)^{-1} B^1 )&=1 \label{rho1}\\
    \rho(I - (\bar{X}^1+\bar{X}^2) \big{)}(D^2)^{-1} B^2 )&=1 \label{rho2}
\end{align}

Since by assumption $(D^2)^{-1}B^2>(D^1)^{-1}B^1$, it must be that \break $(I - (\bar{X}^1+\bar{X}^2) \big{)}(D^2)^{-1} B^2> (I - (\bar{X}^1+\bar{X}^2) \big{)}(D^1)^{-1} B^1$, which from Lemma~\ref{lem:perron_frob} item iv) further implies that $\rho((I - (\bar{X}^1+\bar{X}^2) \big{)}(D^2)^{-1} B^2)> \rho((I - (\bar{X}^1+\bar{X}^2) \big{)}(D^1)^{-1} B^1)$, which contradicts  the conclusions from~\eqref{rho1} and~\eqref{rho2}. Thus, there cannot exist a 2-coexistence equilibrium  of the form  $(\bar{x}^1, \bar{x}^2, \textbf{0})$ with $\textbf{0} \ll \bar{x}^1 \ll \textbf{1}$ and $\textbf{0} \ll \bar{x}^2 \ll \textbf{1}$. The nonexistence of a 2-coexistence equilibrium of the form $(\bar{x}^1,  \textbf{0}, \bar{x}^3)$ with $\textbf{0} \ll \bar{x}^1 \ll \textbf{1}$ and $\textbf{0} \ll \bar{x}^3 \ll \textbf{1}$ can be shown analogously, by leveraging the assumption that $(D^3)^{-1}B^3> (D^1)^{-1}B^1$.\hfill \proofbox

Note that Theorem~\ref{claim:virus1weaker}  conclusively rules out the existence of a 2-coexistence equilibrium of the form $(\bar{x}^1, \bar{x}^2, \textbf{0})$, and of the form $(\bar{x}^1, \textbf{0}, \bar{x}^3)$. It, however, does not preclude the possibility of the existence of 2-coexistence equilibria of the form $(\textbf{0}, \bar{x}^2, \bar{x}^3)$. Therefore, Theorem~\ref{claim:virus1weaker}   does not set out a condition for the nonexistence of any 2-coexistence equilibrium in a tri-virus system.
\par Observe that the condition in Theorem~\ref{claim:virus3:strongest} implies (but is not implied by) the condition in Theorem~\ref{claim:virus1weaker}, and, consequently, Theorem~\ref{claim:virus3:strongest} also rules out the existence of a 2-coexistence equilibrium of the form $(\bar{x}^1, \bar{x}^2, \textbf{0})$, and of the form $(\bar{x}^1, \textbf{0}, \bar{x}^3)$.  

\section{Existence of coexistence 
equilibria}\label{sec:existence:coexistence}
Section~\ref{sec:nonexistence:coexistence:equilibria} focuses on the nonexistence of certain equilibria. In this section, we provide conditions for the existence of a 2-coexistence equilibrium and identify circumstances under which the existence of one or more 3-coexistence equilibria can be inferred. Define, for each $i \in [3]$, $\bar X^i: =\diag{(\bar{x}^i)}$, where $\bar x^i$ is as defined in Section~\ref{sec:prob:formulation}.
The following proposition guarantees the existence of a 2-coexistence equilibrium.
\begin{prop}\label{prop:2-coexistence}
Consider system~\eqref{eq:x1}-\eqref{eq:x3} under Assumptions~\ref{assum:base} and~\ref{assum:irreducible}. 
Suppose that  $\rho((D^k)^{-1}B^k)>1$ for $k \in [3]$. 
There exists at least one 2-coexistence equilibrium, i.e., $(\bar{x}^1, \bar{x}^2, \bar{x}^3)$ where 
$\textbf{0}\ll \bar{x}^i, \bar{x}^j \ll \textbf{1}$
 for some $i, j \in [3], i\neq j$, 
and, for $\ell \neq i, \ell \neq j$, $\bar{x}^\ell=\textbf{0}$. 
\end{prop}
\textit{Proof:} By assumption, we have that $\rho((D^k)^{-1}B^k)>1$ for $k \in [3]$. Therefore, from Proposition~\ref{prop:necessity} it follows that there exist single-virus endemic equilibria (also referred to as boundary equilibria) corresponding to virus~1,~2, and~3, namely $\tilde{x}^1$,  $\tilde{x}^2$, and $\tilde{x}^3$, respectively. Observe that each of these equilibria could be either stable or unstable, which implies that there must be either i) a pair which are both stable or ii) a pair that are both unstable. We consider the two cases separately.\\
\textit{Case i):} Assume, without loss of generality, that the two stable boundary equilibria are $(\tilde x^1,{\bf{0}},{\bf{0}})$ and $({\bf{0}},\tilde x^2,{\bf{0}})$. Observe that for the bivirus system, defined by neglecting $x^3$, the boundary equilibria are $(\tilde x^1,{\bf{0}})$ and $({\bf{0}},\tilde x^2)$, the existence of which, to reiterate, is guaranteed from the assumption that $\rho(B^k)>1$ for $k=1,2$. The bivirus system is known to be monotone \cite[Lemma~3.3]{ye2021convergence}. Therefore, from \cite[Theorem~2.8]{smith1988systems} and \cite[Proposition~2.9]{smith1988systems}, it follows that there is necessarily an unstable equilibrium $(\bar x^1,\bar x^2)$, where $\textbf{0} \ll (\bar x^1,\bar x^2) \ll \textbf{1}$.  Corresponding to this equilibrium of the bivirus system is the equilibrium $(\bar x^1,\bar x^2,{\bf{0}} )$ of the trivirus system.\\
\textit{Case ii):} Assume, without loss of generality, that the two unstable boundary equilibria are $(\tilde x^1,{\bf{0}},{\bf{0}})$ and $({\bf{0}},\tilde x^2,{\bf{0}})$. Since the bivirus system (defined by neglecting $x^3$) is monotone. Since for monotone systems between any two unstable boundary equilibria, there must  exist a stable equilibrium \cite{smith1988systems}, there must be an equilibrium point $(\bar x^1,\bar x^2)$. Moreover, since from \cite[Lemma~3.1]{ye2021convergence}, it is known that, for the bivirus system, besides the DFE, only the equilibrium points $(\tilde x^1,{\bf{0}})$ and $({\bf{0}},\tilde x^2)$ can be on the boundary of the surface defined by the set $\mathcal D^k$, $k=1,2$; it follows that the equilibrium point $(\bar x^1,\bar x^2)$ must lie in the interior of $\mathcal D^k$, $k=1,2$, which further implies that $\textbf{0} \ll (\bar x^1,\bar x^2) \ll \textbf{1}$,  with $(\bar x^1,\bar x^2,{\bf{0}} )$ again an equilibrium of the trivirus system.

Hence, it follows that in either case, there exists a 2-coexistence equilibrium 
for system~\eqref{eq:x1}-\eqref{eq:x3}.\hfill \proofbox
\vspace{2mm}
\par We now turn our attention to identifying a sufficient condition for the existence of a 3-coexistence equilibria. To this end, we will borrow some results on systems invariant in $\mathbb R^n_{+}$ with ultimately bounded trajectories, see \cite{hofbauer1990index}; a tri-virus system indeed meets these requirements.

\par Specifically, and following the terminology of \cite{hofbauer1990index}, we classify equilibria as saturated or unsaturated and derive a counting lower bound on the possible equilibria of tri-virus systems.
Observe that for a boundary equilibrium or a 2-coexistence equilibrium, the associated Jacobian is block-triangular; see~\eqref {jacob}. We say that an equilibrium is saturated if the diagonal block corresponding to the zero entries of the equilibrium is Hurwitz and unsaturated otherwise~\cite{hofbauer1990index}. 
A boundary equilibrium of~\eqref{eq:x1}-\eqref{eq:x3} is saturated if, and only if, said boundary equilibrium is locally exponentially stable; this follows immediately by noting the structure of the Jacobian matrix, evaluated at a boundary equilibrium, see \eqref{eq:boundaryJac}.
However, a 2-coexistence equilibrium might be unstable but still saturated because we only consider one diagonal block  of the Jacobian to determine if it is saturated, while stability is determined by the  remainder of the Jacobian truncated to remove this diagonal block; for a 2-coexistence equilibrium, this may or may not be Hurwitz. In contrast, a 3-coexistence equilibrium, with no diagonal block in the Jacobian matrix corresponding to zero entries of the equilibrium, is always saturated.

\par The index of an equilibrium, as per the convention in 
\cite{hofbauer1990index}, is defined as the sign, i.e., $\pm 1$, of the negative of the determinant of the Jacobian at the point. (Note that all such Jacobians are nonsingular for generic parameter values, as established in Proposition~\ref{prop:finite:equilibria}).  We will require the following results from \cite{hofbauer1990index}: (a) if a solution beginning in the interior converges to a limit as $t\to\infty$, that limit must be a saturated fixed point, and (b) the sum of the indices {\textit{over saturated equilibria}} is necessarily $+1$.  We will now show that there exists at least one saturated equilibrium, and if the only such saturated equilibrium  is a boundary equilibrium, it must be stable.

\begin{thm}\label{thm:hofbauer}
 Consider system~\eqref{eq:x1}-\eqref{eq:x3} under Assumption~\ref{assum:base}, and assume, in the light of genericity of the parameters, that all equilibria are nondegenerate. Suppose that, for each $k \in [3]$, $\rho((D^k)^{-1}B^k)>1$.
There exists at least one of the following:
\begin{enumerate}[label=\roman*)]
\item  a stable boundary equilibrium;
\item a saturated 2-coexistence equilibrium; or
\item a 3-coexistence equilibrium.
\end{enumerate}
\end{thm}

\textit{Proof:} We first show that the assumptions of the Theorem guarantee the existence of boundary equilibria and at least one 2-coexistence equilibrium.\\
By assumption, we have that $\rho((D^k)^{-1}B^k)>1$ for $k \in [3]$. Therefore, from Proposition~\ref{prop:necessity} it follows that there exist single-virus endemic equilibria (also  referred to as boundary equilibria) corresponding to virus~1,~2, and~3, namely $\tilde{x}^1$,  $\tilde{x}^2$, and $\tilde{x}^3$, respectively. Moreover, from Proposition~\ref{prop:2-coexistence}, the existence of the said boundary equilibria guarantees the existence of at least one 2-coexistence equilibrium.\\
Observe that Lemma~\ref{lem:pos} guarantees that, for each $k \in [3]$, $x^k(0)>\textbf{0}$ implies that $x^k(t)>\textbf{0}$ for all $t \in \mathbb{R}_{\geq 0}$, and that the set $\mathcal D$ (which is compact) is forward invariant. Therefore, from \cite[Theorem~2]{hofbauer1990index}, it follows that system~\eqref{eq:x1}-\eqref{eq:x3} has at least one saturated fixed point. There are two cases to consider.\\
Case 1: Suppose the aforementioned saturated fixed point is in the interior of $\mathcal D$. 
 Note that any fixed point in the interior of $\mathcal D$ is, due to Lemma~\ref{lem:equi_non-zero_nonone:1}, of the form $(\bar{x}^1, \bar{x}^2, \bar{x}^3)$, where $\textbf{0} \ll (\bar{x}^1, \bar{x}^2, \bar{x}^3) \ll \textbf{1}$, thus implying that $(\bar{x}^1, \bar{x}^2, \bar{x}^3)$ is a 3-coexistence equilibrium. \\
Case 2: Suppose there are no fixed points  in the interior of $\mathcal D$. This implies that a saturated fixed point must be on the boundary of $\mathcal D$  \cite{hofbauer1990index}. Therefore,  at least one of the following should be true: a) at least one of the single-virus boundary equilibrium is saturated; or b) at least one of the 2-coexistence equilibria 
is saturated.~\hfill\proofbox \vspace{2mm}
The following is an immediate consequence of Theorem~\ref{thm:hofbauer}, and \cite[Theorem~1]{hofbauer1990index}.
\begin{cor}\label{cor:hofbauer}
 Consider system~\eqref{eq:x1}-\eqref{eq:x3} under Assumption~\ref{assum:base}  and with generic parameter values ensuring that all equilibria are isolated. Suppose that, for each $k \in [3]$, $\rho((D^k)^{-1}B^k)>1$. Let $(\tilde{x}^1,\textbf{0}, \textbf{0})$,  $(\textbf{0}, \tilde{x}^2,\textbf{0})$ and  $(\textbf{0}, \textbf{0},\tilde{x}^3)$ denote the boundary equilibria corresponding to viruses~1,2 and 3, respectively. Suppose  further for such  equilibria that, for each pair $(i,j)$ such that $i, j \in [3]$ and $i \neq j$, $\rho((I{-}\tilde{X}^i)(D^j)^{-1}B^j)>1$.  Further, let $(\bar{x}^i, \bar{x}^j, \bar{x}^k)$  be a 2-coexistence equilibrium, i.e., such that for some $i, j \in [3]$, $\bar{x}^i, \bar{x}^j>0$, and for some $k \in [3]$, $\bar{x}^k=\textbf{0}$.
  Suppose further that, for every such triplet $(i,j,k)$, $s(-D^k + (I-\bar X^i-\bar X^j)B^k) \geq 0$. Then there exists at least one 3-coexistence equilibrium. 
In fact, there must be an odd number of 3-coexistence equilibrium.
\end{cor}
\textit{Proof:} By assumption, we have that $\rho((D^k)^{-1}B^k)>1$ for $k \in [3]$. Therefore, from Proposition~\ref{prop:necessity} it follows that there exist boundary equilibria corresponding to virus~1,~2, and~3, namely $\tilde{x}^1$,  $\tilde{x}^2$, and $\tilde{x}^3$, respectively. By assumption, for each pair $(i,j)$ such that $i, j \in [3]$ and $i \neq j$, $\rho((I{-}\tilde{X}^i)(D^j)^{-1}B^j)>1$. Therefore, from Theorem~\ref{thm:local}, each of the boundary equilibrium points, $(\tilde{x}^1,\textbf{0}, \textbf{0})$,  $(\textbf{0}, \tilde{x}^2,\textbf{0})$ and  $(\textbf{0}, \textbf{0},\tilde{x}^3)$,  is unstable,  and, thus, unsaturated. Given the existence of the boundary equilibria, from Proposition~\ref{prop:2-coexistence}, the existence of at least one equilibrium (i.e., 2-coexistence equilibrium), $(\bar{x}^i, \bar{x}^j, \bar{x}^k)$ such that for some $i, j \in [3]$, $\bar{x}^i, \bar{x}^j>0$, and for some $k \in [3]$, $\bar{x}^k=\textbf{0}$, is guaranteed. By assumption, for every such triplet $(i,j,k)$, $s(-D^k + (I-\bar X^i-\bar X^j)B^k) \geq 0$, which implies that no 2-coexistence equilibrium is saturated. Therefore, from Theorem~\ref{thm:hofbauer}, it must be that there exists an equilibrium point $(\bar{x}^i, \bar{x}^j, \bar{x}^k)$ such that   $\bar{x}^i, \bar{x}^j, \bar{x}^k>\textbf{0}$, i.e., a 3-coexistence equilibrium.\hfill \proofbox 



\section{Existence and attractivity of a line of coexistence  equilibria for nongeneric tri-virus networks}\label{sec:line:attractivity}
\par 
Sections~\ref{sec:persistence:one:virus}-\ref{sec:existence:coexistence} pertain to the existence or nonexistence of various endemic equilibria. Note that in those sections, the discussion centers around the (non)existence (and, in Section~\ref{sec:persistence:one:virus} also stability (or lack thereof)) of \emph{an} equilibrium point. In this section, however, we will identify a scenario that guarantees the existence and local exponential attractivity of a line of coexistence  equilibria. As for the bivirus case \cite{ye2021convergence}, special values of the system parameters are needed.

Throughout this section, we will assume that $D^k=I$ for $k=1,2,3$ on the one hand, recall, from the discussion in Section~\ref{sec:tri:virus:loss:of:generality}, that whether one can make such an assumption without loss of generality is an open question. On the other hand, said assumption makes it easier to describe the phenomenon as we will see in the sequel.
\par Let $z$ denote the single-virus endemic equilibrium corresponding to virus~1, with $Z=\diag(z)$. Therefore, 
since, by assumption, $D^1=I$, the vector $z$ fulfils the following:
\begin{equation}\label{eq:z}
    (-I+(I-Z)B^1)z=\textbf{0}.
\end{equation}
Furthermore, since $z$ is an endemic equilibrium, from \cite[Lemma~6]{axel2020TAC} 
it follows that $\textbf{0} \ll z \ll \textbf{1}$. Let $C$ be any nonnegative irreducible matrix for which $z$ is also an eigenvector corresponding to eigenvalue one. That is $Cz=z$. Therefore, from 
\cite[Theorem~2.7]{varga1999matrix}, 
it follows that $\rho(C)=1$, and that the vector $z$, up to a scaling, is the unique eigenvector of $C$ with all entries being strictly positive. Define
\begin{equation}\label{eq:B2}
    B^2:=(I-Z)^{-1}C.
\end{equation}
We have the following result.
\begin{thm}\label{thm:init:condns}
Consider system~\eqref{eq:x1}-\eqref{eq:x3} under Assumption~\ref{assum:base}. Suppose that $D^k=I$ for $k \in [3]$.
    Suppose that 
    $B^1$ and $B^3$ are arbitrary nonnegative irreducible matrices, and the
    vector $z$ and matrix $B^2$ are as defined in~\eqref{eq:z} and~\eqref{eq:B2}, respectively. Then, a set of equilibrium points of the trivirus equations is given by $(\beta_1z, (1-\beta_1)z, \textbf{0})$ for all $\beta_1 \in [0,1]$. Furthermore, \begin{enumerate} [label=\roman*)]
    \item if $s(-I+(I-Z)B^3)<0$, then 
the equilibrium set $(\beta_1z, (1-\beta_1)z, \textbf{0})$, with $\beta_1 \in [0,1]$, is locally exponentially attractive.
    \item if $s(-I+(I-Z)B^3)>0$, then 
the equilibrium set $(\beta_1z, (1-\beta_1)z, \textbf{0})$, with $\beta_1 \in [0,1]$, is unstable.
\end{enumerate}
\end{thm}

\textit{Proof:} We first show that, for all $\beta_1 \in [0,1]$, the  point $(\beta_1z, (1-\beta_1)z, \textbf{0})$ fulfils the equilibrium version of equations~\eqref{eq:x1}-\eqref{eq:x3}. To this end, observe that the right hand side of~\eqref{eq:x1}-\eqref{eq:x3} evaluated at  $(\beta_1z, (1-\beta_1)z, \textbf{0})$ yields:
\begin{align}
    &(-I+(I-\beta_1Z-(1-\beta_1)Z)B^1)\beta_1z \nonumber \\
    &=(-I+(I-Z)B^1)\beta_1z =0, \label{beta1z=0}
\end{align}
where~\eqref{beta1z=0} follows by noting that $\beta_1$ is a scalar, and  $z$ is the single-virus endemic equilibrium corresponding to virus~1. Similarly, 
\begin{align}
    &(-I+(I-\beta_1Z-(1-\beta_1)Z)B^2)(1-\beta_1)z \nonumber \\
    &=(-I+(I-Z)B^2)(1-\beta_1)z \nonumber \\
     &=(-I+(I-Z)(I-Z)^{-1}C)(1-\beta_1)z \nonumber \\
     &=(-I+C)(1-\beta_1)z=0, \label{beta2z=0}
\end{align}
where~\eqref{beta2z=0} follows by noting that $Cz=z$. Thus, from~\eqref{beta1z=0} and~\eqref{beta2z=0}, it is clear that, for every $\beta_1 \in [0,1]$,  $(\beta_1z, (1-\beta_1)z, \textbf{0})$ is an equilibrium point of system~\eqref{eq:x1}-\eqref{eq:x3}; i.e. there is a set of equilibrium points $(\beta_1z, (1-\beta_1)z, \textbf{0})$ with $\beta_1 \in [0,1]$. 
\par Next, observe that, for any $\beta_1 \in [0,1]$, the Jacobian evaluated at $(\beta_1z, (1-\beta_1)z, \textbf{0})$ is as given in~\eqref{jacobian:thm2}.
	{\noindent}
\begin{align} \label{jacobian:thm2}
&J(\beta_1z, (1{-}\beta_1)z, \textbf{0}) =
\scriptsize\begin{bmatrix}
-I+(I-Z)B^1-\diag(B^1\beta_1z) & -\diag(B^1\beta_1z)   & -\diag(B^1\beta_1z)  \\
-\diag(B^2(1-\beta_1)z) & -I+(I-Z)B^2-\diag(B^2(1-\beta_1)z)  & -\diag(B^2(1-\beta_1)z)\\
 \textbf{0}&  \textbf{0}& -I+(I-Z)B^3 \end{bmatrix}\normalsize
\end{align}
Hence, we can rewrite $J(\beta_1z, (1-\beta_1)z, \textbf{0})$ as 
\begin{align}\label{J:forhatx1:partitioned:1}
    J(\beta_1z, (1-\beta_1)z, \textbf{0})
    =
    & \scriptsize
    \begin{bmatrix} 
    \bar{J}(\beta_1z, (1{-}\beta_1)z, \textbf{0}) 
      && \hat{J} \\
      \mathbf{0} && {-}I{+}(I{-}Z)B^3
    \end{bmatrix}, 
\end{align}
\normalsize
where
\begin{align}\label{jacob:hatx1:22submatrix:1} 
&\bar{J}(\beta_1z, (1-\beta_1)z, \textbf{0})  \nonumber\\ 
& = \scriptsize
\begin{bmatrix} 
{-}I{+}(I{-}Z)B^1{-}\diag(B^1\beta_1z) & -\diag(B^1\beta_1z)   \\
-\diag(B^2(1-\beta_1)z)& {-}I{+}(I{-}Z)B^2{-}\diag(B^2(1{-}\beta_1)z)
 \end{bmatrix}, \end{align}
 \normalsize
while $\hat{J} =  \scriptsize \begin{bmatrix}-\diag(B^1\beta_1z) \\ -\diag(B^2(1-\beta_1)z)\end{bmatrix}$.
\par Note that the matrix $J(\beta_1z, (1-\beta_1)z, \textbf{0})$ is block upper triangular. Hence, it is clear that $s(J(\beta_1z, (1-\beta_1)z, \textbf{0})) =\max\{s(\bar{J}(\beta_1z, (1-\beta_1)z, \textbf{0}), s(-I+(I-Z)B^3)\}$. \\

\textit{Proof of statement~i):} Consider the matrix  $\bar{J}(\beta_1z, (1-\beta_1)z, \textbf{0})$.  Define $P: = \begin{bmatrix} I_n && \textbf{0} \\ \textbf{0} && -I_n \end{bmatrix}$. Therefore, 
\begin{align}
    &P\bar{J}((\beta_1z, (1-\beta_1)z, \textbf{0})P \nonumber \\
    & = \scriptsize
\begin{bmatrix} 
{-}I{+}(I{-}Z)B^1{-}\diag(B^1\beta_1z) & \diag(B^1\beta_1z)   \\
\diag(B^2(1-\beta_1)z) & {-}I{+}(I{-}Z)B^2{-}\diag(B^2(1{-}\beta_1)z)
\end{bmatrix}.
\end{align}

Note that $P\bar{J}((\beta_1z, (1-\beta_1)z, \textbf{0})P$ is an irreducible Metzler matrix.  Hence, by considering the element-wise positive vector $\begin{bmatrix}z^\top &&z^\top\end{bmatrix}^\top$, and by invoking the equilibrium version of the  equation of the single-virus system corresponding to virus~1, it follows that $P\bar{J}((\beta_1z, (1-\beta_1)z, \textbf{0})Pz=\textbf{0}$. Therefore, from \cite[Lemma~2.3]{varga1999matrix},
it follows that $s(P\bar{J}(\beta_1z, (1-\beta_1)z, \textbf{0})P)=0$, which implies $s(\bar{J}(\beta_1z, (1-\beta_1)z, \textbf{0})=0$. Consequently, since, by assumption, $s(-I+(I-Z)B^3)<0$, we obtain
$s(J(\beta_1z, (1-\beta_1)z, \textbf{0}))=0$.
Furthermore, since $s(P\bar{J}(\beta_1z, (1-\beta_1)z, \textbf{0})P)=0$, then since \break  $P\bar{J}(\beta_1z, (1-\beta_1)z, \textbf{0})P$ is  an irreducible Metzler matrix, 
by \cite[Theorem~2.7]{varga1999matrix}
we have that the matrix $\bar{J}(\beta_1z, (1-\beta_1)z, \textbf{0})$  has exactly one eigenvalue at the origin, and all other eigenvalues have negative real parts. Therefore, since $s(-I+(I-Z)B^3)<0$, the matrix $J(\beta_1z, (1-\beta_1)z, \textbf{0})$ has exactly one eigenvalue at the origin, and all other eigenvalues have negative real parts. Hence,  the trivirus equations associated with the line of equilibria define a one-dimensional center manifold along which the Jacobian is singular. Therefore, from \cite{khalil2002nonlinear} it follows that the set of equilibrium points $(\beta_1z, (1-\beta_1)z, \textbf{0})$, with $\beta_1 \in [0,1]$, is locally exponentially attractive.

\textit{Proof of statement~ii):} By assumption, $s(-I+(I-Z)B^3)>0$. Hence, \break  $s(J(\beta_1z, (1-\beta_1)z, \textbf{0}))>0$. Therefore, from \cite[Theorem~4.7, statement ii)]{khalil2002nonlinear} it follows that the set of equilibrium points $(\beta_1z, (1-\beta_1)z, \textbf{0})$, with $\beta_1 \in [0,1]$, is unstable.\hfill \proofbox

\vspace{2mm}
Observe that the key idea behind Theorem~\ref{thm:init:condns} is to fix the matrix $B^1$, and then choose $B^2$ (as given in~\eqref{eq:B2}) so as to obtain a locally exponentially attractive (resp. unstable) equilibrium set, namely $(\beta_1z, (1-\beta_1)z, \textbf{0})$, with $\beta_1 \in [0,1]$.
It is straightforward to note that,  using the other two possible pairs, namely $(B^1, B^3)$ and $(B^2, B^3)$, the same idea can be applied so as to secure corresponding locally exponentially attractive (or, in case of condition~ii) in Theorem~\ref{thm:init:condns} being, up to a suitable adjustment of notation, fulfilled, unstable)
equilibrium sets.

The 
result in Theorem~\ref{thm:init:condns} does not contradict the claim for finiteness of equilibria presented in Proposition~\ref{prop:finite:equilibria}.  
Note that the entries in matrices $D^i$, $B^i$, $i=1,2,3$, are either a priori fixed to a specific value or they are not. In  the case where those are not fixed to a specific value, we refer to those 
as free parameters, in the sense that these are allowed to take any value in $\mathbb{R}_{+}$. The dimension of the space of free parameters equals the number of free parameters in the tri-virus system.
Each  choice of free parameters 
results in a realization of system~\eqref{eq:x1}-\eqref{eq:x3}. The set of choices of free parameters that fall within the special case identified by Theorem~\ref{thm:init:condns} has measure zero.

\begin{rem}\label{rem:zero:init:condns}
Theorem~\ref{thm:init:condns} admits a non-zero 
initial infection level in the network for virus~3.  
If one were to assume that 
$x_i^3(0) =0$ for all $i \in [n]$, then the trivirus system of equations (i.e., \eqref{eq:x1}-\eqref{eq:x3})  collapses into the bivirus equation set (i.e., equation~\eqref{eq:full} where $m=2$). Under such a setting, Theorem~\ref{thm:init:condns} coincides with \cite[Proposition~3.9]{ye2021convergence}, which, in turn, subsumes \cite[Theorems~6 and~7]{liu2019analysis}; both with respect to i) admitting a larger class of parameters than \cite[Theorems~6 and~7]{liu2019analysis} (and, assuming the setup in \cite{pare2021multi} is restricted to the bi-virus case,  \cite[Corollaries~2 and~3]{pare2021multi}), and ii) providing guarantees for local exponential attractivity.
\end{rem}

\section{Plane of coexistence equilibria for nongeneric trivirus networks}\label{sec:global:plane}
The focus of Section~\ref{sec:line:attractivity} was on the existence and attractivity (resp. instability) of a line of coexistence equilibria. In this section, we identify two  scenarios that permit the existence of a plane of  coexistence equilibria  (more strictly, the equilibria of interest are defined by the intersection of a plane with the positive orthant),  and, in the case of one of those scenarios, 
we provide conditions for global convergence to the associated plane of  coexistence equilibria.


We consider a case where three identical copies of a virus are spreading over the same graph as formalized next.

\begin{assm}\label{assm:hetero:samegraph}
We suppose that
\begin{enumerate}[label=(\roman*)]
    \item All three viruses are spreading over the same graph.
    \item For all $i \in [n]$ $\delta_i^1 =\delta_i^2 =\delta_i^3>0$.
    \item For all $i=j \in [n]$ and $(i,j) \in \mathcal E$, $\beta_{ij}^1=\beta_{ij}^2=\beta_{ij}^3$.
\end{enumerate}
\end{assm}

Note that for the  special case identified in Assumption~\ref{assm:hetero:samegraph}, assuming that the setting in \cite{pare2021multi} is restricted to the tri-virus case, the existence of a plane of coexistence equilibrium has been secured by \cite[Corollary~3]{pare2021multi}. However, \cite[Corollary~3]{pare2021multi} does not provide guarantees for even local (let alone global)  convergence to the said plane. To address this shortcoming, first, consider the system

\begin{equation}\label{eq:positive_linear_x}
\dot x(t) = (-D+(I-\diag(\tilde x))B)x(t),
\end{equation}
where $\tilde x$ is the unique endemic equilibrium of the single virus SIS system associated with $(D, B)$, and with $B$ irreducible. The matrix $Q: = D-(I-\diag(\tilde x))B$ is a singular irreducible $M$-matrix, with a simple eigenvalue at $0$, and all other eigenvalues have a positive real part. Associated with this simple zero eigenvalue is the right eigenvector $\tilde x \gg {\bf 0}$ and a left eigenvector $\tilde u^\top \gg {\bf 0}^\top$. We assume that $\tilde u^\top$ is normalized to satisfy $\tilde u^\top \tilde x = 1$. Moreover, there exists a positive diagonal matrix $P$ such that $\bar Q := PQ+Q^\top P \succeq 0$. In fact, if we choose $P = \diag(\tilde u_1/\tilde x_1, \hdots, \tilde u_n/\tilde x_n)$, then it can be verified that $\bar Q$ is an irreducible $M$-matrix of rank $n-1$, with the nullvector $\tilde x$ associated with the simple eigenvalue at $0$, see~\cite[Section 4.3.4]{qu2009cooperative}.

\begin{thm}\label{thm:global:plane}
Consider system~\eqref{eq:x1}-\eqref{eq:x3} under Assumptions~\ref{assum:base}, \ref{assum:irreducible} and~\ref{assm:hetero:samegraph}. Further, suppose that $\rho(D^{-1}B)> 1$. 
Then
\begin{enumerate}[label=\roman*)]
    \item For all initial conditions satisfying $x^1(0) >  {\bf 0}_n$, $x^2(0) > {\bf 0}_n$, and $x^3(0) > {\bf 0}_n$, we have that $\lim_{t\to\infty} (x^1(t),x^2(t),x^3(t)) \in \mathcal{E}$  with exponentially fast convergence rate, where $$\mathcal {E} = \{(x^1, x^2, x^3) | \alpha_1 x^1+\alpha_2 x^2+\alpha_3 x^3 = \tilde x, \textstyle\sum_{i=1}^3 \alpha_i = 1\},$$ and $\tilde x$ is the unique endemic equilibrium of the single virus SIS dynamics defined by $(D,B)$.
    \item Every point on the connected set $\mathcal{E}$ is a coexistence equilibrium.
\end{enumerate}
\end{thm}
\textit{Proof:}
\textit{Proof of statement i)}: 
The initial conditions guarantee that, at some finite time $t$, we have $x^1(t) \gg  {\bf 0}_n$, $x^2(t) \gg {\bf 0}_n$, and $x^3(t)\gg {\bf 0}_n$. Define $z = x^1 + x^2 + x^3$, and $Z=X^1+X^2+X^3$.
Therefore, together with Assumption~\ref{assm:hetero:samegraph}, it follows that
\begin{align}
    \dot{z} & = \Big[-D+\big(I-(X^1(t)+X^2(t)+X^3(t))\big)B\Big] \times \nonumber \\
    &~~~~~~~~~~~~~~~~~~~~~~~~~(x^1(t)+x^2(t)+x^3(t)) \\
    & = \big[-D+(I-Z(t))B\big]z(t).\label{eq:single_virus}
\end{align}
Since $z(t) \gg {\bf 0}_n$, and $\rho(D^{-1}B) > 1$ by hypothesis, it follows from \cite[Theorem~2]{khanafer2016stability}
that $\lim_{t\to\infty} z(t) = \tilde x$ exponentially fast. In fact, $\tilde x$ is the exponentially stable equilibrium of \eqref{eq:single_virus}, with domain of attraction $z(0) \in [0,1]^n\setminus {\bf 0}$.  It follows that, for all $z(0) \in [0,1]^n\setminus {\bf 0}$, $\Vert \tilde x - z(t) \Vert \leq ae^{-bt}$ for some positive constants $a, b$, with $\Vert \cdot \Vert$ being the Euclidean norm.

The dynamics for virus~$i$, where $i\in [3]$, can be written as $$\dot{x}^i(t) = -[D+(I-\diag(\tilde x))B]x^i(t) + (\diag(\tilde x) - Z(t))Bx^i(t).$$
Without loss of generality, we consider virus~$1$ and drop the superscript. That is, we study the system
\begin{equation}\label{eq:virus_system_plane}
    \dot{x}(t) = -Qx(t) + (\diag(\tilde x) - Z(t))Bx(t),
\end{equation}
where $Z(t)$ is treated as an external time-varying input, and $Q = D-(I-\diag(\tilde x))B$ is an irreducible singular $M$-matrix, as detailed below \eqref{eq:positive_linear_x}. Define the oblique projection matrix $R = I - \tilde x \tilde u^\top$. Define also $\zeta = Rx$, and note that $\tilde u^\top \zeta = 0$ and $\zeta = {\bf 0} \Leftrightarrow x = \alpha \tilde x$ for some $\alpha \in \mathbb R$. That is, $\zeta$ is always orthogonal to $\tilde u$, and $\zeta$ is the zero vector precisely when $x$ is in the span of $\tilde x$.

Observe that $\dot \zeta(t) = R\dot{x}(t)$. Substituting in the right of \eqref{eq:virus_system_plane} for $\dot{x}(t)$, we obtain
\begin{equation}\label{eq:virus_error_system}
    \dot{\zeta}(t) = -Q\zeta(t) + R(\diag(\tilde x)-Z(t))Bx(t),
\end{equation}
by exploiting the fact that $QR = Q = RQ$. 

Consider the Lyapunov-like function
\begin{equation}\label{eq:V}
    V = \zeta(t)^\top P \zeta(t),
\end{equation}
with $P$ defined below \eqref{eq:positive_linear_x}. It is positive definite in $\zeta$. Differentiating $V$ with respect to time 
yields  \footnotesize
\begin{align}\label{eq:V_dot}
    \dot{V} & =  -2\zeta(t)^\top PQ\zeta(t) +2\zeta(t)^\top PR(\diag(\tilde x)-Z(t))Bx(t) \nonumber \\
    & = -\zeta(t)^\top \bar Q \zeta(t) +2\zeta(t)^\top PR(\diag(\tilde x)-Z(t))Bx(t).
\end{align}
\normalsize 
Due to submultiplicativity of matrix norms, $\Vert \zeta(t)\Vert \leq \Vert R \Vert \Vert x(t)\Vert$. Hence, since $\Vert x(t)\Vert$ is bounded, \eqref{eq:V_dot} yields  \footnotesize
\begin{align}
    \dot{V} & \leq -\zeta(t)^\top \bar Q \zeta(t) + \kappa \Vert \tilde x - z(t) \Vert 
    \leq -\zeta(t)^\top \bar Q \zeta(t) + \bar a e^{-b t}, \label{eq:V_dot_ineq_1}
\end{align}
\normalsize 
where $\kappa$ and $\bar a$ are positive constants, and the second inequality is since $\Vert \tilde x - z(t) \Vert \leq ae^{-bt}$.

Let $\lambda_2$ denote the 
smallest strictly positive eigenvalue of $\bar Q$. The Courant-Fischer min-max theorem~\cite[Theorem~8.9]{zhang2011matrix} yields 
\begin{equation}\label{eq:cf_ineq}
    \frac{\zeta^\top \bar Q \zeta}{\zeta^\top \zeta} \geq \min_{\substack{v^\top \tilde u =  0\\v^\top v = 1}} v^\top \bar Q v = \lambda_2,
\end{equation}
for all $\zeta \neq {\bf 0}$ perpendicular to $\tilde u$. Since $\zeta^\top \tilde u = 0$ holds for all $\zeta \in \mathbb R$ by definition, \eqref{eq:cf_ineq} holds for all $\zeta\neq {\bf 0}$. Next, and recalling the definition of $P$ below \eqref{eq:positive_linear_x}, it follows that
\begin{equation}\label{key:ineq:1:lyap}
    \underline{p}\zeta^\top \zeta \leq \zeta^\top P \zeta \leq \bar p \zeta^\top\zeta,
\end{equation}
for all $\zeta$, where $\underline{p} = \min_{i\in [n]} \tilde u_i/\tilde x_i$ and $\bar p = \max_{i\in[n]} \tilde u_i/\tilde x_i$. 

From~\eqref{eq:V_dot_ineq_1} it follows that $\dot{V} \leq -\lambda_2 \zeta(t)^\top \zeta + \bar a e^{-bt}$, which further implies that $ \dot{V} \leq -\bar \lambda \bar p\zeta(t)^\top \zeta + \bar a e^{-bt}$, where $\bar \lambda = \lambda_2/\bar p$. Consequently, from~\eqref{key:ineq:1:lyap}, it follows that $ \dot{V} \leq -\bar\lambda V + \bar a e^{-bt}$. Thus, and recalling the definition of $\zeta$, it follows that $\lim_{t\to\infty} x(t) = \alpha \tilde x$ exponentially fast, where $\alpha \in (0,1)$ because $x(t) \in (0,1)^n$ for all $t \geq \tau$, for some positive $\tau$.


This analysis holds not only for virus~$1$ but also virus~$2$ and virus~$3$. In other words, $\lim_{t\to\infty} x^i(t) = \alpha_i \tilde x$ for some $\alpha_i \in (0,1)$, for all $i\in [3]$. Recall that $\lim_{t\to\infty} z(t) = \tilde x$, and we immediately conclude that $\sum_{i=1}^n \alpha_i = 1$, thus proving statement i). 

\textit{Proof of statement ii)}: We now prove that every point in $\mathcal{E}$ is an equilibrium (coexistence follows trivially by definition). Consider an arbitrary point $(x^1, x^2, x^3)$ in $\mathcal{E}$. From \eqref{eq:x1}, we have \vspace{-4mm}
\begin{align}
    \dot{x}^1(t) & = (-D+(I-X^1-X^2-X^3)B)x^1 \\
    & = (-D+(I-\diag(\tilde x))B)\alpha_1 \tilde x = {\bf 0}.
\end{align}
By the same arguments, it follows that $\dot{x}^2(t) = {\bf 0}$ and $\dot{x}^3(t) = {\bf 0}$ at the point $(x^1, x^2, x^3)$ in $\mathcal{E}$. In other words, $(x^1, x^2, x^3)$ is an equilibrium of the system system~\eqref{eq:x1}-\eqref{eq:x3}. Since this holds for any arbitrary point in $\mathcal{E}$, the proof of statement~ii) is complete. \hfill \proofbox \\

\vspace{1mm}
We identify yet another special case of the tri-virus model parameters
that admits the existence of a plane of 3-coexistence equilibrium. To this end, let $B^1$ be some arbitrary nonnegative irreducible matrix with $z$ being as defined in~\eqref{eq:z}, and $B^2$ being as defined in~\eqref{eq:B2}. Let $\hat{C} \neq C$ be a nonnegative irreducible matrix such that $z$ is the eigenvector corresponding to eigenvalue unity, that is, $\hat{C}z=z$. Define \begin{equation}\label{eq:B3}
    B^3:=(I-Z)^{-1}\hat{C}.
\end{equation}
We have the following result.

\begin{prop} \label{prop:plane:equilibrium}
Consider system~\eqref{eq:x1}-\eqref{eq:x3} under Assumptions~\ref{assum:base} and~\ref{assum:irreducible}. 
Suppose that 
    $B^1$ is an arbitrary nonnegative irreducible matrix 
 and the 
    vector $z$ together with matrices $B^2$ and $B^3$ are as defined in~\eqref{eq:z}, and~\eqref{eq:B2} and~\eqref{eq:B3}, respectively. Then, a set of equilibrium points of the trivirus equations is given by $(\alpha_1z, \alpha_2z, \alpha_3z)$,  for all $\alpha_i \geq 0$, 
 $i=1,2,3$, and  $\sum_{i=1}^3\alpha_i=~1$. 
\end{prop}
\textit{Proof:}
Assume, without loss of generality, that  $D^k=I$ for $k \in [3]$.
Observe, then,  that the right hand side of~\eqref{eq:x1}-\eqref{eq:x3} evaluated at  $(\alpha_1z, \alpha_2z, \alpha_3z)$ yields:
\begin{align}
    &(-I+(I-\alpha_1Z-\alpha_2Z- \alpha_3Z)B^1)\alpha_1z \nonumber \\
    &=(-I+(I-Z)B^1)\alpha_1z =0, \label{alpha1z=0}
\end{align}
where~\eqref{alpha1z=0} follows by noting that by assumption, $\alpha_1$ is a nonnegative scalar, and $\sum_{i=1}^3\alpha_i=1$,  and  $z$ is the single-virus endemic equilibrium corresponding to virus~1. Similarly, 
\begin{align}
    &(-I+(I-\alpha_1Z-\alpha_2Z- \alpha_3Z)B^2)\alpha_2z \nonumber \\
    &=(-I+(I-Z)B^2)\alpha_2z \nonumber \\
     &=(-I+(I-Z)(I-Z)^{-1}C)\alpha_2z \nonumber \\
     &=(-I+C)\alpha_2z=0, \label{alpha2z=0}
\end{align}
where~\eqref{alpha2z=0} follows by noting that $Cz=z$.
Finally, 
\begin{align}
    &(-I+(I-\alpha_1Z-\alpha_2Z- \alpha_3Z)B^3)\alpha_3z \nonumber \\
    &=(-I+(I-Z)B^3)\alpha_3z \nonumber \\
     &=(-I+(I-Z)(I-Z)^{-1}\hat C)\alpha_3z \nonumber \\
     &=(-I+\hat C)\alpha_3z=0, \label{alpha3z=0}
\end{align}
where~\eqref{alpha3z=0} follows by noting that $\hat{C}z=z$.
Thus, from~\eqref{alpha1z=0}, \eqref{alpha2z=0} and~\eqref{alpha3z=0}, it is clear that, for every $\alpha_i\geq 0$ with $i=1,2,3$, such that $\sum_{i=1}^3 \alpha_i=1$,  $(\alpha_1z, \alpha_2z, \alpha_3z)$ is an equilibrium point of system~\eqref{eq:x1}-\eqref{eq:x3}.\hfill \proofbox 
\vspace{2mm}
\par Note that Proposition~\ref{prop:plane:equilibrium} and, assuming the setting in \cite{pare2021multi} is limited to the tri-virus case, \cite[Corollaries~2 and ~3]{pare2021multi} identify special cases that admit the existence of a plane of coexistence equilibria. Since  Proposition~\ref{prop:plane:equilibrium}  covers a larger class of parameters than that in \cite[Corollaries~2 and ~3]{pare2021multi}, Proposition~\ref{prop:plane:equilibrium} subsumes \cite[Corollaries~2 and ~3]{pare2021multi}.

Proposition~\ref{prop:plane:equilibrium} and Theorem~\ref{thm:global:plane} compare in the following sense: 
On the one hand, note that if $(I-Z)B^1=C=\hat{C}$, then, in view of~\eqref{eq:B2} and~\eqref{eq:B3}, it follows that $B^1=B^2=B^3$, which coincides with the setting of Theorem~\ref{thm:global:plane}. Thus,   
Proposition~\ref{prop:plane:equilibrium} covers a larger class of parameters than Theorem~\ref{thm:global:plane}, and, therefore, with respect to the existence of a set of 3-coexistence equilibria, is more general.  Conversely,  Proposition~\ref{prop:plane:equilibrium} provides no stability guarantees for the said equilibrium set, whereas Theorem~\ref{thm:global:plane} establishes global stability of the  equilibrium set for the smaller class of parameters that it covers. 

\par With respect to the existence of a continuum of equilibria, Proposition~\ref{prop:plane:equilibrium} is stronger than Theorem~\ref{thm:init:condns} because it adds a larger class of $i=1,2,3$, $D^i$, $B^i$ matrices. That said, note that Theorem~\ref{thm:init:condns} admits  arbitrary $B^3$, whereas Proposition~\ref{prop:plane:equilibrium} insists on $B^3$ obeying a specific functional form, as given in~\eqref{eq:B3}. In conclusion, neither subsumes the other.

\section{Simulations}\label{sec:simulations}

We now present a set of simulations that highlight the key theoretical results of our paper. We choose $D^i = I$ for $i = 1,2,3$, with the following $B^i$ matrices, where $\hat{\beta}^k_{ij}$ are constants that are changed depending on the simulation example being presented. 
\[   B^1 = \scriptsize \begin{bmatrix}
    0 & 0 & 0 & 1.5\\1.5 & 0 & 0 & 0\\0 & 1.5 & 0 & 0\\0 &0 &1.5& 0
    \end{bmatrix},\\
    B^2 = \begin{bmatrix}
    0 & 1.5+\hat{\beta}^2_{12} & 0 & 0\\0 & 0 & 1.5 & 0\\0 & 0 & 0 & 1.5\\1.5& 0& 0& 0
    \end{bmatrix},
\]
\begin{equation}
    B^3 =  \scriptsize \begin{bmatrix}
    1 & 0 & 0.5+\hat{\beta}^{3}_{13} & 0\\0 & 1+\hat{\beta}^{3}_{22} & 0.5 & 0\\0 & 0.5 & 0 & 1\\0.3+\hat{\beta}^{3}_{31}& 0 &1.2& 0
    \end{bmatrix}.\nonumber
\end{equation}
Except as otherwise stated, initial conditions are obtained according to the following procedure. First, for $i\in [n]$ and $s\in[4]$, we sample a value $p_{i}^s$ from a uniform distribution $(0,1)$. Then, for $i\in [n]$ and $k\in[3]$ we set $x_{i}^k(0) = p_{ij}^k/\sum_{s=1}^4 p_{ij}^s$, which ensures that the initial conditions are in $\mathcal D$ but otherwise randomized. 

\textit{Example~1:} We set $\hat{\beta}^3_{13} = 0$,  $\hat{\beta}^2_{12}=-0.1$, $\hat{\beta}^3_{22} = -0.1$, and $\hat{\beta}^{3}_{31} = 0.1$. In this example, and following the notation of Section~\ref{ssec:equilibria_types} and Theorem~\ref{thm:local}, we obtain $\rho((I-\tilde X^1)(D^2)^{-1}B^2) = 0.9829$ and $\rho((I-\tilde X^1)(D^3)^{-1}B^3) = 0.99624$. Thus, $(\tilde x^1, \bf{0}, \bf{0})$ is locally exponentially stable, and Fig.~\ref{fig:tri_virus_stable_virus1} shows convergence to $(\tilde x^1, \bf{0}, \bf{0})$. It can be computed that $\rho((I-\tilde X^2)(D^1)^{-1}B^1) = 1.0174$ and $\rho((I-\tilde X^2)(D^3)^{-1}B^3) = 1.0127$ and similarly, $\rho((I-\tilde X^3)(D^1)^{-1}B^1) = 1.003$ and $\rho((I-\tilde X^3)(D^2)^{-1}B^2) = 0.9863$. Hence, in line with Theorem~2, both $({\bf 0},\tilde x^2, {\bf 0})$ and $({\bf 0}, {\bf 0}, \tilde x^3)$ are unstable.



\textit{Example~2:} We set $\hat{\beta}^3_{13} = 0.05$,  $\hat{\beta}^2_{12}=\hat{\beta}^3_{22} = \hat{\beta}^{3}_{31} = 0$. In this example, and following the notation of Theorem~\ref{thm:init:condns}, we have a line of equilibria $(\beta_1 z, (1-\beta_1)z, {\bf 0})$, with $z = \frac{1}{3}\bf{1}$ and $\beta_1\in [0,1]$. Since $\rho((I-Z)B^3) = 1.0043$, in line with statement~ii) in Theorem~\ref{thm:init:condns}, this line of equilibria is unstable; see Figure~\ref{fig:tri_virus_unstable_2line}. We can compute that $({\bf 0}, {\bf 0}, \tilde x^3)$ is locally exponentially stable, following reasoning analogous to that in Theorem~\ref{thm:local}. 
Note, however, that even though the line of equilibria $(\beta_1 z, (1-\beta_1)z, {\bf 0})$ is unstable, it still has an attractive manifold (of measure zero in $\mathcal{D}$); trajectories on this manifold will converge to the line. As established in \cite[Proposition~3.9]{ye2021convergence}, this manifold includes points $(x^1, x^2, {\bf 0})$, where $x^1$ and $x^2$ are in a small open neighbourhood around $\beta_1 z$ and $(1-\beta_1)z$, respectively. 

\textit{Example~3:} We set $\hat{\beta}^3_{13} = -0.1$,  $\hat{\beta}^2_{12}=\hat{\beta}^3_{22} = \hat{\beta}^{3}_{31} = 0$. Similar to Example~2, we have a line of equilibria $(\beta_1 z, (1-\beta_1)z, {\bf 0})$, with $z = \frac{1}{3}\bf{1}$. Now, however, $\rho((I-Z)B^3) = 0.9911$, and this line of equilibria is locally exponentially attractive according to statement i) in Theorem~\ref{thm:init:condns}; see Figure~\ref{fig:tri_virus_stable_2line}. The boundary equilibrium $({\bf 0}, {\bf 0}, \tilde x^3)$ is unstable, following reasoning analogous to that in Theorem~\ref{thm:local}. 

\textit{Example 4:} We set $\hat{\beta}^3_{13} = \hat{\beta}^2_{12}=\hat{\beta}^3_{22} = \hat{\beta}^{3}_{31} = 0$. The aforementioned choice 
of $B^i$ satisfies Proposition~\ref{prop:plane:equilibrium}, and hence there is a plane of equilibria of the form $(\alpha_1 z, \alpha_2 z, \alpha_3 z)$, with $\sum_{i=1}^3 \alpha_i = 1$ and $z = \frac{1}{3}\bf{1}$. Although we do not have a theoretical result for convergence to this plane, Fig.~\ref{fig:tri_virus_stable_3plane_a} and \ref{fig:tri_virus_stable_3plane_b} illustrate convergence to this plane of equilibria from different initial conditions; the particular equilibrium reached depends on the choice of initial conditions. 

We now consider a second set of simulations with $n = 5$ nodes. We again set $D^i = I$ for $i = 1, 2, 3$, but we now study a different set of $B^i$ matrices to demonstrate other relevant theoretical results of the paper. Let \scriptsize
\begin{equation}
    B = \begin{bmatrix}
    1 & 0 & 2 & 0 & 0.5\\0.5 & 2 & 0 & 0 & 0\\ 0 & 5 & 0.1 & 0 & 0 \\ 0.1 & 0 & 0 & 0.2 & 0\\0 & 0 & 0 & 0.1 & 0.9
    \end{bmatrix}.
\end{equation}
\normalsize 
Our $B^i$ matrices will be perturbations of $B$, which we describe in the respective simulation examples. To facilitate this, let $e_i$ be the $i$th basis vector of $\mathbb R^5$, i.e., $e_i$ has $1$ in its $i$th entry, and all other entries are $0$. The initial conditions are selected following the same random procedure described above.

\textit{Example 5:} We set $B^i = B$ for all $i=1,2,3$, i.e., we have three identical copies of a virus spreading over the same graph. According to Theorem~\ref{thm:global:plane}, there is a plane of coexistence equilibria, which in this case is defined by the vector $\tilde x = [0.691,0.610,0.758,0.078,0.051]^\top$. Fig.~\ref{fig:trivirus_stable_copyplane_a} and \ref{fig:trivirus_stable_copyplane_b} show convergence to two different equilibria in this plane of equilibria from two different initial conditions. 

\textit{Example 6:} We set $B^1 = B$, $B^2 = B + 0.5 e_1e_4^\top$, and $B^3 = B^2 + 0.1 e_5 e_1^\top$. 
This tri-virus system satisfies the inequality conditions of Theorem~\ref{claim:virus3:strongest}; it follows that there does not exist a $3$-coexistence equilibrium, while the boundary equilibrium $({\bf 0}, {\bf 0}, \tilde x^3)$ is locally exponentially stable. Fig.~\ref{fig:tri_virus_stable_virus3} shows the tri-virus system converging to the $({\bf 0}, {\bf 0}, \tilde x^3)$ equilibrium. 

\textit{Example 7:}  We set $B^1 = B$, $B^2 = B + 2 e_1e_4^\top$, and $B^3 = B + 0.1 e_5 e_1^\top$. This system satisfies the inequality conditions of Theorem~\ref{claim:virus1weaker} but not those of Theorem~\ref{claim:virus3:strongest}. In Fig.~\ref{fig:tri_virus_stable_2virus}, we observe convergence to a $2$-coexistence equilibrium. Consistent with Theorem~\ref{claim:virus1weaker}, this $2$-coexistence equilibrium is such that only virus $2$ and virus $3$ exist in each node of the network, i.e., it has the form $({\bf 0}, \bar x^2, \bar x^3)$.

\textit{Example 8:}  We set $B^1 = B + 0.7e_3 e_2^\top$, $B^2 = B + 2 e_1e_4^\top$, and $B^3 = B + 0.1 e_5 e_1^\top$. This system satisfies the conditions of Corollary~\ref{cor:hofbauer}. The inequalities concerning the boundary equilibria are easily checked. The inequalities involving the $2$-coexistence equilibria are checked by verifying that there are precisely $3$ such equilibria, one unique equilibrium corresponding to each of virus~$1$, virus~$2$, and virus~$3$ being extinct. This uniqueness can be established by a rapid simulation procedure that exploits the fact that bivirus systems are MDS and involves simulation of just two different initial conditions~\cite[Corollary~3.14]{ye2021convergence}. Based on this knowledge, Corollary~\ref{cor:hofbauer} establishes that there is at least one $3$-coexistence equilibrium, and it is saturated. Indeed, Fig.~\ref{fig:tri_virus_stable_3virus} shows convergence to a $3$-coexistence equilibrium.

\textit{Example 9:} In each of the examples above, except for Example 5, the simulations correspond to theoretical results that establish local exponential convergence (either to an isolated equilibrium or a line of equilibria). Extensive additional simulations with many other randomized initial conditions suggest that, for these examples, convergence to the equilibrium of interest (or line/plane of equilibria) occurs for all initial conditions in the interior of $\mathcal{D}$, thereby suggesting that they (i.e., the equilibrium of interest, or line/plane of equilibria), are globally stable. We conclude here with a final example, which serves as a cautionary note highlighting that tri-virus systems can have multiple attractive equilibria with different regions of attraction with nonzero measure in $\mathcal{D}$. First, define the following matrices:
\begin{align*}
    B_{11} & = \begin{bmatrix} 1.6 & 1\\1 & 1.6 \end{bmatrix}, 
    B_{12} = \begin{bmatrix} 2.1 & 0.156\\3.0659 & 1.1 \end{bmatrix}, \\
    B_{21} & = \begin{bmatrix} 1.7 & 1\\ 1.2 & 0.5 \end{bmatrix}, 
    B_{22} = \begin{bmatrix} 1.6 & 1 \\1.2 & 0 \end{bmatrix}, 
    C = 0.001 \times{\bf 1 1}^\top.
\end{align*}
We consider $n = 4$, and set $D^i = I$ for all $i=1,2,3$, and
\begin{equation*}
    B^1 = \begin{bmatrix} B_{11} & C \\ C & B_{21} \end{bmatrix}, B^2 = \begin{bmatrix} B_{12} & C \\ C & B_{22} \end{bmatrix}, 
    B^3 = \begin{bmatrix} B_{22} & C \\ C & B_{11} \end{bmatrix}.
\end{equation*}
Fig.~\ref{fig:tri_virus_stable_multiequib_virus13} and \ref{fig:tri_virus_stable_multiequib_virus23} show convergence to two different $2$-coexistence equilibria from different initial conditions in the interior of $\mathcal{D}$; we verified that the Jacobian matrix at both these equilibria are Hurwitz, i.e., both equilibria are locally exponentially stable. In Fig.~\ref{fig:tri_virus_stable_multiequib_virus13}, convergence occurs to a $2$-coexistence equilibrium where virus $2$ is extinct, whereas in Fig.~\ref{fig:tri_virus_stable_multiequib_virus23}, the $2$-coexistence equilibrium is such that virus~$1$ is extinct.

\begin{figure*}
\begin{minipage}{0.49\linewidth}
\centering
\subfloat[Example~1]{\includegraphics[width=\columnwidth]{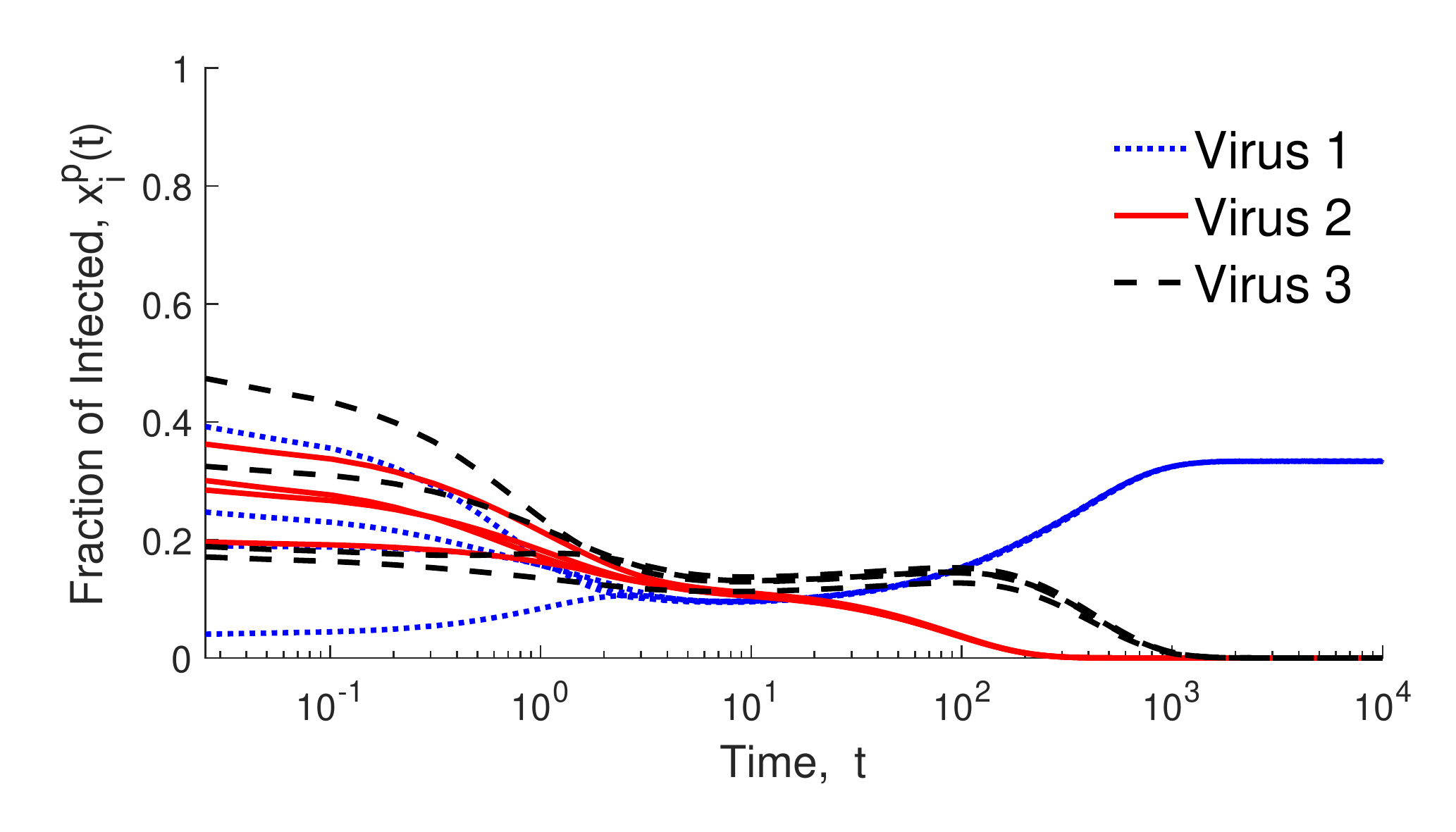}\label{fig:tri_virus_stable_virus1}}
\end{minipage}
\hfill
\begin{minipage}{0.49\linewidth}
\centering
\subfloat[Example~2]{\includegraphics[width=\columnwidth]{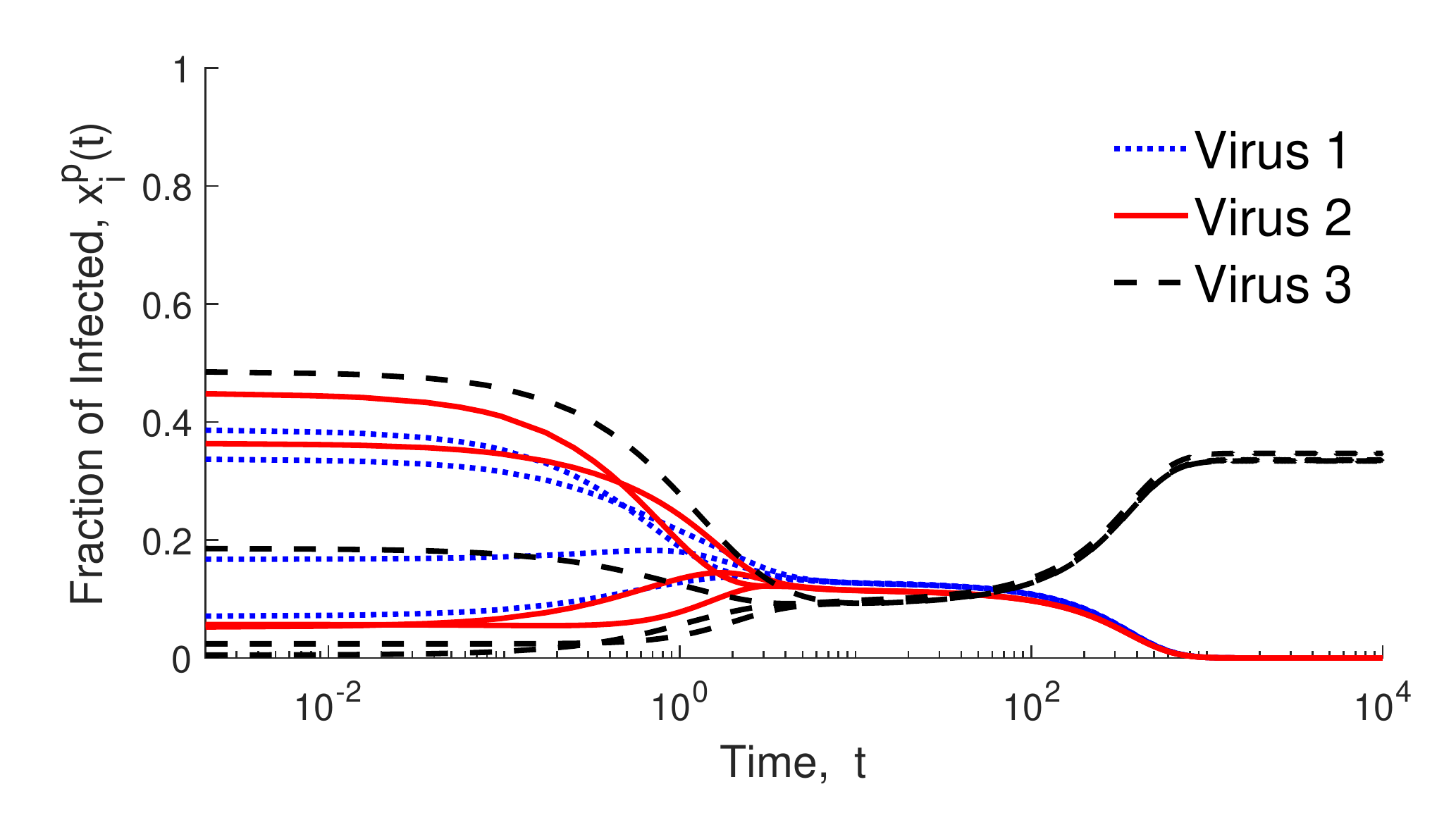}\label{fig:tri_virus_unstable_2line}}
\end{minipage}
\vfill
\begin{minipage}{0.49\linewidth}
\centering\subfloat[Example~3]{\includegraphics[width=\columnwidth]{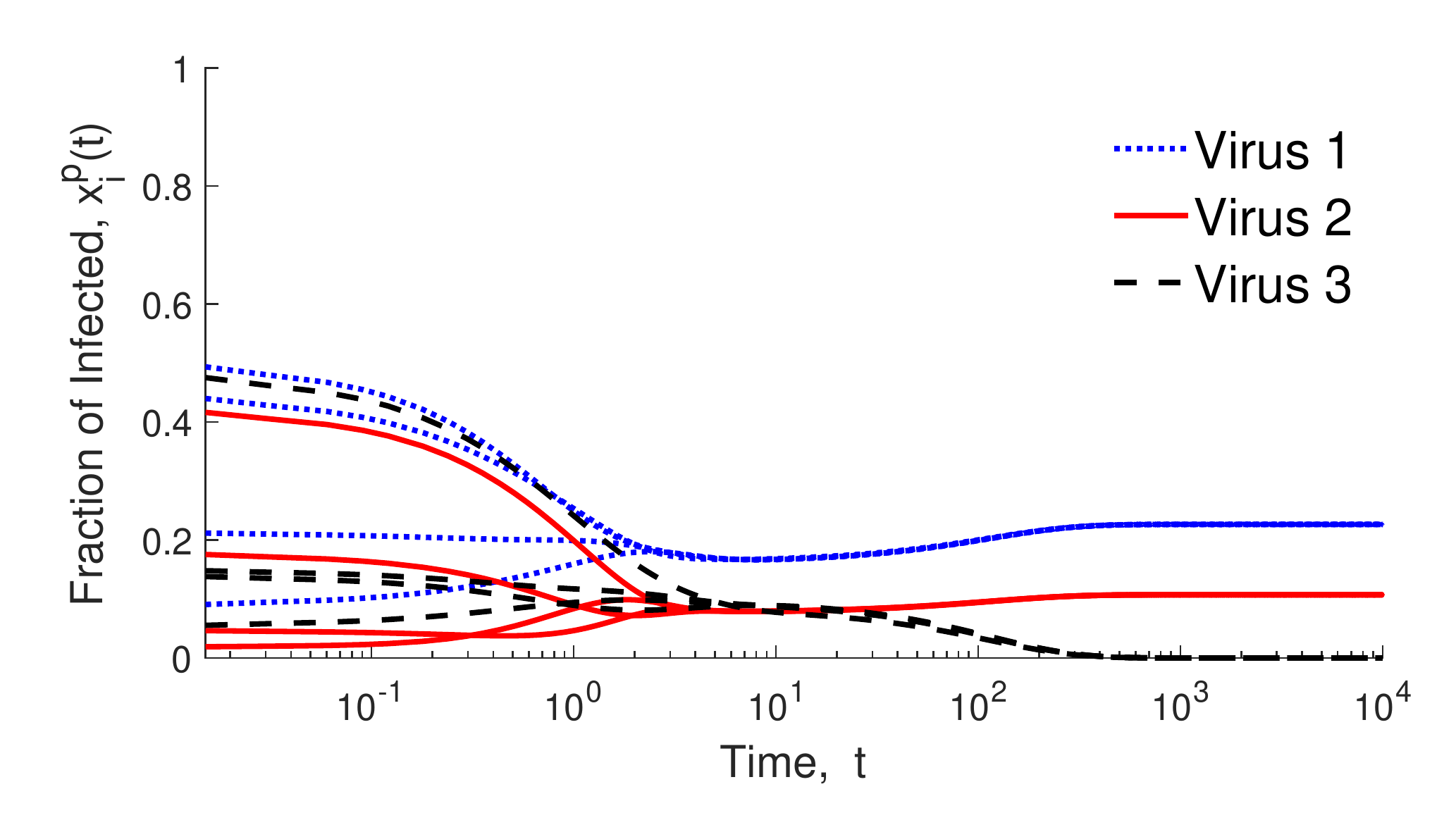}\label{fig:tri_virus_stable_2line}}
\end{minipage}
\hfill
\begin{minipage}{0.49\linewidth}
\centering
\subfloat[Example~4, First Initial Condition]{\includegraphics[width=\columnwidth]{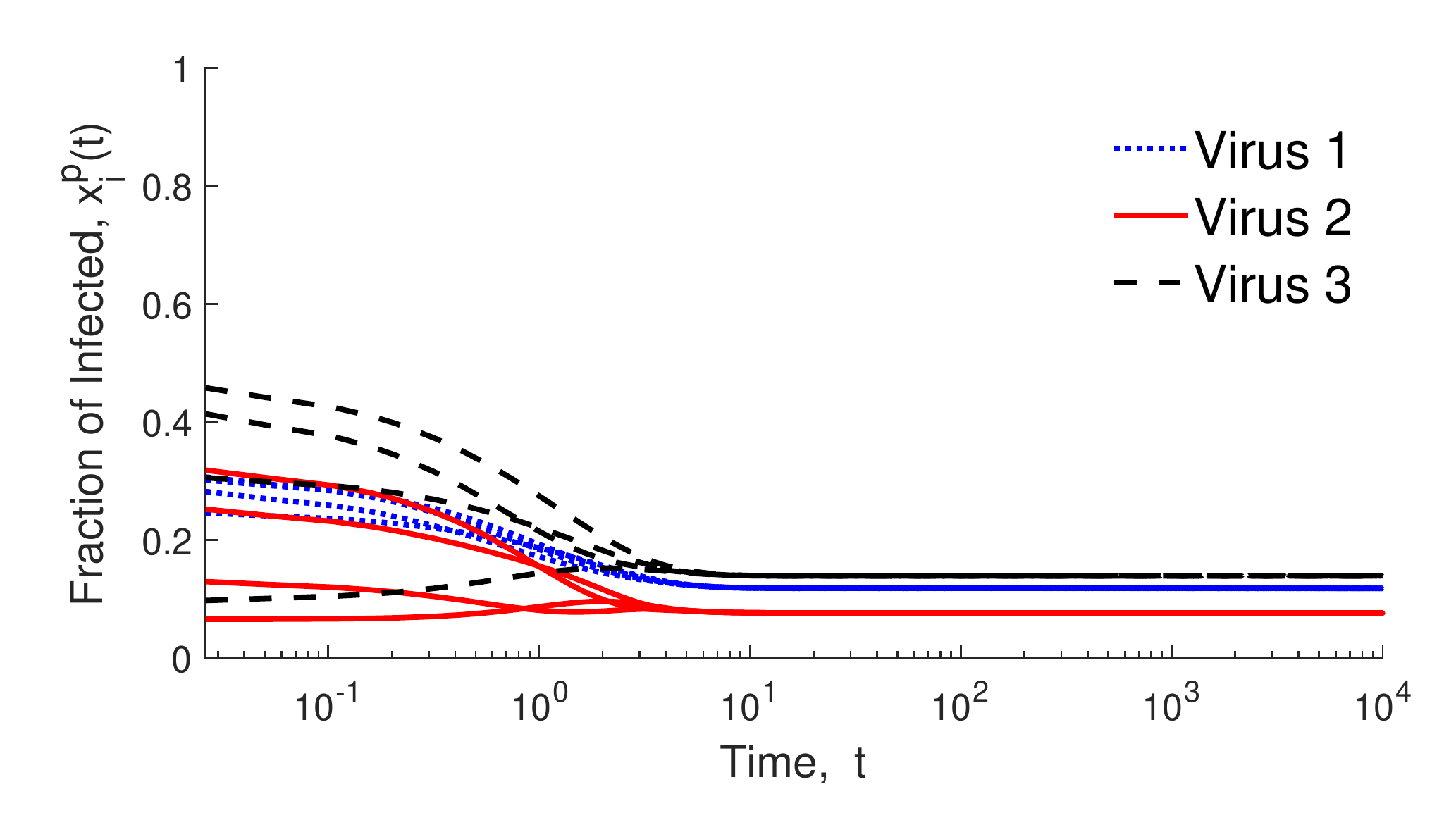}\label{fig:tri_virus_stable_3plane_a}}
\end{minipage}
\vfill
\begin{minipage}{0.49\linewidth}
\centering\subfloat[Example~4, Second Initial Condition]{\includegraphics[width=\columnwidth]{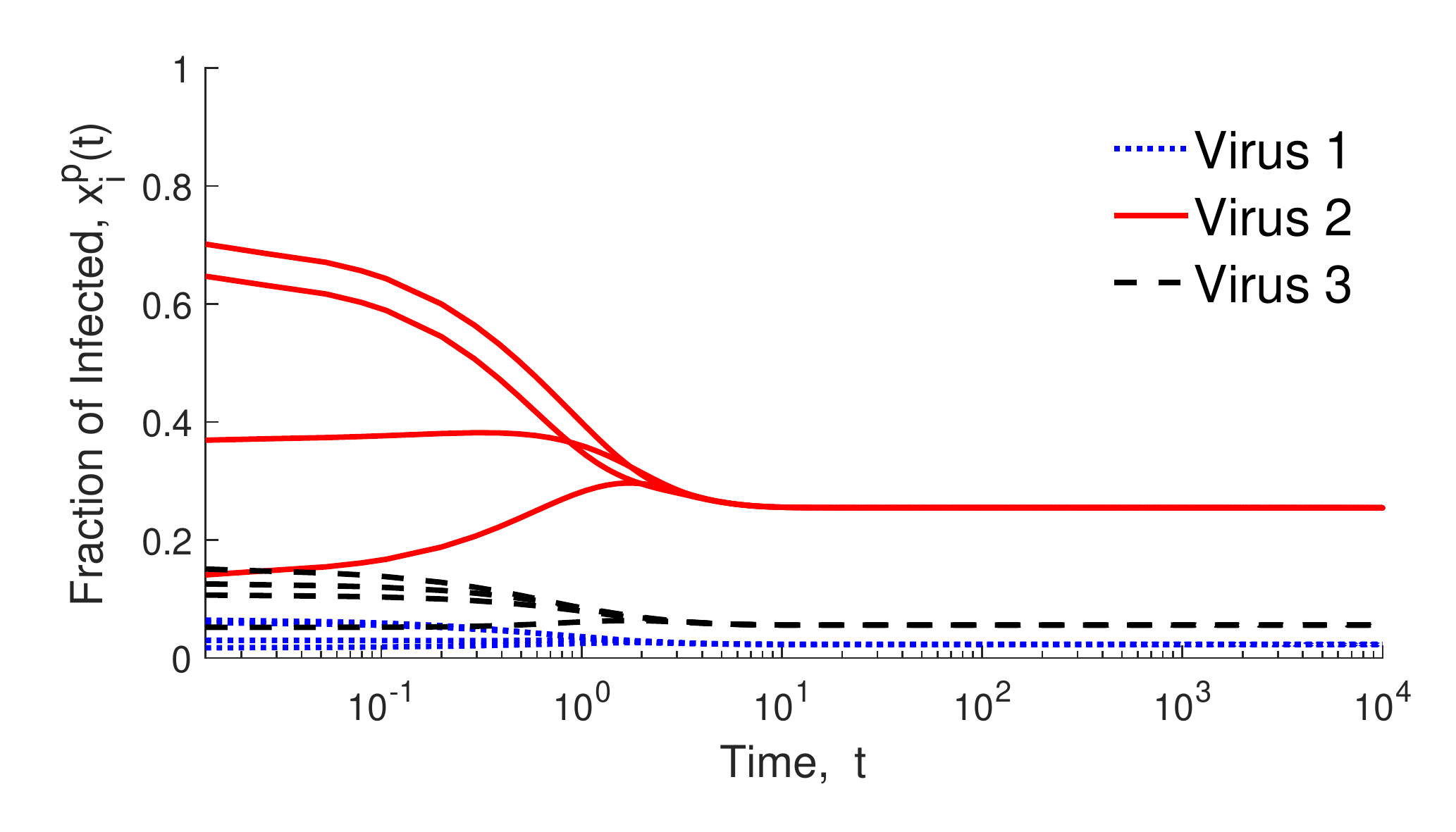}\label{fig:tri_virus_stable_3plane_b}}
\end{minipage}
\hfill
\begin{minipage}{0.49\linewidth}
\centering
\subfloat[Example~5, First Initial Condition]{\includegraphics[width=\columnwidth]{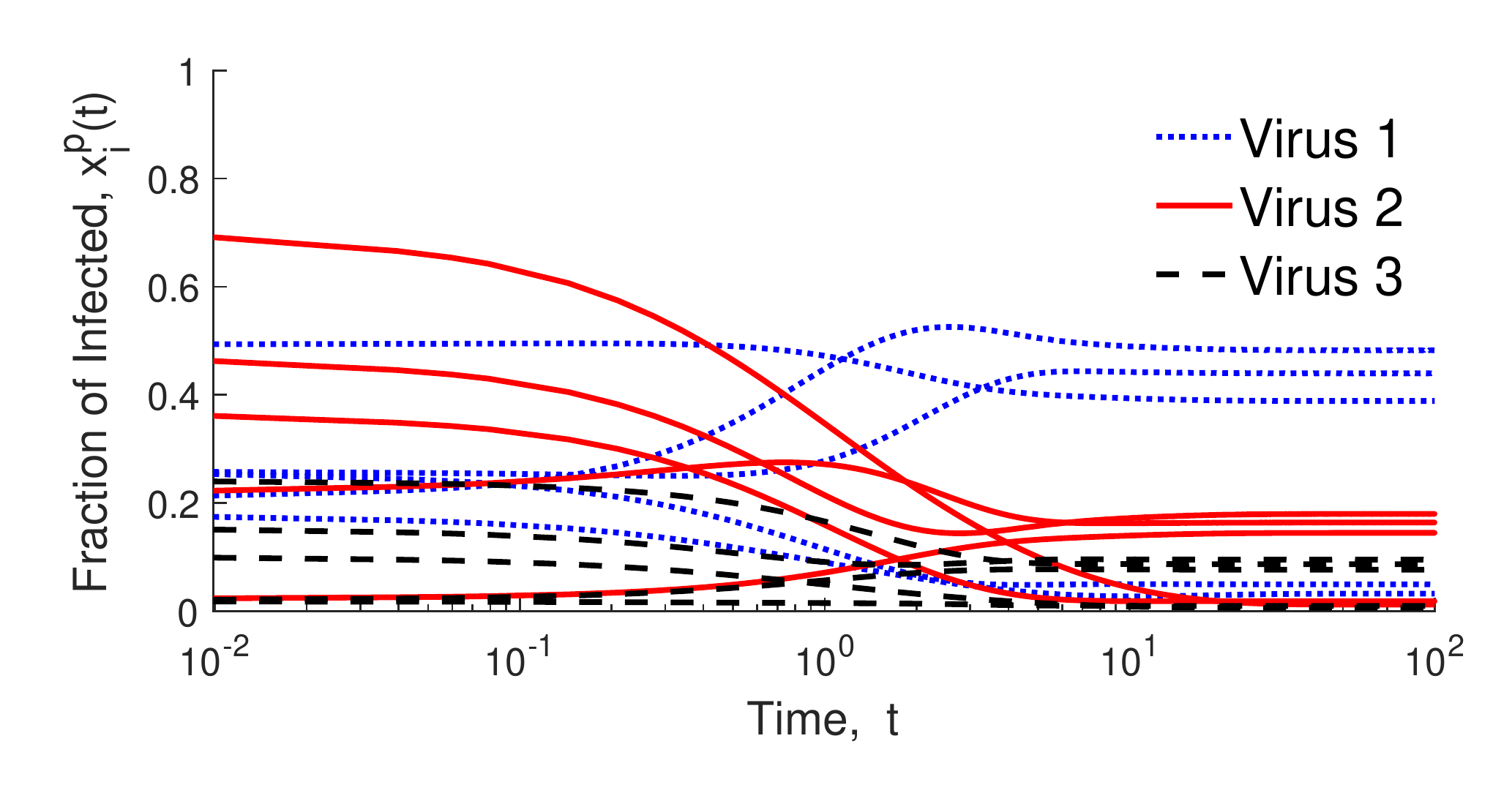}\label{fig:trivirus_stable_copyplane_a}}
\end{minipage}
\vfill
\begin{minipage}{0.49\linewidth}
\centering\subfloat[Example~5, Second Initial Condition]{\includegraphics[width=\columnwidth]{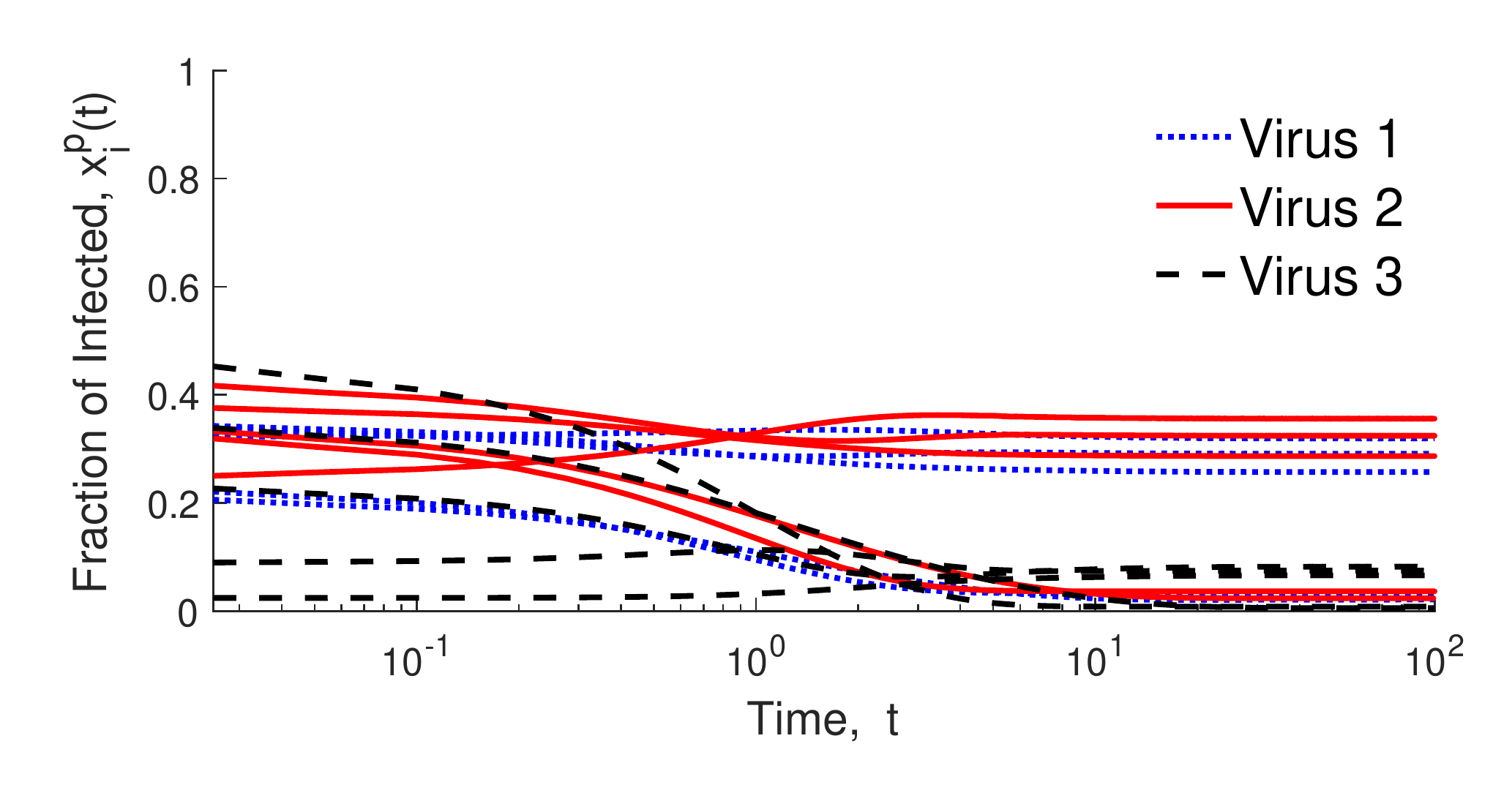}\label{fig:trivirus_stable_copyplane_b}}
\end{minipage}
\hfill
\begin{minipage}{0.49\linewidth}
\centering
\subfloat[Example~6]{\includegraphics[width=\columnwidth]{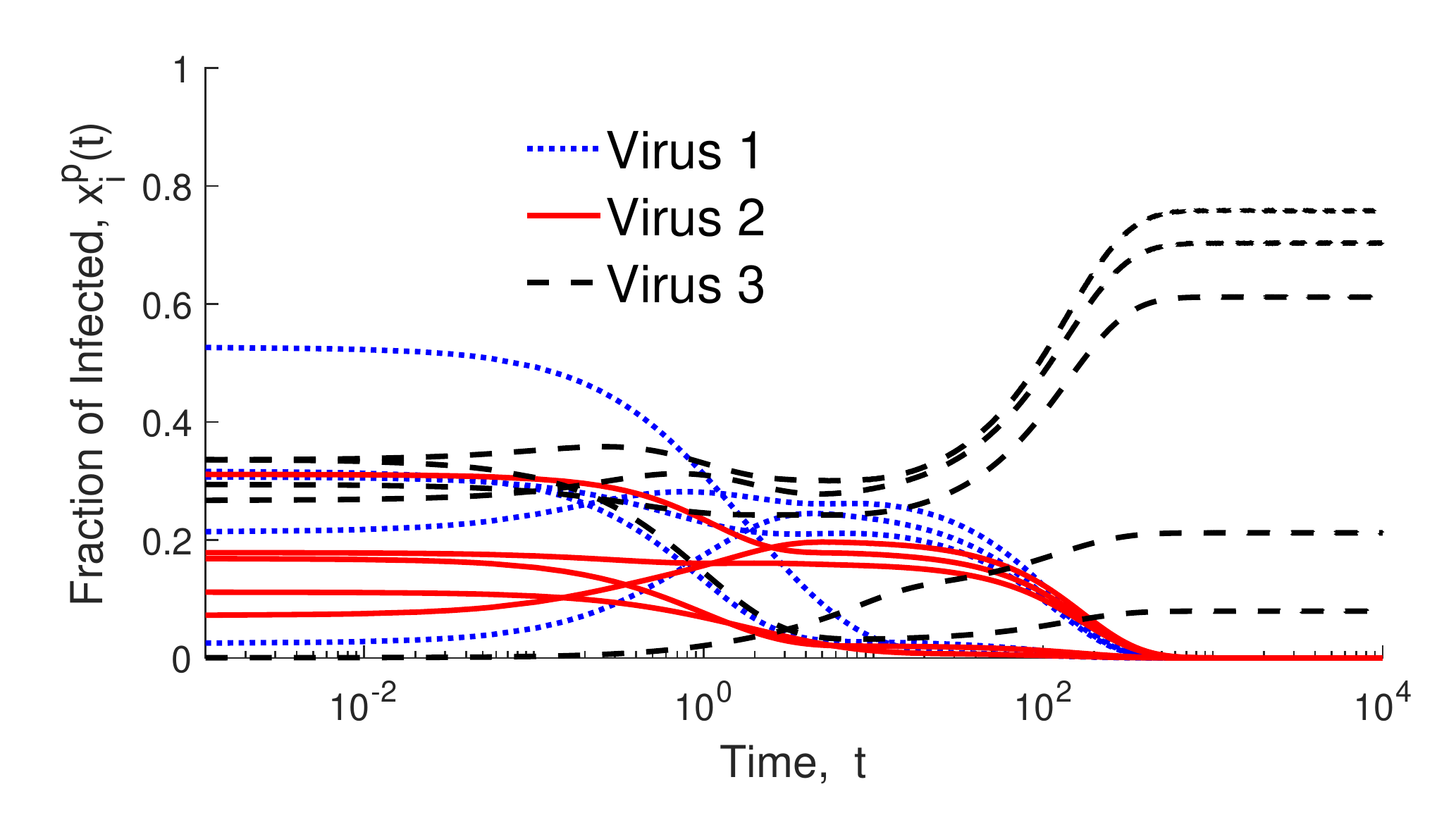}\label{fig:tri_virus_stable_virus3}}
\end{minipage}
\caption{Trajectories of the simulated trivirus system \eqref{eq:full}, for different simulation parameters detailed in Section~\ref{sec:simulations}. }\label{fig:tri_virus_simulations}
\end{figure*}

\begin{figure*}
\begin{minipage}{0.49\linewidth}
\centering\subfloat[Example~7]{\includegraphics[width=\columnwidth]{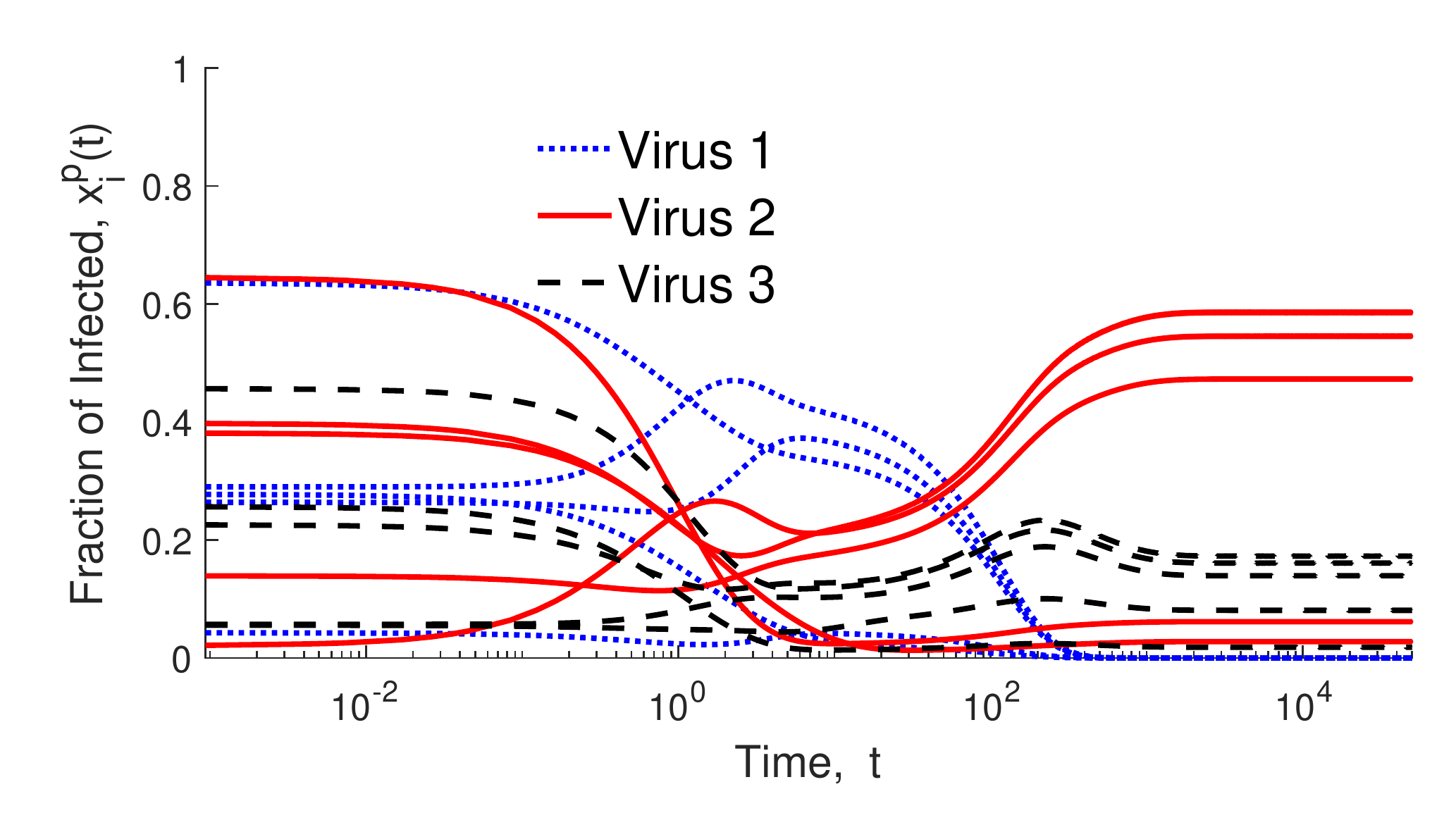}\label{fig:tri_virus_stable_2virus}}
\end{minipage}
\hfill
\begin{minipage}{0.49\linewidth}
\centering
\subfloat[Example~8]{\includegraphics[width=\columnwidth]{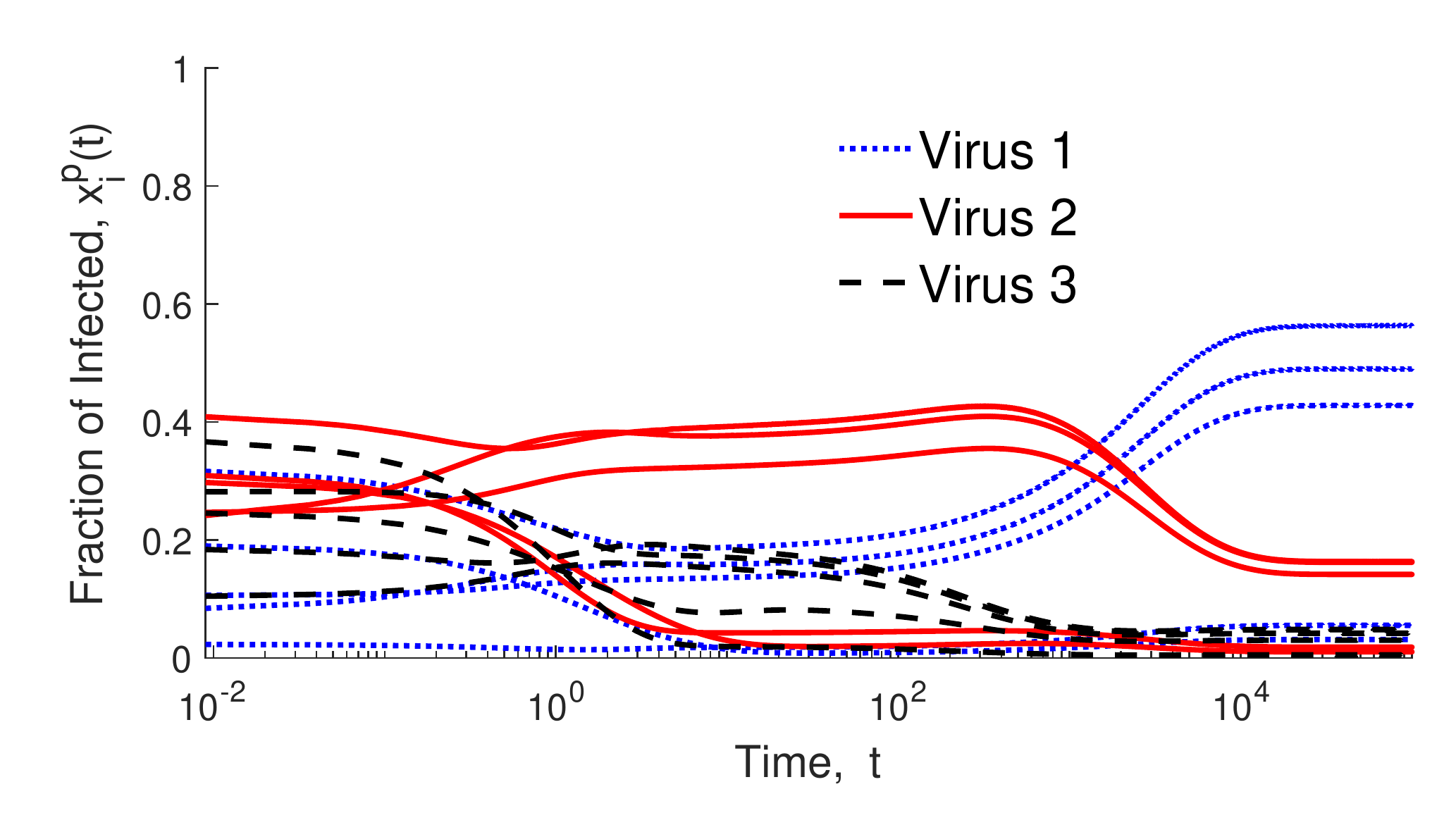}\label{fig:tri_virus_stable_3virus}}
\end{minipage}
\vfill
\begin{minipage}{0.49\linewidth}
\centering\subfloat[Example~9, First Initial Condition]{\includegraphics[width=\columnwidth]{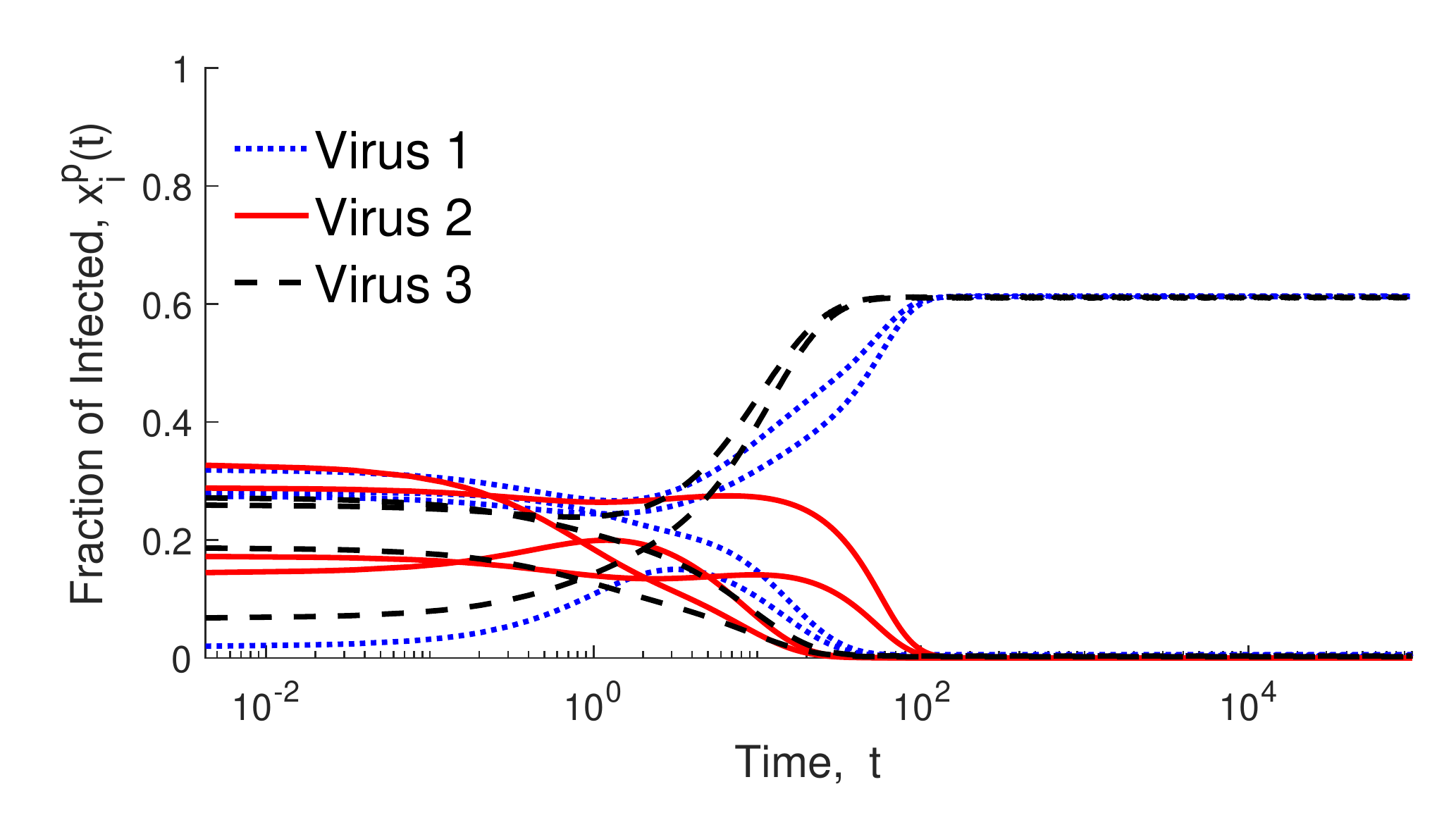}\label{fig:tri_virus_stable_multiequib_virus13}}
\end{minipage}
\hfill
\begin{minipage}{0.49\linewidth}
\centering\subfloat[Example~9, Second Initial Condition]{\includegraphics[width=\columnwidth]{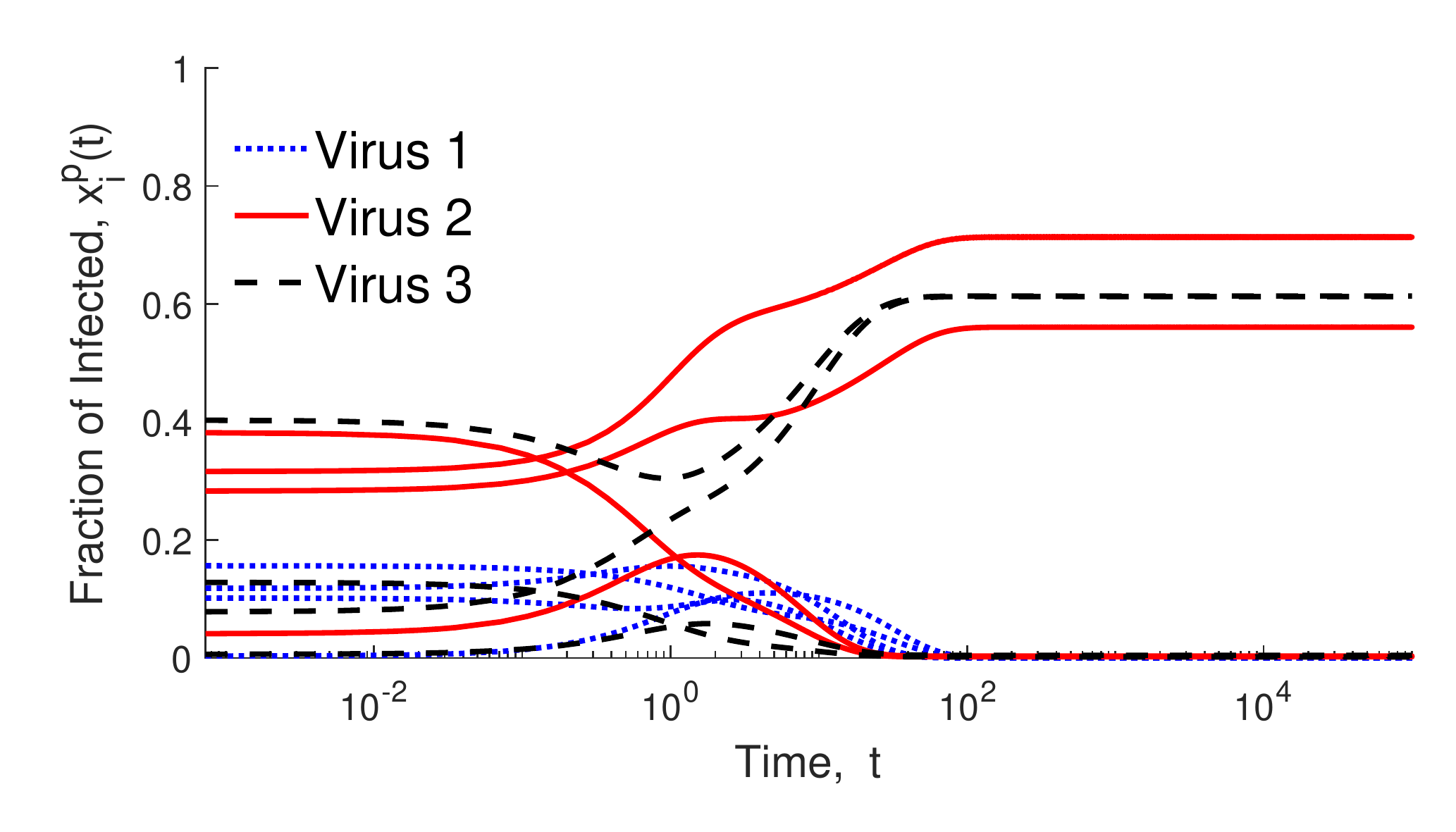}\label{fig:tri_virus_stable_multiequib_virus23}}
\end{minipage}
\caption{Trajectories of the simulated trivirus system \eqref{eq:full}, for different simulation parameters detailed in Section~\ref{sec:simulations}. }\label{fig:tri_virus_simulations_b}
\end{figure*}

\section{Conclusion}\label{sec:conclusions}
This paper analyzed a competitive tri-virus SIS model. Specifically, we established that the tri-virus system is not monotone, unlike the bi-virus system. Subsequently, using the Parametric Transversality Theorem,  we showed that, generically, the tri-virus system has a finite number of equilibria and that the Jacobian matrices associated with each equilibrium are nonsingular. We then identified a sufficient condition for local exponential convergence to a boundary equilibrium. Thereafter, we identified a sufficient condition for the nonexistence of various kinds of coexistence equilibria, followed by the provision of  a sufficient condition for the existence of a 3-coexistence equilibrium (resp. 2-coexistence equilibrium). We identified  
i) a setting that permits the existence and local exponential attractivity of a line of coexistence equilibria; and ii) two settings that permit the existence of, and, in one case, global convergence to, a plane of coexistence equilibria. All of the aforementioned settings have measure zero.
   
\par There are several directions of future research that could be pursued. Recalling the discussion in Section~\ref{sec:tri:virus:loss:of:generality},  a straightforward problem is establishing whether or not the stability properties of a trivirus system defined by $\{D^i,B^i, i=1,2,3\}$ are the same as those defined by $\{I,(D^i)^{-1}B^i, i=1,2,3\}$.
Second, given that the tri-virus system is not monotone, a comprehensive view of the model that, in particular, focuses on its limiting behavior would be extremely relevant since, to provide the said view, novel tools may be required. On a related note, investigating scenarios that admit the existence (resp. exclude the possibility of existence) of limit cycles would provide further insights. 
Third, while the present paper identifies conditions for the existence of a 3-coexistence equilibrium point, no claims are made regarding its uniqueness or its stability (or lack thereof). Hence, one line of investigation could, perhaps by leveraging Morse-Smale inequalities \cite{smale1967differentiable}, focus on providing bounds on the number of 3-coexistence equilibria that a tri-virus system is permitted. Fourth, assuming that one of the competing viruses is benign as compared to the other two, devising
(feedback) control strategies that can drive the virus dynamics to the boundary equilibrium of the benign virus could provide policymakers with another tool for disease mitigation. 

\bibliographystyle{siamplain}
\bibliography{references}

\section*{Appendix~A: Proof of Proposition~\ref{prop:finite:equilibria}}
The use below of the Parametric Transversality Theorem, the principal tool in the proof, requires the preliminary calculation of certain Jacobians. We will consider the case where the $B^i$ are fixed, to demonstrate that for almost all allowed $D^i$, equilibria are nondegenerate (i.e. the associated Jacobians are nonsingular). Given the bounded nature of the set of interest, viz. $\{\bf{0}\leq x^k \leq \bf{1}\}\cap \{ x^1+x^2+x^3\leq {\bf{1}} \}$  and continuity of the Jacobian, finiteness of the number of zeros follows straightforwardly from the nondegeneracy conclusion. 

 It is standard that there is a finite number of equilibria (viz. 1 or 2)), when two of the $x^i$ are zero, the equilibria being expressible using those of a single virus network. If precisely one of the $x^i$ is zero, the equilibria can be defined using a bivirus system, and again, the finiteness for generic $D^k$ is known from \cite{ye2021convergence}. It, therefore, remains to consider coexistence equilibria with all $x^k$ nonzero separately, and in this case, as shown in 
Lemma~\ref{lem:equi_non-zero_nonone:1}
we can assume $x^k\gg 0, k=1,2,3$.

Without loss of generality, we will assume there is some fixed positive $\bar d<1$ such that all diagonal entries of the $D^i$ lie in $(\bar d,\bar d^{-1})$. Let $\mathcal D$ denote the manifold defined by the set of allowed $D^i$ consistent with the choice of $\bar d$.

Let $\mathcal X$ denote the manifold $\{\bf{0}\ll x^k \ll {\bf{1}}\}\cap \{x^1+x^2+x^3\ll {\bf{1}}\}$. Consider for any $x=[(x^1)^{\top},(x^2)^{\top},(x^3)^{\top}]^{\top}\in\mathcal X$ and $\delta=\mbox{vec}[D^1,D^2,D^3]\in\mathcal D$ the map
\begin{align*}
f_{\delta}:&\mathcal X\times\mathcal D\to\mathcal Y,\\ 
(x,\delta)&\mapsto y=\begin{bmatrix}[-D^1+(I_n-X^1-X^2-X^3)B^1]x^1\\
[-D^2+(I_n-X^1-X^2-X^3)B^2]x^2\\
[-D^3+(I_n-X^1-X^2-X^3)B^3]x^3
\end{bmatrix}
\end{align*}
We now claim:
\begin{lem}
With notation as above, the matrix $\frac{\partial f_{\delta}(x,\delta)}{\partial (x,\delta)}$ has full row rank at any coexistence equilibrium. 
\end{lem}

\textit{Proof:}
We first observe that
\begin{small}
\begin{align}\label{eq:jacxdelta}
    \frac{\partial f_{\delta}(x,\delta)}{\partial (x,\delta)}&=\frac{\partial f_{\delta}(x^1,x^2,x^3,D^1,D^2,D^3)}{\partial (x^1,x^2,x^3,D^1,D^2,D^3)}\\
    &=\begin{bmatrix}\left(\frac{\partial f_{\delta}(x^1,x^2,x^3,D^1,D^2,D^3)}{\partial (x^1,x^2,x^3)}\right)^{\top}\\\nonumber
    \left(\frac{\partial f_{\delta}(x^1,x^2,x^3,D^1,D^2,D^3)}{\partial (D^1,D^2,D^3)}\right)^{\top}\\
    \end{bmatrix}^{\top}
\end{align}
\end{small}
Observe next, using the diagonal nature of the $D^i$, that there holds (with $e_j$ denoting the $j$-th unit vector)
\begin{align}
\frac{\partial f_{\delta}(x^1,x^2,x^3,D^1,D^2, D^3)}{\partial \delta^1_j}&=x^1_je_j\\\nonumber 
\frac{\partial f_{\delta}(x^1,x^2,x^3,D^1,D^2,D^3)}{\partial \delta^2_j}&=x^2_je_{j+n}\\\nonumber
\frac{\partial f_{\delta}(x^1,x^2,x^3,D^1,D^2,D^3)}{\partial \delta^3_j}&=x^3_je_{j+2n}
\end{align}
It follows that 
\begin{align}
     \frac{\partial f_{\delta}(x^1,x^2,x^3, D^1,D^2, D^3)}{\partial \delta}&=\frac{\partial f_{\delta}(x^1,x^2,x^3, D^1,D^2, D^3)}{\partial (D^1,D^2, D^3)}\\\nonumber
     &=\begin{bmatrix}
    X^1&0&0\\0&X^2&0\\0&0&X^3
    \end{bmatrix}
\end{align}

At any equilibrium in $\mathcal X$, this matrix has full row rank, as then does the Jacobian matrix $\frac{\partial f_{\delta}(x,\delta)}{\partial (x,\delta)}$ in \eqref{eq:jacxdelta} of which it is a submatrix. 

Let $\mathcal Z$ be the submanifold of $\mathcal Y$ consisting of the single point ${\bf{0_{3n}}}$. It is, in fact, the image of those points in $\mathcal X \times \mathcal D$ defined by $f_{\delta}(x,\delta)={\bf{0_{3n}}}$. Hence at such a point, $x^k\gg 0$ or $X^k$ is nonsingular. Hence the Jacobian of $f_{\delta}$ with respect to $(x,\delta)$ has full rank at such a point, which means that the map  $f_{\delta}$ is transversal to $\mathcal Z$. Then by the Parametric Transversality Theorem \cite[see p.145]{lee2013introduction}, \cite[see p.68]{guillemin2010differential}, it follows that for almost all choices $\bar\delta$ of $\delta$ , the mapping $f_{\bar\delta}:\mathcal X\to \mathcal Y, f_{\bar \delta}(x)=f_{\delta}(x,\bar \delta)$  will be transversal to $\mathcal Z$, i.e. the Jacobian $\frac{\partial f_{\delta}(x,\bar \delta)}{\partial x}$  will be of full rank $3n$  at the zeros of $f_{\bar\delta}$. Since this matrix is square and of size $3n\times 3n$, this says that  every zero of $f(x,\bar\delta)$ is nondegenerate and as a consequence, isolated. As already noted, since the zeros occur in a bounded set, they are finite in number. The choices of $\delta$ for which they are not finite in number, if indeed such choices exist, define a set of measure zero.

\end{document}

%% file: ex_shared.tex
\usepackage{subfig}
\usepackage{lipsum}
\usepackage{amsfonts}
\usepackage{graphicx}
\usepackage{epstopdf}
\usepackage{algorithmic}
\ifpdf
  \DeclareGraphicsExtensions{.eps,.pdf,.png,.jpg}
\else
  \DeclareGraphicsExtensions{.eps}
\fi

\newsiamremark{remark}{Remark}
\newsiamremark{hypothesis}{Hypothesis}
\crefname{hypothesis}{Hypothesis}{Hypotheses}
\newsiamthm{claim}{Claim}

\headers{Competitive tri-virus model: Endemic Behavior}{S.~Gracy, M.~Ye, BDO~Anderson, and CA~Uribe}

\title{Towards Understanding the Endemic Behavior of a Competitive  Tri-Virus SIS Networked Model\thanks{Submitted to the editors 
on \today. A  preliminary version of the paper has been accepted for publication in the Proceedings of the $2023$ American Control Conference, \cite{gracy2022trivirus}.
\funding{Part of this material is based upon work supported by the National Science Foundation under Grants \#2211815 and No. \#2213568. The work of SG was also funded by a grant from the Rice Academy of Fellows, Rice University, Houston, TX.}}}

\author{Sebin~Gracy\thanks{Department of Electrical and Computer Engineering, Rice University, Houston, TX, USA.  
  (\email{sebin.gracy@rice.edu, cesar.uribe@rice.edu}).} 
\and Mengbin~Ye\thanks{Centre for Optimisation and Decision Science, Curtin University, Australia.
  (\email{mengbin.ye@curtin.edu.au})}. 
\and Brian~D.O.~Anderson \thanks{School of Engineering, Australian National University, Canberra, Australia.
  (\email{brian.anderson@anu.edu.au})}
\and C\'esar A Uribe\footnotemark[2]}

\usepackage{amsopn}
\DeclareMathOperator{\diag}{diag}


%% file: ex_article.bbl
\begin{thebibliography}{10}

\bibitem{anderson2022equilibria}
{\sc B.~Anderson and M.~Ye}, {\em Equilibria analysis of a networked bivirus
  epidemic model using poincar$\backslash$'e--hopf and manifold theory}, arXiv
  preprint arXiv:2210.11044,  (2022).

\bibitem{arcede2020accounting}
{\sc J.~P. Arcede, R.~L. Caga-Anan, C.~Q. Mentuda, and Y.~Mammeri}, {\em
  Accounting for symptomatic and asymptomatic in a {SEIR}-type model of
  {COVID-19}}, Mathematical Modelling of Natural Phenomena, 15 (2020), p.~34.

\bibitem{armington1969theory}
{\sc P.~S. Armington}, {\em A theory of demand for products distinguished by
  place of production (une th{\'e}orie de la demande de produits
  diff{\'e}renci{\'e}s d'apr{\`e}s leur origine)(una teor{\'\i}a de la demanda
  de productos distingui{\'e}ndolos seg{\'u}n el lugar de producci{\'o}n)},
  Staff Papers-International Monetary Fund,  (1969), pp.~159--178.

\bibitem{bradshaw2017troops}
{\sc S.~Bradshaw and P.~Howard}, {\em Troops, trolls and troublemakers: A
  global inventory of organized social media manipulation}, Computational
  Propaganda Research Project,  (2017).

\bibitem{castillo1989epidemiological}
{\sc C.~Castillo-Chavez, H.~W. Hethcote, V.~Andreasen, S.~A. Levin, and W.~M.
  Liu}, {\em Epidemiological models with age structure, proportionate mixing,
  and cross-immunity}, Journal of {M}athematical {B}iology, 27 (1989),
  pp.~233--258.

\bibitem{carlos2}
{\sc C.~Castillo-Chavez, W.~Huang, and J.~Li}, {\em Competitive exclusion and
  coexistence of multiple strains in an {SIS STD} model}, SIAM Journal on
  Applied Mathematics, 59 (1999), pp.~1790--1811.

\bibitem{doshi2022convergence}
{\sc V.~Doshi, S.~Mallick, et~al.}, {\em {Convergence of Bi-Virus Epidemic
  Models With Non-Linear Rates on Networks—A Monotone Dynamical Systems
  Approach}}, IEEE/ACM Transactions on Networking,  (2022).

\bibitem{FallMMNP07}
{\sc A.~Fall, A.~Iggidr, G.~Sallet, and J.~J. Tewa}, {\em Epidemiological
  models and {L}yapunov functions}, Mathematical Modelling of Natural
  Phenomena, 2 (2007), pp.~55--73.

\bibitem{fall2007epidemiological}
{\sc A.~Fall, A.~Iggidr, G.~Sallet, and J.-J. Tewa}, {\em Epidemiological
  models and lyapunov functions}, Mathematical Modelling of Natural Phenomena,
  2 (2007), pp.~62--83.

\bibitem{gracy2022modeling}
{\sc S.~Gracy, P.~E. Par{\'e}, J.~Liu, H.~Sandberg, C.~L. Beck, K.~H.
  Johansson, and T.~Ba{\c{s}}ar}, {\em Modeling and analysis of a coupled {SIS}
  bi-virus model}, Automatica, https://arxiv.org/pdf/2207.11414.pdf,  (2023).
\newblock {N}ote: Under Review, $2^{nd}$ Round.

\bibitem{gracy2020analysis}
{\sc S.~Gracy, P.~E. Par{\'e}, H.~Sandberg, and K.~H. Johansson}, {\em Analysis
  and distributed control of periodic epidemic processes}, IEEE Transactions on
  Control of Network Systems, 8 (2020), pp.~123--134.

\bibitem{gracy2022trivirus}
{\sc S.~Gracy, M.~Ye, B.~D. Anderson, and C.~A. Uribe}, {\em On the endemic
  behavior of a competitive tri-virus {SIS} networked model}, in 2023 American
  Control Conference (ACC), IEEE, 2023.
\newblock {N}ote: Accepted.

\bibitem{guillemin2010differential}
{\sc V.~Guillemin and A.~Pollack}, {\em {Differential Topology}}, vol.~370,
  American Mathematical Soc., 2010.

\bibitem{hethcote2000mathematics}
{\sc H.~W. Hethcote}, {\em The mathematics of infectious diseases}, SIAM
  Review, 42 (2000), pp.~599--653.

\bibitem{hirsch1988stability}
{\sc M.~W. Hirsch}, {\em Stability and convergence in strongly monotone
  dynamical systems}, Journal fur die reine und angewandte Mathmatik,  (1988).

\bibitem{hofbauer1990index}
{\sc J.~Hofbauer}, {\em An index theorem for dissipative semiflows}, The Rocky
  Mountain Journal of Mathematics, 20 (1990), pp.~1017--1031.

\bibitem{axel2020TAC}
{\sc A.~Janson, S.~Gracy, P.~E. Par\'e, H.~Sandberg, and K.~H. Johansson}, {\em
  Networked multi-virus spread with a shared resource: Analysis and mitigation
  strategies}, https://arxiv.org/pdf/2011.07569.pdf,  (2020),
  \url{https://arxiv.org/pdf/2011.07569.pdf}.

\bibitem{johnson2002updating}
{\sc N.~P. Johnson and J.~Mueller}, {\em Updating the accounts: Global
  mortality of the 1918-1920 ``{Spanish}" influenza pandemic}, Bulletin of the
  History of Medicine,  (2002), pp.~105--115.

\bibitem{karrer2011competing}
{\sc B.~Karrer and M.~E. Newman}, {\em Competing epidemics on complex
  networks}, Physical Review E, 84 (2011), p.~036106.

\bibitem{kermack1932contributions}
{\sc W.~O. Kermack and A.~G. McKendrick}, {\em Contributions to the
  mathematical theory of epidemics. ii.—the problem of endemicity},
  Proceedings of the Royal Society of London. Series A, containing papers of a
  mathematical and physical character, 138 (1932), pp.~55--83.

\bibitem{khalil2002nonlinear}
{\sc H.~Khalil}, {\em Nonlinear Systems}, Prentice Hall, 2002.

\bibitem{khanafer2016stability}
{\sc A.~Khanafer, T.~Ba{\c{s}}ar, and B.~Gharesifard}, {\em Stability of
  epidemic models over directed graphs: A positive systems approach},
  Automatica, 74 (2016), pp.~126--134.

\bibitem{lajmanovich1976deterministic}
{\sc A.~Lajmanovich and J.~A. Yorke}, {\em A deterministic model for gonorrhea
  in a nonhomogeneous population}, Mathematical Biosciences, 28 (1976),
  pp.~221--236.

\bibitem{laurie2018evidence}
{\sc K.~L. Laurie, W.~Horman, L.~A. Carolan, K.~F. Chan, D.~Layton, A.~Bean,
  D.~Vijaykrishna, P.~C. Reading, J.~M. McCaw, and I.~G. Barr}, {\em Evidence
  for viral interference and cross-reactive protective immunity between
  influenza b virus lineages}, The Journal of infectious diseases, 217 (2018),
  pp.~548--559.

\bibitem{lee2013introduction}
{\sc J.~M. Lee}, {\em Introduction to smooth manifolds}, Springer, 2013.

\bibitem{li1995global}
{\sc M.~Y. Li and J.~S. Muldowney}, {\em Global stability for the {SEIR} model
  in epidemiology}, Mathematical biosciences, 125 (1995), pp.~155--164.

\bibitem{liu2019analysis}
{\sc J.~Liu, P.~E. Par{\'e}, A.~Nedi{\'c}, C.~Y. Tang, C.~L. Beck, and
  T.~Ba{\c{s}}ar}, {\em Analysis and control of a continuous-time bi-virus
  model}, IEEE Transactions on Automatic Control, 64 (2019), pp.~4891--4906.

\bibitem{macarthur1967limiting}
{\sc R.~MacArthur and R.~Levins}, {\em The limiting similarity, convergence,
  and divergence of coexisting species}, The american naturalist, 101 (1967),
  pp.~377--385.

\bibitem{matouk2020complex}
{\sc A.~Matouk}, {\em Complex dynamics in susceptible-infected models for
  {COVID-19} with multi-drug resistance}, Chaos, Solitons \& Fractals, 140
  (2020), p.~110257.

\bibitem{mei2017dynamics}
{\sc W.~Mei, S.~Mohagheghi, S.~Zampieri, and F.~Bullo}, {\em On the dynamics of
  deterministic epidemic propagation over networks}, Annual Reviews in Control,
  44 (2017), pp.~116--128.

\bibitem{meyer2000matrix}
{\sc C.~Meyer}, {\em Matrix Analysis and Applied Linear Algebra}, Other Titles
  in Applied Mathematics, SIAM, 2000,
  \url{https://books.google.se/books?id=-7JeAwAAQBAJ}.

\bibitem{newman2005threshold}
{\sc M.~E. Newman}, {\em Threshold effects for two pathogens spreading on a
  network}, Physical Review Letters, 95 (2005), p.~108701.

\bibitem{nowak1991evolution}
{\sc M.~Nowak}, {\em The evolution of viruses. competition between horizontal
  and vertical transmission of mobile genes}, Journal of theoretical biology,
  150 (1991), pp.~339--347.

\bibitem{pare2020modeling}
{\sc P.~E. Par{\'e}, C.~L. Beck, and T.~Ba{\c{s}}ar}, {\em Modeling,
  estimation, and analysis of epidemics over networks: An overview}, Annual
  Reviews in Control, 50 (2020), pp.~345--360.

\bibitem{pare2021multi}
{\sc P.~E. Par{\'e}, J.~Liu, C.~L. Beck, A.~Nedi{\'c}, and T.~Ba{\c{s}}ar},
  {\em Multi-competitive viruses over time-varying networks with mutations and
  human awareness}, Automatica, 123 (2021), p.~109330.

\bibitem{pare2020analysis}
{\sc P.~E. Par{\'e}, D.~Vrabac, H.~Sandberg, and K.~H. Johansson}, {\em
  Analysis, online estimation, and validation of a competing virus model}, in
  2020 American Control Conference (ACC), IEEE, 2020, pp.~2556--2561.

\bibitem{peng2010epidemic}
{\sc C.~Peng, X.~Jin, and M.~Shi}, {\em Epidemic threshold and immunization on
  generalized networks}, Physica A: Statistical Mechanics and Its Applications,
  389 (2010), pp.~549--560.

\bibitem{prakash2012winner}
{\sc B.~A. Prakash, A.~Beutel, R.~Rosenfeld, and C.~Faloutsos}, {\em Winner
  takes all: competing viruses or ideas on fair-play networks}, in Proceedings
  of the 21st international conference on World Wide Web, 2012, pp.~1037--1046.

\bibitem{prakash2010virus}
{\sc B.~A. Prakash, H.~Tong, N.~Valler, M.~Faloutsos, and C.~Faloutsos}, {\em
  Virus propagation on time-varying networks: Theory and immunization
  algorithms}, Joint European Conference on Machine Learning and Knowledge
  Discovery in Databases,  (2010), pp.~99--114.

\bibitem{qu2009cooperative}
{\sc Z.~Qu}, {\em Cooperative control of dynamical systems: applications to
  autonomous vehicles}, Springer Science \& Business Media, 2009.

\bibitem{rami2014switch}
{\sc M.~A. Rami, V.~S. Bokharaie, O.~Mason, and F.~R. Wirth}, {\em Stability
  criteria for {SIS} epidemiological models under switching policies}, Discrete
  and Continuous Dynamical Systems Series, 19 (2014), pp.~2865--2887.

\bibitem{rothe2020transmission}
{\sc C.~Rothe, M.~Schunk, P.~Sothmann, G.~Bretzel, G.~Froeschl, C.~Wallrauch,
  T.~Zimmer, V.~Thiel, C.~Janke, W.~Guggemos, et~al.}, {\em Transmission of
  2019-{nCoV} infection from an asymptomatic contact in {Germany}}, New England
  Journal of Medicine, 382 (2020), pp.~970--971.

\bibitem{sahneh2014competitive}
{\sc F.~D. Sahneh and C.~Scoglio}, {\em Competitive epidemic spreading over
  arbitrary multilayer networks}, Physical Review E, 89 (2014), p.~062817.

\bibitem{santos2015bi}
{\sc A.~Santos, J.~M. Moura, and J.~M. Xavier}, {\em Bi-virus {SIS} epidemics
  over networks: Qualitative analysis}, IEEE Transactions on Network Science
  and Engineering, 2 (2015), pp.~17--29.

\bibitem{smale1967differentiable}
{\sc S.~Smale}, {\em Differentiable dynamical systems}, Bulletin of the
  American mathematical Society, 73 (1967), pp.~747--817.

\bibitem{smith1988systems}
{\sc H.~L. Smith}, {\em Systems of ordinary differential equations which
  generate an order preserving flow. a survey of results}, SIAM review, 30
  (1988), pp.~87--113.

\bibitem{smith2008monotone}
{\sc H.~L. Smith}, {\em Monotone dynamical systems: an introduction to the
  theory of competitive and cooperative systems: an introduction to the theory
  of competitive and cooperative systems}, no.~41, American Mathematical Soc.,
  2008.

\bibitem{sontag2007monotone}
{\sc E.~D. Sontag}, {\em Monotone and near-monotone biochemical networks},
  Systems and synthetic biology, 1 (2007), pp.~59--87.

\bibitem{trpevski2010model}
{\sc D.~Trpevski, W.~K. Tang, and L.~Kocarev}, {\em Model for rumor spreading
  over networks}, Physical Review E, 81 (2010), p.~056102.

\bibitem{van2014exact}
{\sc P.~Van~Mieghem}, {\em Exact {M}arkovian {SIR} and {SIS} epidemics on
  networks and an upper bound for the epidemic threshold}, arXiv preprint
  arXiv:1402.1731,  (2014).

\bibitem{van2008virus}
{\sc P.~Van~Mieghem, J.~Omic, and R.~Kooij}, {\em Virus spread in networks},
  IEEE/ACM Transactions On Networking, 17 (2008), pp.~1--14.

\bibitem{van2009virus}
{\sc P.~Van~Mieghem, J.~Omic, and R.~Kooij}, {\em Virus spread in networks},
  IEEE/ACM Transactions on Networking, 17 (2009), pp.~1--14.

\bibitem{varga1999matrix}
{\sc R.~Varga}, {\em Matrix Iterative Analysis}, Springer Series in
  Computational Mathematics, Springer Berlin Heidelberg, 1999,
  \url{https://books.google.se/books?id=U2XYs1DyKiYC}.

\bibitem{wang2020clinical}
{\sc M.~Wang, Q.~Wu, W.~Xu, B.~Qiao, J.~Wang, H.~Zheng, S.~Jiang, J.~Mei,
  Z.~Wu, Y.~Deng, et~al.}, {\em Clinical diagnosis of 8274 samples with
  2019-novel coronavirus in {Wuhan}}, MedRxiv,  (2020).

\bibitem{wang2003epidemic}
{\sc Y.~Wang, D.~Chakrabarti, C.~Wang, and C.~Faloutsos}, {\em Epidemic
  spreading in real networks: An eigenvalue viewpoint}, in Proceedings of the
  22nd International Symposium on Reliable Distributed Systems, 2003,
  pp.~25--34.

\bibitem{wang2019systematic}
{\sc Y.~Wang, M.~McKee, A.~Torbica, and D.~Stuckler}, {\em Systematic
  literature review on the spread of health-related misinformation on social
  media}, Social science \& medicine, 240 (2019), p.~112552.

\bibitem{wang2012dynamics}
{\sc Y.~Wang, G.~Xiao, and J.~Liu}, {\em Dynamics of competing ideas in complex
  social systems}, New Journal of Physics, 14 (2012), p.~013015.

\bibitem{wei2013competing}
{\sc X.~Wei, N.~C. Valler, B.~A. Prakash, I.~Neamtiu, M.~Faloutsos, and
  C.~Faloutsos}, {\em Competing memes propagation on networks: A network
  science perspective}, IEEE Journal on Selected Areas in Communications, 31
  (2013), pp.~1049--1060.

\bibitem{wu2020interference}
{\sc A.~Wu, V.~T. Mihaylova, M.~L. Landry, and E.~F. Foxman}, {\em Interference
  between rhinovirus and influenza a virus: a clinical data analysis and
  experimental infection study}, The Lancet Microbe, 1 (2020), pp.~e254--e262.

\bibitem{yang2017bi}
{\sc L.-X. Yang, X.~Yang, and Y.~Y. Tang}, {\em A bi-virus competing spreading
  model with generic infection rates}, IEEE Transactions on Network Science and
  Engineering, 5 (2017), pp.~2--13.

\bibitem{ben:brian:opinion:lcss}
{\sc M.~Ye and B.~D. Anderson}, {\em Competitive epidemic spreading over
  networks}, IEEE Control Systems Letters,  (2022), pp.~1--1,
  \url{https://doi.org/10.1109/LCSYS.2022.3199165}.

\bibitem{ye2021convergence}
{\sc M.~Ye, B.~D. Anderson, and J.~Liu}, {\em Convergence and equilibria
  analysis of a networked bivirus epidemic model}, SIAM Journal on Control and
  Optimization, 60 (2022), pp.~S323--S346.

\bibitem{ye2021_PH_TAC}
{\sc M.~Ye, J.~Liu, B.~D.~O. Anderson, and M.~Cao}, {\em {Applications of the
  Poincar\'e--Hopf Theorem: Epidemic Models and Lotka--Volterra Systems}}, IEEE
  Transactions on Automatic Control, 67 (2022), pp.~1609--1624,
  \url{https://doi.org/10.1109/TAC.2021.3064519}.

\bibitem{zhang2022networked}
{\sc C.~Zhang, S.~Gracy, T.~Ba{\c{s}}ar, and P.~E. Par{\'e}}, {\em A networked
  competitive multi-virus {SIR} model: Analysis and observability},
  IFAC-PapersOnLine, 55 (2022), pp.~13--18.

\bibitem{zhang2011matrix}
{\sc F.~Zhang}, {\em {Matrix Theory: Basic Results and Techniques}}, Springer,
  second~ed., 2011.

\end{thebibliography}
